\documentclass[10pt,letter]{article}

\pdfoutput=1

\usepackage{amsmath,amssymb,amsfonts,amscd,mathrsfs}
\usepackage{xcolor}
\definecolor{darkblue}{rgb}{0.1,0.1,.7}
\usepackage[colorlinks, linkcolor=darkblue, citecolor=darkblue, urlcolor=darkblue, linktocpage]{hyperref} 
\usepackage[square, comma, compress,numbers]{natbib}
\usepackage[]{graphicx}
\usepackage{geometry}
\usepackage[margin=10pt,font=small,labelfont=bf]{caption}
\usepackage{ifthen}
\usepackage{tikz}
\usepackage{subcaption}
\usepackage{booktabs,multirow}
\usepackage{hhline}
\usepackage{tablefootnote}

\usepackage{tikz-cd}

\usepackage{dsfont} 

\usepackage{titlesec}
\titleformat*{\section}{\large\bfseries}
\titleformat*{\subsection}{\normalsize\bfseries}
\titleformat*{\subsubsection}{\normalsize\it}
\titleformat*{\paragraph}{\normalsize\bfseries}
\titleformat*{\subparagraph}{\normalsize\bfseries}

\newcommand{\reef}[1]{(\ref{#1})}

\def\eps{\epsilon}

\newcommand{\beq}{\begin{equation}} 
\newcommand{\eeq}{\end{equation}}
 
\def\nn{\nonumber} 
\def\bR {\mathbb{R}} 
 
   \def\calC {{\cal C}} \def\calD {{\cal D}} 
    
 \def\calJ {{\cal J}}   
 \def\calN {{\cal N}}  \def\calO {{\cal O}}

\def\bZ {\mathbb{Z}}

\def\ge{\geqslant}
\def\le{\leqslant}
\def\geq{\geqslant}
\def\leq{\leqslant}
\def\<{\langle}
\def\>{\rangle}
\newcommand{\diffop}[2]{\ifthenelse{\equal{#2}{1}}{\frac{\mrm{d}}{\mrm{d} #1}}{\frac{\mrm{d}^#2}{\mrm{d} #1^#2}}}

\newcommand{\be}{\begin{equation}}
\newcommand{\ee}{\end{equation}}
\newcommand{\bea}{\begin{eqnarray}}
\newcommand{\eea}{\end{eqnarray}}

\newcommand{\ket}[1]{|#1\rangle}
\newcommand{\bra}[1]{\langle #1|}

\newcommand{\mrm}[1]{{\mathrm #1}}

\def\kappa{\varkappa} 
\def\tr{\mathrm{tr}}

\newcommand{\Rep}{{\sf Rep}}
\newcommand{\Rephat}{\widehat{\Rep}}
\newcommand{\Reptilde}{\widetilde{\Rep}}
\newcommand{\Kar}{{\sf Kar}}
\newcommand{\Hom}{{\rm Hom}}
\newcommand{\ba}{{\bf a}}
	
\newcommand{\ra}{\rightarrow}
\newcommand{\fcy}[1]{\mathcal{#1}}

\usepackage{amsthm}
\theoremstyle{plain}
\newtheorem{prop}{Proposition}[section]
\newtheorem{thm}[prop]{Theorem}

\usepackage{accents}
\newlength{\dhatheight}

\numberwithin{equation}{section}
\usepackage{tocloft}

\begin{document}

\vspace*{-.6in} \thispagestyle{empty}
\begin{flushright}
\verb|PUPT-2601|
\end{flushright}
\vspace{1cm} {\large
\begin{center}
{\bf Deligne Categories in Lattice Models and Quantum Field Theory,\\
	\emph{or} Making Sense of $O(N)$ Symmetry with Non-integer $N$}
\end{center}}
\vspace{1cm}
\begin{center}
{\bf Damon J. Binder$^a$, Slava Rychkov$^{b,c}$}\\[2cm] 
{
$^{a}$ Joseph Henry Laboratories, Princeton University, Princeton, NJ 08544, USA\\
$^b$  Institut des Hautes \'Etudes Scientifiques, Bures-sur-Yvette, France\\
$^c$  
Laboratoire de Physique de l’Ecole normale sup\'erieure, ENS,\\ 
{\small Universit\'e PSL, CNRS, Sorbonne Universit\'e, Universit\'e de Paris,} F-75005 Paris, France
}
\vspace{1cm}
\end{center}

\vspace{4mm}

\begin{abstract}
When studying quantum field theories and lattice models, it is often useful to analytically continue the number of field or spin components from an integer to a real number. In spite of this, the precise meaning of such analytic continuations has never been fully clarified, and in particular the symmetry of these theories is obscure. We clarify these issues using Deligne categories and their associated Brauer algebras, and show that these provide logically satisfactory answers to these questions. Simple objects of the Deligne category generalize the notion of an irreducible representations, avoiding the need for such mathematically nonsensical notions as vector spaces of non-integer dimension. We develop a systematic theory of categorical symmetries, applying it in both perturbative and non-perturbative contexts. 
A partial list of our results is: categorical symmetries are preserved under RG flows; continuous categorical symmetries come equipped with conserved currents; CFTs with categorical symmetries are necessarily non-unitary.
 \end{abstract}
\vspace{.2in}
\vspace{.3in}
\hspace{0.7cm} November 2019

\newpage

\setcounter{tocdepth}{3}

{
\tableofcontents
}

\section{Introduction}

In quantum field theory, the number of field components $N$ should naively be an integer. When performing calculations however, it is often fruitful to analytically continue results to non-integer values.\footnote{The spacetime dimension $d$ is also often analytically continued. In this paper we focus on the internal symmetry, but our considerations are also relevant to spacetime symmetry; see also the discussion.} These constructions are conceptually puzzling; in particular, what remains of the symmetries of a model when we analytically continue $N$? The textbook definition of a symmetry requires a symmetry \emph{group}, and groups like $O(N)$ do not make sense as mathematical objects when we go to non-integer $N$. Our purpose here will be to give a definition which applies when the textbook definition fails. Here are some of our punchlines: 

\begin{itemize}
	\item 
	Some families of `symmetries' allow meaningful analytic continuation in $N$. These include continuous groups, such as $O(N)$, $Sp(N)$, and $U(N)$, and also discrete groups, such as $S_N$. Others families, like $SO(N)$ or $\bZ_N$, do not.
	
	\item We do not analytically continue the group, nor any specific representation of a group, but rather the whole `representation theory'. The algebraic structure underlying this analytically continued `representation theory' is known in mathematics as `Deligne categories' \cite{Deligne}. 
	
	\item A certain algebra of string diagrams underlies the Deligne categories (such as the Brauer algebra for the $O(N)$ case). It is used for practical computations, and `explains' the meaning of $\delta_{ab}$ tensors with non-integer $N$.
\end{itemize}

To set the stage, here are a few examples why one may care about non-integer $N$ {(other physical applications will be given in section \ref{sec:pasha}, as well as sprinkled throughout)}:

\begin{itemize}
	\item 
	Even if interested one is mostly in integer $N$, one may wish to understand how physical observables behave as a function of $N$, which will then necessarily involve intermediate non-integer values. For instance in the theory of critical phenomena, one might study critical exponents of $O(N)$ models as a function of $N$. This is mostly easily done in perturbation theory, where at each order $N$ enters polynomially via contractions of invariant tensors. However, non-perturbative analytic continuations also exist, as the next example shows.  
	
	\item
	
	Certain non-integer $N$ models are interesting in their own right. One famous example are loop models in statistical physics, which are probabilistic ensembles of loops living on a lattice. Every loop contributes to the probability weight a factor $N K^{\text{loop length}}$, where $K$ is a coupling, see Fig.~\ref{fig-loop1}.
	\begin{figure}[h!]
		\centering
		\includegraphics[scale=0.6]{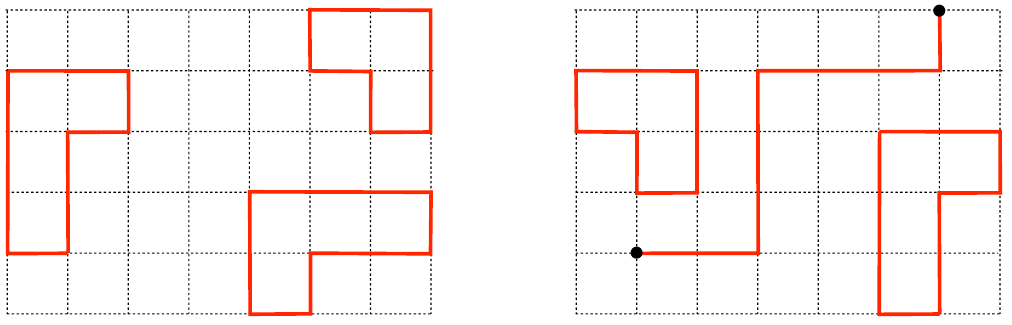}
		\caption{\label{fig-loop1} \emph{Left:} Partition function of an $O(N)$ loop model sums up weights of loop configurations living on a lattice. E.g.~the shown configuration of three loops has weight: $N^3 K^{8+10+10}$, where 8, 10, 10 are the loop lengths. \emph{Right:} Typical contribution to the defect-defect correlator in the loop model, which corresponds to the spin-spin correlator of the spin model.}
	\end{figure}
	This is referred to as $O(N)$ loop model, because it can be obtained by analytic continuing a spin model in the $O(N)$ universality class (this will be reviewed below). 
	This correspondence extends to correlation functions: in the loop formulation we introduce point defects where lines are forced to end. There is a vast literature on the $O(N)$ loop models for continuous values of $N$ and the CFTs describing their phase transitions \cite{Nienhuis2008,Jacobsen2012}. While many quantities of physical interest have been computed, to our knowledge the question of symmetry has never been properly explained.
	
	\item Extensions to nearby non-integer $N$ help to learn about aspects of models with integer $N$. For instance, the fractal dimension of same-sign clusters of spins at the phase transition of the Ising model is a non-local observable, and is difficult to compute within the Ising model itself. It becomes more easily accessible when the model is extended to $N=1+\eps$.\footnote{A great reference about such limits, which often lead to logarithmic conformal field theories, is \cite{Cardy:2013rqg}.}
	
\end{itemize}

The rest of the paper is structured as follows. After a brief reminder of the properties of symmetries (section \ref{sec:need}), we introduce the Brauer algebra in section \ref{sec:brouer} and with its help give an operational definition of non-integer $O(n)$ symmetry in section \ref{sec:oper}. Section \ref{sec:transcat} explains the meaning of this symmetry in the language of Deligne categories. Armed with category theory, in section \ref{sec:lattice} we discuss how to construct the most general $O(n)$-symmetric lattice model. We also discuss Wilsonian renormalization in this context and establish that the categorical symmetry parameter $n$ is not renormalized, just like for group symmetries. In section \ref{sec:QFT} we consider various aspects of QFTs with categorical symmetries: free theories, path integrals, perturbation theory, global symmetry currents, explicit and spontaneous symmetry breaking, conformal field theories and unitarity. Highlights here include the categorical Noether's theorem (section \ref{sec:current}), the categorical Goldstone theorem (conjectured in section \ref{sec:spont}), and two theorems about CFTs with categorical symmetries: Theorem \ref{thm:AllObs} about completeness of the global symmetry spectrum (new even for group symmetries), and Theorem \ref{thm:nonUnitary} about the lack of unitarity. In section \ref{sec:pasha} we emphasize how Deligne categories lead to a non-trivial interplay of algebra and analysis in situations of physical interest. In section \ref{sec:other} we define Deligne categories which interpolate $U(N)$, $Sp(N)$, and $S_N$. We also report what is known about the other families of categories, and about the so far elusive possibility of interpolating the exceptional Lie groups. Then we conclude. Appendix \ref{sec: category} is a self-contained review of tensor categories, and Appendix \ref{sec:CONTCATS} outlines the theory of continuous categories.


\emph{Reader's guide:} This article is rather long, and depending on your interest you may take different routes. The first five sections should be read by everyone. Section \ref{sec:lattice} is mostly for statistical physicists interested in lattice models, while section \ref{sec:QFT} is written with a high energy physics audience in mind, and section \ref{sec:pasha} is for mathematicians interested in how Deligne categories give rise to interesting behaviour in physical systems.

\section{What do we need from a symmetry?}
\label{sec:need}

It is hardly necessary to explain the central role of symmetries\footnote{In this paper we will use the word symmetry to refer to global (i.e.~internal) symmetries.} as an organizing principle in physics. Here is a partial list of their uses, relevant for quantum field theory and statistical physics:
\begin{enumerate}
	
	\item
	States and local operators are classified by representations of the symmetry group. 
	
	\item
	Symmetries restricts the form of measured quantities, such as correlation functions, scattering matrices, and transfer matrices. These must all be invariant tensors of the symmetry group.
	
	\item
	Symmetry is preserved along the renormalization group (RG) flow. If a microscopic theory has a certain symmetry, a CFT describing its long-distance limit should have the same symmetry.\footnote{As is well known, there are two types of exceptional situations where this statement requires qualifications. First, symmetries may break spontaneously. Second, additional symmetries may emerge at long distances because operators which break them are irrelevant in the renormalization group sense.} 
	In other words, universality classes of phase transitions can be classified by their symmetries.
	
	\item Continuous symmetries in systems with local interactions lead to local conserved current operators.
	
	\item
	When a continuous group symmetry is spontaneously broken at long distances, this leads to Goldstone bosons and constraints on their interactions.
	
\end{enumerate}

A `non-integer $N$ symmetry' should be a conceptual framework with similar consequences:

\begin{enumerate}
	
	\item[$1'$.] There has to be a notion which replaces that of irreducible representation and which is used to classify states and local operators. 

\item[$2'$.] We would like to know what the algebraic objects are to which correlation functions and transfer matrices belong, and which replace invariant tensors.

\item[$3'$.] We would like to know that symmetry is preserved along the RG flow. For example, parameter $N$ should not be renormalized. We should then be able to use this `symmetry' to classify universality classes of phase transitions with non-integer $N$. In particular, this would provide a robust explanation why phase transitions of loop models on different lattices are described by the same CFT.

\end{enumerate}

\section{The Brauer algebra}
\label{sec:brouer}

Our first goal is to explain the mathematical meaning of the Kronecker delta tensor, $\delta_{IJ}$, for non-integer $n$.\footnote{Small $n$ will stand for a real number, which may or may not be integer. Capital $N$ will be reserved for positive integers: $N\in\bZ_+$.}
In the usual naive approach, these tensors are manipulated using a few simple rules:\footnote{Here we do not distinguish upper and lower indices, but later we will.}
\begin{subequations}
	\begin{gather}
	\delta_{IJ}=\delta_{JI},\label{eq:rule1}\\ 
	\delta_{IJ}\delta_{JK}=\delta_{IK},\label{eq:rule2}\\ 
	\delta_{IJ}\delta_{JI}=n\,. \label{eq:rule3}
	\end{gather}
\end{subequations}
From experience, these rules are consistent, i.e.~give the same answer when applied in any order. For instance we have associativity:
\beq
(\delta_{IJ}\delta_{JK})\delta_{KL}=\delta_{IJ}(\delta_{JK}\delta_{KL})=\delta_{IL},\qquad\text{etc.}
\eeq
For $n=N\in\bZ_+$, consistency is guaranteed by properties of matrix multiplication, but what are we doing for $n\in\bR$ and why is this consistent?

The answer is: stop thinking of $\delta_{IJ}$ as a tensor and $I,J$ as indices in a vector space, since vector spaces of non-integer dimension do not exist. Instead, view it as a notation for a string connecting a pair of points labelled $I$ and $J$: 
\beq
\delta_{IJ}\equiv \includegraphics[trim=0 2.5em 0 0,scale=0.4]{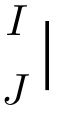}
\eeq\\[-0.5em]
Then, \reef{eq:rule1} is manifest, \reef{eq:rule2} becomes a rule for concatenating strings and erasing the midpoint, while \reef{eq:rule2}
replaces a circle by $n$: (dashed line shows where the concatenation happens)
\beq
\raisebox{-2em}{\includegraphics[scale=0.6]{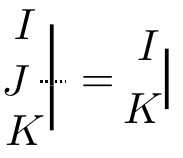}}\qquad\qquad
\raisebox{-1em}{\includegraphics[trim=0 0 0 0,scale=0.6]{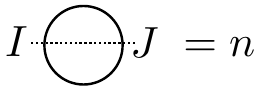}}\,.
\eeq
Similarly, products of several $\delta$-tensors are replaced by diagrams containing several strings, e.g.\footnote{It does not matter which string passes above which, only who is connected to whom.}
\beq
\raisebox{-1em}{\includegraphics[trim=0 0em 0 0,scale=0.45]{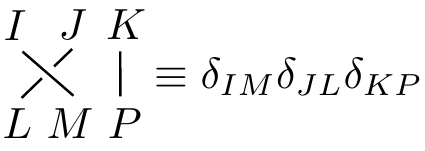}}\,.
\eeq
Such string diagrams can then be multiplied by concatenation: putting them one on top of the other and erasing middle vertices. Each loop produced in the process is erased and replaced by a factor of $n$. This operation replaces contraction of $\delta$-tensors. The resulting string diagram may then be `straightened up', just to improve the visibility of the remaining connections. Here are two examples:
\beq
\raisebox{-1em}{\includegraphics[scale=0.6]{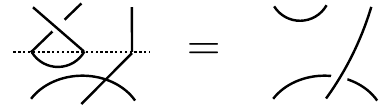}}\,,\qquad \raisebox{-1em}{\includegraphics[scale=0.6]{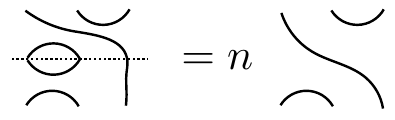}}\,.
\eeq
For a fixed integer $k$, consider then all string diagrams with $k$ points in the bottom and top rows. By taking formal linear combinations of we can construct a finite-dimensional vector space. By the above multiplication rules this vector space turns into an algebra, called the Brauer algebra $B_k(n)$. We can state our first lesson:

\emph{For $n\in \bR$, formal manipulation rules involving $\delta$-tensors are nothing but operations in the Brauer algebra of string diagrams. The rules are consistent because the Brauer algebra exists (and vice versa).}

We can also generalize the Brauer algebra, considering $(k_1,k_2)$-diagrams with $k_1$ points in the bottom and $k_2$ in the top row. The product of a $(k_1,k_2)$-diagram with a $(k_2,k_3)$-diagram is defined as the concatenated $(k_1,k_3)$-diagram, times a factor $n^{\#\text{(erased loops)}}$. We shall call this mathematical structure the `category $\Rephat\,O(n)$', and it is the first step towards the Deligne category $\Reptilde\,O(n)$. Using the language of category theory will assist us in both reasoning about these algebraic structures and in simplifying notation. We will start using it in section \ref{sec:rephat}, but for now let us proceed thinking more pictorially by using string diagrams.

\section{Operational definition of `$O(n)$ symmetry', $n\in\bR$}
\label{sec:oper}
 
To define something as fundamental as symmetry needs care. We will try to give a reasonably general definition, making sure that it is neither circular nor tautological. We will also separate the `definition' from the `construction': first we shall define our notion of symmetry, and then we shall construct symmetric models.

Imagine someone hands us a model as a black box, an oracle which can be queried for values of observables.\footnote{By observable we mean any measurable quantity which has a numerical value, like a correlation function. No relation to observable in quantum mechanics.} How can we determine whether the model has a symmetry?

\subsection{First attempt}
If it is a model of spins ($N\in\bZ_+$), we might act as follows. Let us query the oracle for a correlator of multiple spins $s(x_i)$. Probing all possible spin components $I_i=1\ldots  N$, we can check that the correlator is an $O(N)$ invariant tensor, i.e. expandable in a basis of products of $\delta_{I_i I_j}$. For instance, for a 4pt function we should find:
\beq
\langle s_{I}(x_1) s_{J}(x_2)s_{K}(x_3)s_{L}(x_4)\rangle = C_1(x_i) 
\delta_{IJ}\delta_{KL}+C_2(x_i)\delta_{IK}\delta_{JL}+C_3(x_i)\delta_{IL}\delta_{JK}\,.
\label{eq:ssss}
\eeq
If this holds for any correlator we check, it is tempting to declare that the model has $O(N)$ symmetry. 

Consider now a loop model and try to devise a test for whether there is a non-integer `$O(n)$ symmetry', whatever that might mean. The analogue of the above would be to query the oracle for correlation functions of defect operators $D$ where lines can end. Any such correlator is a linear combination of string diagrams, e.g. 4pt function  
\beq
\< D(x_1) D(x_2) D(x_3) D(x_4)\>= C_1(x_i)\, \includegraphics[trim=0 1em 0 0,scale=0.65,angle=180]{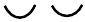} +
C_2(x_i)\,\includegraphics[trim=0 1em 0 0,scale=0.65,angle=180]{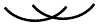} +
C_3(x_i)\, \includegraphics[trim=0 1em 0 0,scale=0.65,angle=180]{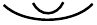} \,,
\label{eq:DDDD}
\eeq
where the numerical values of coefficients will be provided by the oracle, while the string diagrams are just formal placeholders saying who's connected to whom. Compared to \reef{eq:ssss}, we lost the external indices. In fact, it looks like Eq.~\reef{eq:DDDD} contains zero information about symmetry: it would be valid for an $O(n)$ loop model with any $n$, as well as for \emph{ad hoc} loop model without any symmetry. This seeming uselessness of \reef{eq:DDDD} compared to \reef{eq:ssss} is paradoxical.

The paradox will be resolved as follows: in fact both \reef{eq:ssss} and \reef{eq:DDDD} are insufficient to determine if we have a symmetry. In addition to each observable having the expected form, we will have to examine relations between different observables. Think of it as ``looking under the hood'' of the oracle. 

\subsection{A definition which works}

We will first give a slightly more detailed way to check for $O(N)$ symmetry of the spin models. Instead of correlators, let us query the oracle for `incomplete partition functions':
\beq
Z_k :=Z|_{s(x_i)=s_i,i=1\ldots  k}\,,
\eeq 
computed with $k$ spins held at some fixed values, while the rest are integrated over. We can call this the `joint probability distribution' of the $k$ chosen spins. It is a powerful observable which contains information about all correlators: the latter can be computed from it performing the remaining integrals weighted by the spin components of interest.

An oracle test of $O(N)$ symmetry, $N\in\bZ_+$, consists in checking two properties of $Z_k$'s:
\begin{align}
&1. \text{ (Invariance) Each $Z_k$ must be a function of scalar products $s_i\cdot s_j$.}
\label{eq:invar} \\[2mm]
&  2. \text{ (Consistency) If we integrate $Z_k$ over one of the spins, we should get $Z_{k-1}$:}\nn\\
&\hspace{1cm} \int ds_1\, Z_k = Z_{k-1}\,,
\label{eq:consist}\\
&\text{where in the r.h.s. the spin at $x_1$ is no longer fixed, while the rest remain fixed to the same values.}\nn
\end{align}
Here $\int ds$ is an $O(N)$ invariant integration from the model's path integral. For each even $m\in \bZ_{\ge0}$ we have:
\beq
\int ds\, s^{I_1}\ldots  s^{I_{m}} = \calJ_{m} \bigl[\delta^{I_1 I_2}\ldots  \delta^{I_{m-1}I_m}+\text{ other pairings}\bigr],
\label{eq:ONint}
\eeq
where all pairings are included with the same overall coefficient $\calJ_m$. For $m$ odd there are no pairings so the integral vanishes. One example is the integral over the unit sphere, for which 
\beq
\label{eq:Jmsphere}
\calJ_m=\frac{\text{Area}(S^{N-1})}{N(N+2)\ldots (N+m-2)}\,.
\eeq
We want to allow models including radial degrees of freedom, which effectively modifies $\calJ_m$.\footnote{ Imagine that the original field of the model is $\phi^I=r s^I$ where $s^I$ lives on the sphere and $r$ is the radial direction, and the integration measure on each site includes a factor $\int dr f(r)$ with some fixed weight $f(r)$. Then once we integrate out the $r$ direction, the resulting effective coefficients $\calJ_m$ will be modified w.r.t.~\reef{eq:Jmsphere} by factors $\int_0^\infty dr\, r^m\, f(r)$.} So our general $O(N)$ invariant integral is defined by \reef{eq:ONint} with some fixed but {\bf arbitrary} $\calJ_m$. We require that one $\int ds$ should work for all spin integrations and for all $k$'s. Notice that it is not necessary to query the oracle directly about this integral: knowing $Z_k$'s, we can decide if a $\int ds$ exists which makes \reef{eq:consist} true, and so determine the $\calJ_m$'s if it does exist.

Furthermore, this point of view generalizes to the loop model case. The loop analogue of leaving $k$ spins unintegrated is to allow an arbitrary number $m_i$ of lines to end at the corresponding lattice points $x_i$. The line ends can be interconnected in an arbitrary way in the bulk of the lattice. These connections are the analogues of $(s_i\cdot s_j)$ in the spin model. Then, the loop version of $Z_k$ is an infinite formal sum:
\beq
Z_k = \sum_{m_1=0}^\infty \ldots  \sum_{m_k=0}^\infty \sum_{\alpha} C_\alpha \raisebox{-2em}{\includegraphics[trim=0 0em 0 0,scale=0.65]{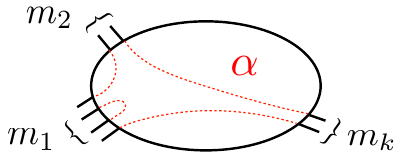}}\,,
\label{eq:Zkloop}
\eeq
where $\alpha$ labels all possible pairwise interconnections of external line ends (with one example shown), and we may query the oracle for the numerical coefficients $C_\alpha$. Simple defect correlators are contained in $Z_k$: they correspond to all $m_i=1$.
Unlike its spin analogue \reef{eq:invar}, Eq.~\reef{eq:Zkloop} is not yet a test of anything: the test will come from the consistency condition.

By analogy with \reef{eq:ONint}, we define the loop model integral as a permutation-invariant formal sum of string diagrams pairing $m$ points with some numerical coefficients $\calJ_m$:
\beq
\int = \sum_{m\ge 0\text{ even}} \calJ_m \int_m\ ,\qquad\qquad \int_m = \overbrace{\vphantom{I}\includegraphics[trim=0 1em 0 0,scale=0.6,angle=180]{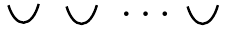}}^{m\text{ points}}+\text{ permutations}\,.
\label{eq:intloops}
\eeq
Consistency condition says that if we take $Z_k$ in the form \reef{eq:Zkloop} and apply the integral \reef{eq:intloops} at point $x_1$, we should get $Z_{k-1}$ corresponding to the remaining points $x_2,\ldots , x_k$:
\beq
\int_{\text{at }x_1} Z_k = Z_{k-1}\,.
\eeq
Integral is applied by concatenating the r.h.s.~of \reef{eq:intloops} with the lines in $Z_k$ ending at $x_1$. Only the terms with $m=m_1$ give non-zero answer and are simplified by the Brauer algebra rules, replacing circles by factors of $n\in\bR$. If we can define the integral (i.e. specify $\calJ_m$'s) so that this consistency condition holds, we say that the loop model has an `$O(n)$ symmetry', where $n$ is not necessarily an integer.

\section{Translation to categories}
\label{sec:transcat}
Having described in some detail how $O(n)$ models and their symmetry for $n\in\bR$ can be defined, let us introduce some category theory language to describe this more succinctly. A category is a collection of objects with connections (`morphisms') among them which can be composed.\footnote{See appendix \ref{sec: category} for an exposition of the necessary parts of category theory.} For the full picture we will need to introduce three categories: 
\begin{itemize}
	\item $\Rep\, O(N)$, for $N\in\bZ_+$, of finite-dimensional complex representations of the group $O(N)$
	\item $\Rephat\,O(n)$, the category of string diagrams
	\item The Deligne category $\Reptilde\,O(n)$ built on top of $\Rephat\,O(n)$.
\end{itemize} 
The latter two categories are defined for $n\in\bR$ and provide an `analytic continuation' of the former one, in a sense which will be made precise. It is convenient to start with $\Rephat\,O(n)$

\subsection{$\Rephat\,O(n)$, $n\in\bR$}
\label{sec:rephat}
This category collects the string diagrams into a single algebraic structure. Consider the diagrams introduced in section \ref{sec:brouer}. Each diagram is a series of strings connecting $k_1$ points at the bottom to $k_2$ points at the top. For each $k\in\mathbb Z_{\geq 0}$ we define an object $[k]$ in $\Rephat\,O(n)$. Each diagram $f$ connecting $k_1$ points at the bottom to $k_2$ points at the top is an element of the set $\Hom([k_1]\rightarrow[k_2])$, which we call the morphisms of $\Rephat\,O(n)$. We will often write $f:[k_1]\rightarrow[k_2]$ in order to denote an element $f\in\Hom([k_1]\rightarrow[k_2])$. We can also consider linear combinations of string diagrams, turning each set of morphisms $\Hom([k_1]\rightarrow[k_2])$ into a vector space. In particular the combination with zero coefficients is the zero morphism ${0:[k_1]\rightarrow[k_2]}$, which satisfies $0f = 0$ and $f+0 = f$ for any other morphism $f$.

To define a category we must have a way to compose morphisms. In $\Rephat\,O(n)$ we compose any two diagrams $f:[k_1]\rightarrow[k_2]$ and $g:[k_2]\rightarrow[k_3]$ by stacking $g$ on top of $f$, and simplifying them using the rules of section \ref{sec:brouer} (in particular every loop gives a factor of $n$). This gives us a new morphism, which we denote as $g\circ f$. This composition rule is clearly associative:
\begin{equation}f\circ (g\circ h) = (f\circ g)\circ h\,.\end{equation}
Composition (and the tensor product defined below) extends by linearity to arbitrary linear combinations of diagrams. Because the morphisms $\Hom([k_1]\to[k_2])$ are now vectors and the morphism composition is linear, we call $\Rephat\,O(n)$ a \emph{linear category}.

We would like to warn the reader that $[k]$ are just names of inequivalent objects, and we do not think of $[k]$ as sets of anything (which means that our category is `abstract' as opposed to `concrete').\footnote{So although we may pictorially imagine $[k]$ by drawing $k$ points, it would be inappropriate to think of $[k]$ as a `set consisting of these points'.} Also morphisms are not thought as maps from one set to another, but just as abstract elements of vector spaces $\Hom([k_1]\to [k_2])$ on which the associative composition operation is defined. In this sense $f:[k_1]\rightarrow[k_2]$ is just a convenient notation.

For every object $[k]$, there is an identity morphism $\text{id}_{[k]}:[k]\rightarrow[k]$ which is the diagram:
\begin{equation}
\text{id}_{[k]} = \raisebox{-0.75em}{\includegraphics[trim=0 0em 0 0,scale=0.4]{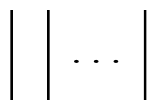}}\,.
\end{equation}
This diagram acts trivially when composed with other diagrams 
$$\text{id}_{[k]}\circ f = f = f\circ\text{id}_{[k]}.$$
The objects, morphisms, composition law $\circ$ and identities $\text{id}_{\bf a}$ together form the data we need to define the category $\Rephat\,O(n)$.

\begin{figure}[t!]
	\centering
	\includegraphics[width=0.5\linewidth]{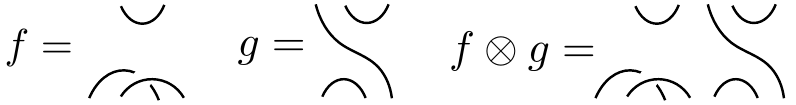}
	\caption{\label{fig-fotimesg} An illustration of the tensor product definition. Here $f:[4]\to[2]$, $g:[3]\to[3]$, and $f\otimes g: [7]\to [5]$.
	}
\end{figure}

In the category $\Rephat\,O(n)$ there is another way we can combine two diagrams $f$ and $g$: we can place them next to each other, to define $f\otimes g$, see Fig.~\ref{fig-fotimesg}.
We call $\otimes$ a ``tensor product'', for reasons that will become apparent later. For any two objects, we define $[k_1]\otimes[k_2] = [k_1+k_2]$. This means that for morphisms $f:[k_1]\rightarrow[k_2]$ and $g:[l_1]\rightarrow[l_2]$, their tensor product is a morphism from $f\otimes g:[k_1]\otimes[l_1]\rightarrow[k_2]\otimes[l_2]$. The tensor product $\otimes$ turns $\Rephat\,O(n)$ into a `monoidal category' (see appendix \ref{sec: category}).

We can now translate correlation functions in loop models into this new language. If we take the correlator of $k$ defect operators $D(x)$, as for instance in \eqref{eq:DDDD}, the answer will be a linear combination of string diagrams. This is simply a morphism from $[k]\rightarrow[0]$. We can think of each defect operator $D(x)$ as being associated with a single dot $[1]$, and that when we have multiple operators $D(x_1)\ldots  D(x_k)$ we should associated these with the tensor product $[1]\otimes\ldots \otimes[1] = [k]$. The correlation function is then a morphism from $[k]\rightarrow [0]$.

This may seem all well and good, but how do we connect this to physical observables, which are numbers? Any morphism can be expanded in the basis of string diagrams, and the expansion coefficients are numbers which one can physically measure or compute in (say) a Monte Carlo simulation. Equivalently, we can think as follows. Let us consider morphisms from $[0]\rightarrow[0]$. Because in the Brauer algebra closed loops give us factors of $n$, all morphisms from $[0]\rightarrow[0]$ are proportional to $\text{id}_{[0]}$, which is the empty diagram. For this reason, if we have any morphism $f:[0]\rightarrow[k]$, then we can compose it with our correlator to find
\begin{equation}\label{eq:measD}
\langle D(x_1)\ldots  D(x_k)\rangle \circ f= \lambda\,\text{id}_{[0]}
\end{equation}
for some number $\lambda$, which is physically measurable. Choosing different $f$'s we can determine all the expansion coefficients of the correlator {when $n$ is not integer}.\footnote{This follows from the semisimplicity of the Brauer algebra \cite{Wenzl88}. A similar statement holds in any semisimple category, see proposition \ref{pr:ndeg} for more details. In the next section we shall see that for integer $n$ certain structures become ``null'' and vanish when composed with any morphism $[0]\rightarrow[m]$.}

Now let us consider the `incomplete partition functions' $Z_k$ defined in the previous section. As we can see in \eqref{eq:Zkloop}, for each set of values $m_1, \ldots  m_k$ it gives rise to a morphism
\beq
\sum_{\alpha} C_\alpha \raisebox{-2em}{\includegraphics[trim=0 0em 0 0,scale=0.65]{fig-alpha.pdf}} \equiv f_{m_1,\ldots ,m_k}\in\Hom\left([0]\to [m_1]\otimes[m_2]\ldots \otimes [m_k]\right)\,.
\label{eq:ZktoCat2}
\eeq
In time, we shall see how to define a notion $\oplus$ analogous to a direct sum, and we will then see that $Z_k$ itself can be thought of as a morphism.

Finally, let us consider the loop model integration we define in \eqref{eq:intloops}. For every even $m$, that equation defines a morphism $\int_m:[0]\rightarrow[m]$ as a sum of diagrams from $[0]\rightarrow[m]$. 
Like $Z_k$, $\oplus$ will allow us to think of the morphisms $\int_m$ as part of a single morphism $\int$.

\subsection{$\Rep\,O(N)$, $N\in\bZ_+$, and its relation to $\Rephat\,O(n)$}
Representation theory of $O(N)$ studies representations, and the $O(N)$ covariant maps between them. There is a standard way to package all of this information into a single algebraic structure: category $\Rep\,O(N)$. The objects ${\bf a}, {\bf b}, \ldots  $ of $\Rep\,O(N)$ are representations of $\Rep\,O(N)$, and the morphisms $\Hom({\bf a}\rightarrow{\bf b})$ are $O(N)$ covariant tensors between the representations ${\bf a}$ and $\bf b$. We will denote the identity tensor mapping a representation $\bf a$ to itself by $\text{id}_{\bf a}$.

As a category, $\Rep\,O(N)$ has a lot of additional structure. For instance we have a tensor product $\otimes$, which we can use to combine any two representations $\bf a$ and $\bf b$ into a new representation ${\bf a}\otimes{\bf b}$. We can also tensor together any two covariant tensors $f:{\bf a}\rightarrow{\bf b}$ and $g:{\bf c}\rightarrow{\bf d}$ to produce a new covariant tensor $f\otimes g:{\bf a}\otimes{\bf c}\rightarrow{\bf b}\otimes{\bf d}.$ We also have the trivial representation $\bf 1$, which acts trivially on any other representation.

Because we have a tensor product, $\Rep\,O(N)$ is what is known as a monoidal category. We give a detailed description of what this means in appendix \ref{sec:MONOID}, but in practise the rules for manipulating $\otimes$ in a monoidal category generalize the usual tensor product rules.

We can also consider category $\Rephat\,O(n)$. That category was defined in section \ref{sec:rephat} for any $n\in\bR$, but now we would like to consider it for $n=N\in\bZ_+$ and determine its relationship to the category $\Rep\,O(N)$. To do so we need the notion of a functor, that is, a map between two categories. A functor $F$ between categories $\fcy C$ and~$\fcy D$ associates to every object $\bf a\in\fcy C$ an object $F(\bf a)\in\fcy D$ and to every morphism $f:\bf a\rightarrow\bf b$ a morphism $F(f):F({\bf a})\rightarrow F({\bf b})$, such that function composition is preserved:
\begin{equation}F(f\circ g) = F(f)\circ F(g)\,,\quad F(\text{id}_{\bf a}) = \text{id}_{F(\bf a)}.\end{equation}
If we have a functor between between two monoidal categories, we also want it to preserve the tensor product:\footnote{In mathematical parlance this is a strict monoidal functor. }
\begin{equation}
F({\bf a}\otimes{\bf b}) = F({\bf a})\otimes F({\bf b})\,,\quad F(f\otimes g) = F(f)\otimes F(g)\,,\quad F({\bf 1}_{\fcy C}) = {\bf 1}_{\fcy D}.
\end{equation}

In this language, the relation between $\Rephat\,O(n)$ and $\Rep\,O(N)$ is expressed by saying that for $N\in\bZ_+$ we can construct a functor $\fcy S: \Rephat\,O(N)\rightarrow \Rep\,O(N)$ from the former to the latter. This functor takes string diagrams and relates them to invariant tensors in $\Rep\,O(N)$. Under $\fcy S$, we map the objects $[k]$ to ${\bf N}^{\otimes k}$ where ${\bf N}$ is the $N$-dimensional vector representation of $O(N)$. The morphisms from $[k_1]\rightarrow[k_2]$ are translated into invariant tensors from ${\bf N}^{\otimes k_1}\rightarrow{\bf N}^{\otimes k_2}$ by attaching indices $a_1,a_2,\ldots $ to each dot and then associating to each string a tensor $\delta_{a_ia_j}$. The function $\fcy S: \Hom([k_1]\rightarrow[k_2])$ to $\Hom({\bf N}^{\otimes k_1}\rightarrow{\bf N}^{\otimes k_2})$ is surjective; in category theoretic parlance, we say that $\fcy S$ is a full functor.

Because $\fcy S$ is a functor, if we compose invariant diagrams using the diagrammatic rules and then apply $\fcy S$, we will get the same result as if we first apply $\fcy S$ and then compose the tensors. In this sense, the category $\Rephat\,O(n)$ captures the essential rules of tensor composition in $\Rep\,O(N)$. There are however some very important differences between the two categories. 

The only objects in $\Rephat\,O(n)$ are of the form $[k]$, and under $\fcy S$ these map to $\fcy S([k]) = {\bf N}^{\otimes k}$ in $\Rep\,O(N)$. But we know that there are many other representations in $\Rep\,O(N)$, such as symmetric and antisymmetric tensors. Furthermore, in $\Rephat\,O(n)$ we have no notion of direct sum $\oplus$, while this notion is very important in $\Rep\,O(N)$, as it allows us to decompose tensor products into sums of irreducible representations. This leads us to ask: what is the $\Rephat\,O(n)$ analogue of the decomposition of ${\bf N}^{\otimes k}$ into irreducible representations?

\subsection{Irreducible representations}
To answer this question, let us first consider how to recover irreducible representations in $\Rep\,O(N)$ from a more categorical point of view. In a semisimple\footnote{See appendix \ref{sec:LINCAT} for a precise definition of seimisimplicity.} category, an object ${\ba}$ is called simple if every morphism $\Hom({\ba}\rightarrow{\ba})$ is proportional to the identity $\text{id}_{\ba}$. Thus, by Schur's lemma, irreps are precisely the simple objects ${\ba}\in\Rep\,O(N)$.

We know that  ${\bf N}^{\otimes k}\in\Rep\,O(N)$ will be decomposable into the direct sum of irreducible representations:
\begin{equation}\label{eq:NktoSim}
{\bf N}^{\otimes k} = {\ba_1}\oplus{\ba_2}\oplus\ldots
\end{equation}
We know what this equation means in the language of the usual representation theory, but now let us translate it into category theory.
For each representation $\ba_i$ appearing in \eqref{eq:NktoSim} there is a pair of morphisms; $\pi_{\ba_i}:{\bf N}^{\otimes k}\rightarrow{\ba_i}$ which projects down onto $\ba_i$, and $\iota_{\ba_i}:{\ba_i}\rightarrow{\bf N}^{\otimes k}$ which embeds $\ba_i$ into ${\bf N}^{\otimes k}$. These satisfy the relationship
\begin{equation}
\label{eq:projIda}
\pi_{\ba_i}\circ\iota_{\ba_i} = \text{id}_{\ba_i}\,.
\end{equation}
Composing these maps in the reverse order, we define a `projector' morphism:
\begin{equation}
\label{eq:projIda1}
P_{\ba_i} = \iota_{\ba_i}\circ\pi_{\ba_i}\in\Hom({\bf N}^{\otimes k}\rightarrow{\bf N}^{\otimes k})\,.
\end{equation}
It follows from \eqref{eq:projIda} that these morphisms are idempotent: $P_{\ba_i}^2 = P_{\ba_i}$, justifying the name ``projector''. Furthermore, we demand that for any two terms in \reef{eq:NktoSim}, the projectors are `orthogonal':
\begin{equation}\label{eq:projij}
P_{\ba_i}\circ P_{\ba_j} = \iota_{\ba_i}\circ\pi_{\ba_i}\circ\iota_{\ba_j}\circ\pi_{\ba_j} = 0\,\qquad (i\ne j)\,.
\end{equation}
If ${\ba_i}$ and ${\ba_j}$ are two different irreps, then orthogonality is automatic, because the middle piece $\pi_{\ba_i}\circ\iota_{\ba_j}:{\ba_i}\rightarrow{\ba_j}$ is then a morphism between two distinct irreducible representations, and such morphisms are trivial by Schur's lemma. If the direct sum contains several copies of the same irrep, then we can achieve orthogonality of the corresponding projectors by a change of basis.

The decomposition \eqref{eq:NktoSim} of ${\bf N}^{\otimes k}$ into irreducible representations thus corresponds to a decomposition of the identity morphism
\begin{equation}\label{eq:idNP}\text{id}_{\bf N^{\otimes k}} = \sum_i P_{{\ba}_i}\end{equation}
as a sum of mutually orthogonal projectors. Furthermore, because each $\ba_i$ is irreducible this decomposition is maximal, that is, $P_{\ba_i}$ cannot be further decomposed.

Although in $\Rephat\,O(n)$ we cannot decompose $[k]$ as the sum of simpler objects, we can still generalize \eqref{eq:NktoSim} by decomposing $\text{id}_{[k]}$ as the sum of morphisms which are idempotent and mutually orthogonal, since these concepts make sense even in an abstract category setting. Taking the example of $k=2$, the three string diagrams forming the basis of $\Hom([2]\rightarrow[2])$ are
\beq	
	\includegraphics[width=0.4\linewidth]{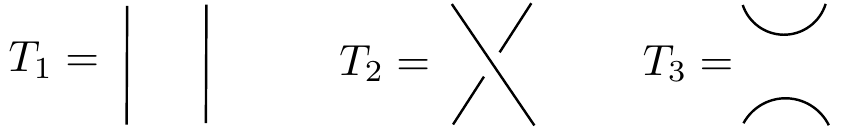}\,.
\eeq
From these we can define the three idempotent morphisms
\begin{equation}
\label{eq:idemp2}
P_{\bf 1} = \frac 1n T_3\,,\quad P_{\bf S} = \frac 12(T_1+T_2)-\frac 1n T_3\,,\quad P_{\bf A} = \frac 12(T_1-T_2)\,,
\end{equation}
which are mutually orthogonal and satisfy
\begin{equation}
P_{\bf 1} + P_{\bf S} + P_{\bf A} = T_1 = \text{id}_{[2]}\,.
\end{equation}
By analogy to $\Rep\,O(N)$, it is tempting to say that $[2]$ ``splits into three irreducible representations $\bf 1$, $\bf S$ and $\bf A$". For the moment we are not allowed to use this language, but the Deligne category $\Reptilde\,O(n)$ will allow us to do so.

This analogy with $\Rep\,O(N)$ can be further extended by defining a `trace' on $\Rephat\,O(n)$. Recall that in $\Rep\,O(N)$ we can for any morphism $f:{\bf N}^{\otimes k}\rightarrow{\bf N}^{\otimes k}$ take the trace $\text{tr}(f)\in\mathbb C$ by summing over all diagonal components of $f$. For a projector $P_{\ba_i}$
\begin{equation}\label{eq:trdim}\text{tr}(P_{\ba_i}) = \text{dim}({\ba_i}) \end{equation}
where $\text{dim}({\ba_i})$ is the dimension of ${\ba_i}$.

Likewise, in $\Rephat\,O(n)$ we define the trace of a diagram $f:[k]\rightarrow[k]$ as
\begin{equation}
\label{eq:traceo(n)}
\text{tr}(f)=\raisebox{-2.2em}{\includegraphics[trim=0 0 0 0,scale=0.4]{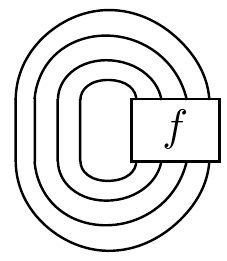}}\,,
\end{equation}
 which is just a number: $n$ to the power of number of closed loops in the r.h.s. We extend this definition by linearity to any morphism in $\Hom([k]\rightarrow[k])$. Now let us compute
\begin{equation}\label{eq:2decomp}
\text{tr}(P_{\bf 1}) = 1\,,\quad \text{tr}(P_{\bf S}) = \frac{(n+2)(n-1)} 2\,,\quad \text{tr}(P_{\bf A}) = \frac{n(n-1)}{2}\,.
\end{equation}
For $n\in\mathbb Z_+$ these numbers are precisely the dimensions of the trivial representation, the traceless symmetric tensor, and the antisymmetric tensor; indeed, the functor $\fcy S$ maps these three idempotents onto the relevant projectors in ${\bf N}^{\otimes 2}\rightarrow{\bf N}^{\otimes 2}.$ Furthermore $\fcy S$ preserves the trace, which is why the ``dimensions'' in $\Rephat\,O(n)$ match those of $\Rep\,O(n)$. So $\Rephat\,O(n)$ gives a precise algebraic meaning to, and hence allows us to extend, the usual dimension formulae and fusion rules for any value of $n$.

We can repeat this procedure for each $[k]$, searching for idempotents in $\text{Hom}([k]\rightarrow[k])$. For integer $N$ these idempotents correspond to subrepresentations of ${\bf N}^{\otimes k}$. The sum of two orthogonal idempotents is also an idempotent, just as the direct sum of any two representations is itself a representation. We are most interested in the \emph{simple idempotents}, which cannot themselves be decomposed as the sum of non-trivial idempotents. Under the action of $\fcy S$ they correspond to projectors onto simple objects (irreps) in $\Rep\,O(N)$. By computing the trace of the idempotents, we can compute the dimension of the corresponding representations.

Let us now consider \eqref{eq:2decomp} for $n = 1$. In this case we find that
\begin{equation}\text{tr}(P_{\bf S}) = \text{tr}(P_{\bf A}) = 0\,.\end{equation}
As we know from linear algebra, the only idempotent matrix with null trace is the null matrix, and so both these idempotents belong to the kernel of $\fcy S$. 

More generally we will define a null idempotent to be a simple idempotent with zero trace. They can be found for any $N\in\mathbb Z_+$. Consider for instance the morphism in $\Hom([k]\rightarrow[k])$ given by the antisymmetrized linear combination of diagrams:
\begin{equation}
P_{{\bf A}^k} = \raisebox{-0.6em}{\includegraphics[trim=0 0em 0 0,scale=0.5]{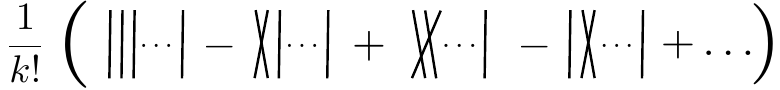}},
\label{eq:PAk}
\end{equation}
This is clearly a simple idempotent, and it is easy to compute the trace:
\begin{equation}
\text{tr}(P_{{\bf A}^k}) = \frac{n(n-1)...(n-k+1)}{k!}.
\label{eq:trPAk}
\end{equation}
If $n=N\in\mathbb Z_+$, this trace vanishes for $k\ge N$. So $P_{{\bf A}^k}$ is null, hence it belongs in the kernel of $\fcy S.$ In $\Rep\,O(N)$ the morphism $\fcy S(P_{{\bf A}^k})$ is the projection onto fully antisymmetric tensors, and indeed such tensors can only exist for $k<N$.

For $n\notin \mathbb Z_+$, null idempotents do not exist (Wenzl \cite{Wenzl88}). For instance, in this case the trace \reef{eq:trPAk} does not vanish for any $k$. (Although notice that it may be negative for $k>n+1$, so in this abstract setting non-integer `dimensions' may also become negative.) This is one manifestation of an important difference between integer and non-integer $n$: the Brauer algebra $\Hom([k]\rightarrow[k])$ is always semisimple\footnote{This means that the algebra is isomorphic to a direct sum of finite-dimensional matrix algebras.} for $n\not\in\mathbb Z$, but this is not true for integer $N$. If, however, for integer $N$ we take the \emph{quotient} of the Brauer algebra by null idempotents $\fcy N$, the result is a semisimple algebra \cite{Wenzl88}. As we have already seen, the functor $\fcy S$ naturally implements such a quotient. Conversely, one can show that the kernel of $\fcy S$ is precisely generated by the null idempotents, and hence $\text{Hom}([k]\rightarrow[k])/\fcy N$ is isomorphic to  $\text{Hom}({\bf N}^{\otimes k}\rightarrow{\bf N}^{\otimes k})$.

We have seen that the decomposition
\begin{equation}
\text{id}_{[k]} = \sum_{i = 1}^m P_i
\label{eq:iddecomp}
\end{equation}
 of $\text{id}_{[k]}$ as the sum of mutually orthogonal simple idempotents is closely related to the decomposition of ${\bf N}^{\otimes k}$ into the direct sum of irreps. Can we always decompose $\text{id}_{[k]}$ in this way, and is such a decomposition unique?
Wenzl \cite{Wenzl88} showed that the answer is yes, the decomposition \reef{eq:iddecomp} always exists, and is unique up to conjugation by invertible morphisms\footnote{A morphism $U$ is invertible if another morphism $U^{-1}$ exists such that $U^{-1}\circ U=\text{id}_{[k]}$. E.g., a single diagram in ${\rm Hom}[k]\to [k]$ is invertible if and only if each point at the bottom is connected to a point at the top.} $U\in\text{Hom}([k]\rightarrow[k]).$ 

Indeed, note that for any invertible $U\in\Hom([k]\rightarrow[k])$, the morphisms $U^{-1}P_iU$ are also simple idempotents summing to $\text{id}_{[k]}.$ Note as well that the trace on $\Rephat\,O(n)$ is, like the usual trace, cyclic:
\begin{equation}\text{tr}(AB) = \text{tr}(BA)\,,\end{equation}
and so the two decompositions have the same dimensions
\begin{equation}\tr(U^{-1}P_iU) = \tr(P_i)\,.\end{equation}

To summarize, we have found that $\Rephat\,O(n)$ contains all the information needed to construct the various representations in $\Rep\,O(N)$. The functor $\fcy S$ relates idempotents to the projections onto various representations in $\Rep\,O(N)$, while null idempotents are mapped to $0$. Unlike $\Rep\,O(N)$ however, we can define $\Rephat\,O(n)$ for any value of $n\in\mathbb R$, and give precise algebraic meanings to the 
dimensions of ``analytically continued'' $O(N)$ representations.

\subsection{Deligne's category $\Reptilde\,O(n)$}
\label{sec:reptON}
While $\Rephat\,O(n)$ goes a long way towards generalizing $\Rep\,O(N)$ to non-integer $n$, it does not have all of the usual properties of $\Rep\,O(N)$. We would like to be able to use rigorously the usual language of representation theory. In particular we would like to have objects corresponding to irreducible representations, and to be able to decompose other representations as the direct sum of these irreducible representations. Within $\Rephat\,O(n)$, this is impossible. We have idempotent morphisms which should somehow correspond to identity morphisms acting on simple objects analogous to `irreducible representations', but these simple objects are not part of the category $\Rephat\,O(n)$. Also, similar-looking idempotent morphisms will be found for some value of $k$ and for all larger $k$'s, and they need to be identified somehow. Finally, $\Rephat\,O(n)$ does not even have a notion of direct sum.

Fortunately, there exists a standard (and essentially unique) method to build from $\Rephat\,O(n)$ a new category with the properties we desire: the Deligne category $\Reptilde\,O(n)$. Our construction proceed in two standard steps: `Karoubi envelope' and `additive completion'.
These names may sound scary, but as you will see it is primarily a matter of introducing an appropriate language. 


The Karoubi envelope construction (first described by Freyd \cite{freyd1964}) can be applied to any category and allows to associate objects with idempotents. Here we describe the Karoubi envelope of $\Rephat\,O(n)$, which we call $\mathsf{Kar}\,O(n)$. The objects of this new category are pairs $([k],P)$ where $[k]\in\Rephat\,O(n)$ and $P:[k]\rightarrow[k]$ is idempotent (though not necessarily simple). We then define $\Hom(([k_1],P_1)\rightarrow([k_2],P_2))$ to be the subspace of morphisms ${f\in\Hom([k_1]\rightarrow[k_2])}$ such that
\begin{equation}\label{eq:KarCond}f\circ P_1 = f = P_2\circ f.\end{equation}
With this definition, $P$ becomes the identity morphism ${\rm id}_P$ on $([k],P)$.

We compose and tensor morphisms together using the usual morphism composition and tensor product in $\Rephat\,O(n)$, while for objects we define the tensor product as:
\begin{equation}
([k_1],P_1)\otimes([k_2],P_2) = ([k_1+k_2],P_1\otimes P_2)\,.
\end{equation}
It is straightforward to verify that if \eqref{eq:KarCond} is satisfied by $f$ and $g$, the it is also satisfied by $f\circ g$ and $f\otimes g$.

We can naturally embed $\Rephat\,O(n)$ into $\mathsf{Kar}\,O(n)$ via a functor which sends $[k]$ to the object $([k],\text{id}_{[k]})$.
Intuitively, we think of the object $([k],P)$ as a ``sub-object'' in $([k],\text{id}_{[k]})$. Categorically, this statement is made precise as follows. We can define $\pi_P:([k],\text{id}_{[k]})\rightarrow ([k],P)$ and $\iota_P:([k],P)\rightarrow([k],\text{id}_{[k]})$ to be the morphism $P$, which trivially satisfies \eqref{eq:KarCond}. It then follows that:
\begin{equation}
P = \iota_{P}\circ\pi_{P}\,,\quad \text{id}_{P} = \pi_{P}\circ\iota_{P},
\end{equation}
These equations are analogous to Eqs.~\reef{eq:projIda}, \reef{eq:projIda1} from the previous section, so that $\pi_P$ and $\iota_P$ are completely analogous to the projection and embedding morphisms there.

So, by considering idempotents from $[k]\to [k]$ in $\Rephat\,O(n)$, we construct objects in $\Kar\,O(n)$. When $n\neq\mathbb Z$, the simple idempotents give rise to simple objects. By repeating this construction for all $k$, we get all objects of $\Kar\,O(n)$. It's important to realize however that some of the so constructed objects will be isomorphic. This should not come as a great surprise.
We can consider for analogy what happens in $\Rep\,O(N)$. In this category, if ${\bf a}$ is a subrepresentation of ${\bf N}^{\otimes k}$ then we also can embed ${\bf a}$ within ${\bf N}^{\otimes (k+2)}$. Also, in $\Rep\,O(N)$ the same representation ${\bf a}$ 
may occur several times inside ${\bf N}^{\otimes k}$ for a given $k$.

Indeed, in $\Rep\,O(N)$, we think of two representations ${\bf a}$ and ${\bf b}$ as being ``the same'' if there is an isomorphism between them; that is, there is an $O(N)$ covariant map $f:{\bf a}\rightarrow{\bf b}$ which has an inverse. Likewise, in $\mathsf{Kar}\,O(n)$ we should consider objects to be ``the same'' if there is an isomorphism (i.e., an invertible morphism) between them. This avoids a proliferation of redundant objects.

Let us consider an example. For an idempotent $P:[k]\rightarrow[k]$ and any of the three idempotents $P_{\ba}:[2]\to [2]$ from Eq.~\reef{eq:idemp2} we can construct a new idempotent $P\otimes P_{\bf a}$ from $[k+2]\to [k+2]$. This gives us distinct objects $([k],P)$ and $([k+2],(P\otimes P_{\bf a}))$ in $\Kar\,O(n)$. Let us see however that for $\ba={\bf 1}$ these two objects are isomorphic (in full analogy with the tensor product with the trivial representation).
The isomorphism is
\begin{equation} f = P\otimes \rotatebox{180}{\includegraphics[trim=0 1em 0 0,scale=0.4]{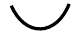}}\,,\quad f^{-1} = \frac 1 n P\otimes \includegraphics[trim=0 0em 0 0,scale=0.4]{fig-semicirc.pdf}.\end{equation}
It's easy to check that
\beq
f\circ f^{-1}=P={\rm id}_P\,,\quad f^{-1}\circ f = P\otimes P_{\bf 1}={\rm id}_{P\otimes P_{\bf 1}}\,.
\eeq
the morphisms in the r.h.s.~being the identity morphisms on the considered objects of $\Kar\,O(n)$.

For a pair of objects ${\bf a}$ and ${\bf b}$ in $\mathsf{Kar}\,O(n)$ (or more generally, in any linear category) the direct sum ${\bf a}\oplus{\bf b}$, if it exists is defined to be the unique (up to unique isomorphism) object satisfying:
\begin{enumerate}
\item There exists \emph{embedding} morphisms $\iota_1:{\bf a}\rightarrow {\bf a}\oplus{\bf b}$ and $\iota_2:{\bf b}\rightarrow {\bf a}\oplus{\bf b}$
\item There exist \emph{projection} morphisms $\pi_1:{\bf a}\oplus{\bf b}\rightarrow{\bf a}$ and $\pi_2:{\bf a}\oplus{\bf b}\rightarrow{\bf b}$
\item These maps satisfy the equations
\begin{equation}\pi_1\circ \iota_1 = \text{id}_{\bf a}\,,\quad \pi_2\circ \iota_2 = \text{id}_{\bf b}\,,\quad \iota_1\circ \pi_1 + \iota_2\circ \pi_2 = \text{id}_{\bf a\oplus \bf b}.
\end{equation}
\end{enumerate}
As can example, given two objects $([k],P_1)$ and $([k],P_2)$ such that $P_1P_2 = 0$, the direct sum $([k],P_1)\oplus ([k],P_2)$ does exist and is isomorphic to $([k],P_1+P_2)$. If we can decompose $\text{id}_{[k]}$ as the sum of mutually orthogonal idempotents $P_i$, we then find that in $\mathsf{Kar}\,O(n)$ ($\approx$ means isomorphic)
\begin{equation}([k],\text{id}_{[k]}) \approx ([k],P_1)\oplus ... \oplus ([k],P_{l}).\end{equation}

Any $\Hom([k]\rightarrow[k])$ contains the zero morphism, which is very trivially idempotent. In $\mathsf{Kar}\,O(n)$ the objects $([k],0)$ are all isomorphic, and the isomorphisms between them are also zero morphisms. For this reason, we will use $\bf 0$ to denote any object isomorphic to $([k],0)$. For any other object $([k],P)$ in $\mathsf{Kar}\,O(n)$, we see that there is a unique morphism $([k],P)\rightarrow{\bf 0}$ and ${\bf 0}\rightarrow([k],P)$, both of which are themselves also zero morphisms. In category theory language, $\bf 0$ is called the zero object. We can think of it as a ``zero-dimensional representation''.

It is not hard to check that $\bf 0$ is the additive identity, so that for any other $\ba\in\mathsf{Kar}\,O(n)$,
\begin{equation}{\bf 0}\oplus \ba \approx \ba \approx \ba\oplus{\bf 0}\,.\end{equation}
The existence of $\bf 0$ is also useful for various technical reasons, as it allows us to define notions of kernels, cokernel, and quotients.

Although we have defined the direct sum in $\mathsf{Kar}\,O(n)$, it does not always exist. Indeed, denote by $\bf n$ the object in $\mathsf{Kar}\,O(n)$ corresponding to $[1]$ in $\Rephat\,O(n)$. It is not hard to see that ${\bf n}\oplus{\bf 1}$ does not exist in $\mathsf{Kar}\,O(n)$, because ${\bf 1}$ can only be embedded into ${\bf n}^{\otimes k}$ for even $k$, while $\bf n$ can only be embedded for odd $k$. To fix this we define $\Reptilde\,O(n)$ to be the \emph{additive completion} of $\mathsf{Kar}\,O(n)$ to construct $\Reptilde\,O(n)$. This consists of formally defining for any series of objects ${\ba_1},...,{\ba_n}\in\mathsf{Kar}\,O(n)$ an object ${\ba_1}\oplus...\oplus {\ba_n}$, along with projection and embedding morphisms
\begin{equation}
\pi_{\ba_i}:{\ba_1}\oplus...\oplus {\ba_n}\rightarrow {\ba_i}\,,\quad \iota_{\ba_i}:{\ba_i}\rightarrow{\ba_1}\oplus...\oplus {\ba_n} 
\end{equation}
which satisfy
\begin{equation}
\iota_{\ba_i}\circ \pi_{\ba_i} = \text{id}_{\ba_i}\,,\quad \sum_{i = 1}^n \pi_{\ba_i}\circ\iota_{\ba_i} = \text{id}_{{\ba_1}\oplus...\oplus {\ba_n}}.
\end{equation}

The reader may again be concerned that $\Reptilde\,O(n)$ will have many duplicate objects. But because the direct sum is unique up to unique isomorphism, if two objects ${\bf a},{\bf b}\in\mathsf{Kar}\,O(n)$ already have a direct sum ${\bf c}\in\mathsf{Kar}\,O(n)$, then in $\Reptilde\,O(n)$ we will find that ${\bf c}\approx{\bf a}\oplus{\bf b}$. 
The additive completion merely makes sure that all objects have direct sums, and it is a construction that can be uniquely applied to any linear category. If we take the additive completion of a category that is already additive, then we construct a category which is equivalent to the old category.

\subsection{Properties of $\Reptilde\,O(n)$}
Having constructed $\Reptilde\,O(n)$, let us describe some of its properties and see how these generalize those of $\Rep\,O(N)$. For $n\not\in\mathbb Z$ it is a semisimple category, which means that any object ${\bf a}\in\Reptilde\,O(n)$ can be written as the direct sum of a finite number of simple objects. This is a consequence of Wenzl's result that the Brauer algebra is semisimple. 

Finding all simple objects in $\Reptilde\,O(n)$ and computing their dimensions is an exercise in combinatorics. As shown in \cite{Deligne}, simple objects ${\bf x}_{(\lambda_1\lambda_2...\lambda_k)}$ are labelled by a series of non-increasing numbers $\lambda_i\in\mathbb Z_+$, which we can think of as a sort of ``generalized Dynkin index''. The dimension of these objects is given by the interpolation of the Weyl dimension formula:
\begin{equation}
\text{dim}({\bf x}_{(\lambda_1\lambda_2...\lambda_k)}) = \prod_{i=1}^k\frac{(n-i-r)_{\lambda_i}}{(1-i+r)_{\lambda_i}}\prod_{1\leq i<j\leq k}\frac{(\lambda_i+\lambda_j+n-i-j)(\lambda_i-\lambda_j+j-i)}{(n-i-j)(j-i)}.
\end{equation}
This function is a polynomial in $n$, and is never zero for $n\not\in\mathbb Z,$ although for sufficiently large $k$ it may become negative. A particularly simple case are the objects ${\bf x}_{(1\ldots 1)}$ which correspond to the antisymmetric projectors \reef{eq:PAk} and whose dimension is given by \reef{eq:trPAk}. 

For integer $N$ the category $\Reptilde\,O(N)$ is no longer semisimple because of the presence of null idempotents. Under the Karoubi construction a null idempotent $P_o:[k]\rightarrow[k]$ gives rise to an object $([k],P_o)$ which is simple and has dimension $0$. As we show in proposition \ref{pr:zeroDim}, this is not possible in a semisimple category. So we can think $\Reptilde\,O(N)$ as a bigger version of $\Rep\,O(N)$, containing additional null idempotents and zero dimensional objects. More precisely, there is a full functor $\fcy F: \Reptilde\,O(N)\rightarrow\Rep\,O(N)$ which maps all of the null idempotents to $0$ and the zero dimensional objects to the $0$ dimensional representation.


\section{Deligne categories and lattice models}
\label{sec:lattice}

Having introduced Deligne's categories mathematically, we come back to their physical applications.
We will refer to the integer $n$ symmetry as a `group symmetry', while non-integer $n$ symmetry as a `categorical symmetry'.

States and local operators of a theory with categorical symmetry will be classified by the \emph{simple objects} of the associated category, while correlation functions and transfer matrices will be \emph{morphisms} of this category.

Let us describe how to construct the most general lattice model with a categorical symmetry, generalizing considerations of section \ref{sec:oper}.
For a group symmetry, one puts on every lattice site a ``spin'': a variable transforming in a certain representation, the simplest case being when this representations is irreducible and the same for all lattice sites. For categorical symmetry, we place on every lattice site $x$ a ``spin'' $s(x)$, which is an object of Deligne's category. We will write $s(x)$ or $s_x$ interchangeably. For simplicity, we consider the case where all these objects are isomorphic to a fixed simple object: $s(x)\approx \ba$. The setup of section \ref{sec:oper} can be seen as corresponding to $\ba=[1]$. For a group symmetry, we could also write $s^I(x)$ where $I$ is an index in the representation in which $s(x)$ transforms. For a categorical symmetry, only the index-free way of writing makes sense.

We next introduce interactions. For a group symmetry, interaction between two sites $H_{xy}$ is constructed contracting products of $s_x^I$ and $s_y^J$ with invariant tensors of the group. E.g., for spins $\vec s_x$ in the fundamental of $O(N)$, we can take any function of $\vec s_x\cdot \vec s_y$, $\vec s_x\cdot \vec s_x$, $\vec s_y\cdot \vec s_y$. We are not imposing here the constraint $\vec s_x\cdot \vec s_x=1$.
 
For a categorical symmetry, this is generalized by saying that $H_{xy}$ is a morphism
 \beq
 H_{xy} \in \text{Hom}\left({\bf 1}\to \bigoplus_{m_1,m_2\ge 0} S^{m_1}(s_x)\otimes S^{m_2}(s_y)\right)\,.
 \label{eq:Hxy}
 	\eeq
 	where $S^m$ is the symmetrized product (see section \ref{sec:sym}), appropriate here if we are defining a theory with bosonic lattice variables.
 	
 We can then consider the total interaction, which is naturally a member of a Hom space:
\beq
H\in \text{Hom}\, ({\bf 1}\to {\bf A})\,,\qquad {\bf A}=\bigoplus_{(m_x)} \bigotimes_x S^{m_x}(s_x) \,,\label{eq:HHom}
\eeq
where $\bigoplus_{(m_x)}$ is over all sequences $(m_x)$ of non-negative integers, one for each lattice point $x$. We can restrict to sequences with only two non-zero elements if $H=\sum_{xy} H_{xy}$ is a sum of pairwise interactions. Sequences with more non-zero elements would arise if $H$ contains interactions between three, four etc. sites at a time.

In physics, we often consider the exponentiated interaction $e^H$. To understand the precise meaning of this in the categorical world, we first define the morphism $\mu:{\bf A}\otimes{\bf A}\rightarrow{\bf A}$ which projects each term ${\Big(\bigotimes_x S^{m_x}(s_x)\Big)\otimes\Big(\bigotimes_x S^{m_x'}(s_x)\Big)}$ in ${\bf A}\otimes{\bf A}$ down to $\left(\bigotimes_x S^{m_x+m_x'}(s_x)\right)$ in $\bf A$. We can then define $H^n:{\bf 1}\to{\bf A}$ recursively as
\begin{equation}
 H^n := \mu\circ (H\otimes H^{n-1})\,,\qquad \text{with}\qquad H^0 := \eta\,,
\end{equation}
where $\eta:\bf 1\rightarrow{\bf A}$ is the trivial embedding\footnote{\label{foot:int}The object $\bf A$, together with the morphisms $\mu$ and $\eta$,  is an example of an \emph{algebra internal} to a category. The morphism $\mu$ satisfies a form of associativity, and $\eta$ acts as the identity on $\mu$:
\begin{equation}
\mu\circ(\mu\otimes\text{id}_{\bf A}) = \mu\circ(\text{id}_{\bf A}\otimes\mu)\,,\qquad \mu\circ(\text{id}_{\bf A}\otimes\eta) = \text{id}_{\bf A}\,.
\end{equation}
This is an example of a very general idea known as \emph{internalization}, whereby algebraic structures can be extended from sets to more general objects in a category. A general discussion can be found at \url{https://ncatlab.org/nlab/show/internalization}.} of $\bf 1$ into $\bigotimes_x S^{0}(s_x)\approx {\bf 1}\in\bf A$. Leaving convergence issues aside, we can now define $e^H$ as formal power series expansion in tensor products:
\beq
e^H := \sum_{n=0}^\infty \frac 1{n!} H^n\,.
\eeq
The so defined $e^H$ is also a morphism from ${\bf 1}\to {\bf A}$ (although in general it will not be given as a sum of pairwise interactions even if $H$ had this form).

Finally we would like to define a path integral
\beq
Z = \prod_x \int {d s_x}\, e^H\,.
\eeq
Viewed from lattice site $x$, $e^H$ is a linear combination of morphisms from ${\bf 1}\to {\bf b}_m\otimes S^m( s_x )$, with some ${\bf b}_m$. The integral $\int d s_x $ must produce out of such a term a morphism ${\bf 1}\to {\bf b}_m $, acting in a linear fashion. 
This can be codified by declaring that the `categorical integral' is a morphism (independent of $x$ if the model is translationally invariant):
\beq
\int d \ba \in {\rm Hom}\Bigl(\bigoplus_{m\ge 0} S^m(\ba)\to {\bf 1}\Bigr)\,,
\eeq
and integration in $s_x$ consists in replacing $\ba$ with $s_x$ and composing with $e^H$. This generalizes Eq.~\reef{eq:intloops}.

Now suppose we integrated over all the lattice apart from $k$ sites, obtaining the incomplete `incomplete partition function' $Z_k$ (see section \ref{sec:oper}). Generalizing Eq.~\reef{eq:ZktoCat2}, these objects are naturally elements of 
\beq
\text{Hom} \Bigl({\bf 1}\to \bigoplus_{m_1\ldots m_k } S^{m_1} (s_{x_1})\otimes\ldots \otimes S^{m_k} (s_{x_k})\Bigr)\,.
\eeq
 They will by construction satisfy consistency conditions like Eq.~\reef{eq:consist}, relating $Z_k$ to $Z_{k-1}$.

With these definitions, the fully integrated partition function is just a number. We will now define correlation functions $\< s_{x_1} \ldots s_{x_k}\>$. Take any morphism $f$ from ${\bf 1}\to s_{x_1} \otimes \ldots \otimes s_{x_k}$. With the above definition we can compute the correlation function $\< f\>$, a number, defined as an integral:
\beq
\< f\>=Z^{-1} \prod_x \int d s_x\, \Bigl(f e^H\Bigr)\,.
\eeq
This operation being linear in $f$, there exists a morphism $C$ from $ s_{x_1} \otimes \ldots \otimes s_{x_k}\to {\bf 1}$ so that
\beq
\< f \> = C\circ f\,.
\eeq
Thinking of $\< f \>$ as a `contracted correlation function', we declare this morphism $C$ as the correlation function of spin objects themselves:
\beq
C =: \< s_{x_1} \ldots s_{x_k}\> \,.
\eeq

Notice that for a group symmetry we can compute both the correlation function and contracted correlation function by doing the usual integral. The above equations in components would take the form:
\beq
C^{I_1\ldots I_k} = \< s^{I_1}_{x_1} \ldots s^{I_k}_{x_k}\>\,,\qquad \< f\> =  \< f_{I_1\ldots I_k} s^{I_1}_{x_1} \ldots s^{I_k}_{x_k}\> = f_{I_1\dots I_k} C^{I_1\ldots I_k}\,.
\eeq
For a categorical symmetry this rewriting of course does not make sense, since we cannot access the indices. The correlation function $\< s_{x_1} \ldots s_{x_k}\> $ has to be defined via the dual one given above: it would not make sense to stick the symbol $s_{x_1} \ldots s_{x_k}$ under the categorical integral sign, since only morphisms can be integrated.

\subsection{Example: $O(n)$ loop model}

The simplest example is the 2d $O(n)$ loop model, with the following textbook definition (e.g.~\cite{DiFrancesco:1997nk}, section 7.4.6). On the honeycomb 2d lattice, put a unit $N$-component spin $\vec s_x$ on every lattice site $x$, and consider the partition function
\beq
Z=\prod_x  \int   d\vec s_x \prod_{\<xy\>}(1+K \vec s_x\cdot \vec s_y)
\eeq
So far $N$ is integer and the integrals over the unit sphere are understood in the usual sense.
The integrand can be written as $e^{H}$ with $H=\sum _{\<x y\>} H_{xy}$, $H_{xy}= \log(1+K \vec s_x\cdot \vec s_y)$. This unusual form of $H_{xy}$, which leads to a simple $e^H$, is chosen so that $Z$ has a very simple expansion in powers of $K$ (that the lattice has coordination number 3 helps as well). Namely, we have
\beq
Z = \sum_{\text{loop ensembles}} N^{\#(\text{loops})} K^\text{total length}\,,
\label{eq:ZN}
\eeq
where the sum is over `loop ensembles'---collections of non-self-intersecting loops on the lattice, weighted by the shown factor depending on the number of loops and the total length. In this form partition function can be analytically continued to a non-integer value $N\to n$. 

Correlation functions $\<s^{I_1}(x_1)\ldots s^{I_k}(x_k)\>$ themselves cannot be analytically continued, but expanding them in the basis of invariant tensors $T^{I_{1}\ldots I_{k}}$, expansions coefficients can be analytically continued. In terms of loops, these analytic continuations can be understood as probabilities of configurations containing, in addition to loops, fluctuating lines pairing points $x_1$\ldots $x_k$ as prescribed by the decomposition of the chosen invariant tensors $T^{I_{1}\ldots I_{k}}$ in a product of Kronecker $\delta$'s. These are sometimes called `defect correlation functions', as in Eq.~\reef{eq:DDDD}. 
All this is standard  \cite{Nienhuis2008,Jacobsen2012}. 

The $O(n)$ loop models have been studied for many years, however their symmetry has never been clarified. Physicists speak of $O(n)$ group symmetry even for non-integer $n$ when the group does not exist, and understand it as a recipe to extract from a generally nonsensical equation a part which continues to make sense when analytically continued to non-integer $n$. Results obtained by this recipe often can be checked independently, e.g.~by direct Monte Carlo simulations of loop ensembles. This hints that there must be some truth to the recipe. But it is still mysterious why using nonsensical basic building blocks (like tensors in vector spaces of non-integer dimensions, or irreducible representations of a group which does not exist) leads to a sensible final result.

We can now for the first time explain this basic mystery: the $O(n)$ loop model is an example of a model with a categorical symmetry. We can write its partition function in a way which makes no reference to the integer $N$ case and the subsequent analytic continuation, but directly for non-integer $n$. For this we put on every lattice point an object of Deligne's category $\Reptilde\, O(n)$. We call this object $s_x$ and choose all of them isomorphic: $s_x\approx [1]$. Recall that $[1]$ is a `single point' object of the category $\Rephat\, O(n)$, which is inherited by the category $\Reptilde\, O(n)$. This simple object is the categorical analogue of the fundamental representation. The partition function is defined as an integral
\beq
Z =  \prod_x \int ds_x \prod_{\<xy\>}(1+K H_{xy})\,,
\eeq
where $H_{xy}$ is the morphism from ${\bf 1}\to s_x\otimes s_y$ given by the string diagram $\includegraphics[trim=0 0em 0 0,scale=0.4]{fig-semicirc.pdf}$. (In \reef{eq:Hxy} we have $m_1=m_2=1$). The integral over $s_x$ is a linear combination of morphisms from $S^m([1])\to 1$, as in \reef{eq:intloops}. Because the hexagonal lattice has coordination number 3, and the interaction has a particularly simple form, the only terms in the integral which enter are $m=0$ and $m=2$. Let us choose $J_0=J_2=1$, i.e.
\beq
\int d[1] = 1+\includegraphics[trim=0 1.8em 0 0,scale=0.4,angle=180]{fig-semicirc.pdf}\,.
\eeq
When we evaluate the partition function using the rules of morphism compositions, we get an expression identical to \reef{eq:ZN} analytically continued $N\to n$.

We can also consider correlation functions, e.g.~$\<s_x s_y s_z s_t\>$. Let $f$ be a general morphism from ${\bf 1}\to s_x\otimes s_y \otimes s_z\otimes s_t$, which is a linear combination
\beq
f= f_1\, \includegraphics[trim=0 0em 0 0,scale=0.65]{fig-4pt1.pdf} +
f_2\,\includegraphics[trim=0 0em 0 0,scale=0.65]{fig-4pt2.pdf} +
f_3\, \includegraphics[trim=0 0em 0 0,scale=0.65]{fig-4pt3.pdf}\,. 
\label{eq:fmorph}
\eeq
With the above definitions we can compute the ``average'' $\< f \>$ as
\beq
\< f \> =Z^{-1} \prod_x \int ds_x\, f\, \prod_{\<xy\>}(1+K H_{xy})\,.
\eeq
Correlation function $C=\<s_x s_y s_z s_t\>$ is defined as a morphism from $s_x\otimes s_y \otimes s_z\otimes s_t \to {\bf 1}$ so that $\<f\>= f \circ C$. Let us expand $C$ as 
\beq
C = C_1\, \includegraphics[trim=0 1em 0 0,scale=0.65,angle=180]{fig-4pt1.pdf} +
C_2\,\includegraphics[trim=0 1em 0 0,scale=0.65,angle=180]{fig-4pt2.pdf} +
C_3\, \includegraphics[trim=0 1em 0 0,scale=0.65,angle=180]{fig-4pt3.pdf}\,. 
\eeq
A moment's thought shows that $C_i$ can be identified as probabilities of configurations containing lines joining points $x_i$. So the new definition is identical to the old definition of the defect correlation function. 

As a final remark, we note that the $O(n)$ loop model can be studied in any number of dimensions, not just in $d=2$. Categorical symmetry is valid in any $d$. In $d=3$,  the $O(n)$ loop model should have a critical point for any real $n$, and not just for $-2\le n\le 2$ as in 2d. The 3d case remains rather poorly studied compared to 2d, and it would be interesting to perform Monte Carlo simulations of $d=3$ models with non-integer $n$ to detect these critical points.\footnote{See \cite{Nahum:2014jaa} for interesting work about different 3d loop models.}

\subsection{Wilsonian renormalization}
\label{sec:RG}

The presence of symmetries in theories is particularly powerful when combined with the renormalization group (RG). Since group symmetries are preserved under RG, we can consider effective descriptions at long distances which realize the same symmetry as the microscopic Hamiltonian. This effective description may be a CFT if the theory is critical, or a theory of Goldstone bosons if one is in the phase of spontaneously broken continuous global symmetry. 

We will now explain that categorical symmetries are, like group symmetries, preserved under RG flows.
One application of this result is that the parameter $n$ characterizing the categorical symmetry does not renormalize: $n_{\text{IR}}=n_{\text{UV}}$. For group symmetries, one could argue by saying that $N$ cannot perform a discrete jump under a continuous RG transformation, when integrating out an infinitesimal momentum shell. Although this argument breaks down for categorical symmetries, we will show that nevertheless $n$ remains invariant under RG flows. This excludes a scenario where $n$ gets renormalized with integer values of $n$ being fixed points.\footnote{\label{note:paradox}While so far we are focussing on the $O(n)$ symmetry, the non-renormalization conclusion equally holds for other categorical symmetries defined below (section \ref{sec:other}), e.g. to the symmetry $S_q$ relevant for the Potts model. We are aware of one seemingly conflicting statement in the literature: Eq.~(10) of Ref.~\cite{LykkeJacobsen:1998olx} states that the parameter $q$ of the random-bond Potts model gets renormalized. It would be interesting to understand how this apparent contradiction gets resolved.}

Consider first RG-invariance of a group symmetry, for a system of spins $s$ sitting at sites of a lattice, with a Hamiltonian $H(s)$. A Wilson-Kadanoff RG transformation consists in splitting the lattice into cells and associating with each cell a spin $s'$, subject to a condition 
\beq
\sum_{\{s'\}} P(s',s)=1\,,
\label{eq:weight}
\eeq
where $P(s',s)$ is a weight factor, depending on site and cell spin configurations $\{s\}$ and $\{s'\}$.
Different choices of the weight factor $P(s',s)$ produce different RG transformations.
For some RG transformation the cell configuration $\{s'\}$ may be uniquely determined by the site configuration $\{s\}$ (as for the majority rule on the triangular lattice, or the decimation rule), in which case $P(s',s)$ is 1 for this configuration and zero otherwise. In more general situations cell configurations $\{s'\}$ are chosen with some probability depending on $\{s\}$, and $P(s',s)$ gives this probability distribution.
We then define the renormalized Hamiltonian $H(s')$ for the cell spin system by \cite{Niemeijer}
\beq
\exp H'(s')= \sum_{\{s\}} P(s',s) \exp H(s)\,.
\eeq
By \reef{eq:weight}, the partition function remains invariant: $\sum_{\{s'\}}\exp H'(s')=\sum_{\{s\}}\exp H(s)$.

Suppose now we have a group symmetry: a global symmetry group $G$ acting on spins $s$ such that the Hamiltonian $H(s)$ is invariant. We then impose an additional requirement that the weight factor $P(s',s)$ be invariant when $G$ acts simultaneously on $s$ and $s'$. Under this requirement, the renormalized Hamiltonian $H(s')$ will be invariant. This is what we mean by RG-invariance of a group symmetry.\footnote{One sometimes asks: what if we choose the weight factor $P(s',s)$ which is not $G$-invariant? Well, then $H'(s')$ will not be $G$-invariant either. Thus, a silly choice of RG transformation may hide manifest $G$-invariance of the lattice model. This is not surprising since any symmetry may be obscured by a bad choice of variables.} 

Let us now adapt the above argument to categorical symmetries. The spin configurations $\{s\}$ and $\{s'\}$ now consist of objects of Deligne category. The sum $\sum_{\{s\}}$ has to be replaced by the categorical integral
$
\prod_x \int ds_x\,,
$
and similarly for $s'$. The weight factor $P(s',s)$ has to be a morphism from 1 to tensor products of $s$ and $s'$: this condition replaces the $G$-invariance of $P(s',s)$ in the group symmetry case. The RG transformed Hamiltonian $H'(s')$ is, then, a morphism of the same Deligne category as the original one $H(s)$. Categorical symmetry is thus preserved under RG in exactly the same fashion as regular group symmetry.

 \section{QFTs with categorical symmetry}
 \label{sec:QFT}
 \subsection{Basic axioms}
 
 In the previous section we have focused intentionally on lattice models, to make the point that the categorical symmetry is a non-perturbatively meaningful concept. We will now discuss categorical symmetry in the context of continuum limit quantum field theories (QFTs). One way to obtain such QFTs would be to consider the aforementioned lattice models at or close to their critical points.
 
 Let $\fcy C$ be a braided tensor category (see appendix \ref{sec: category}). We will say that a (Euclidean) QFT has symmetry $\calC$ if: 
 \begin{enumerate} 
 	\item 
 	Local operators $\phi(x)$ are classified by objects ${\bf a}\in\fcy C$. We will abuse group symmetry terminology and say that `$\phi(x)$ transforms as ${\ba}$'. We denote this by object isomorphism $\phi(x)\approx\ba$, or more verbosely as $\phi({\bf a},x)$.\footnote{{\bf Note added in February 2023:} This terminology is the reverse of the one used for group theory symmetries. To wit, in a normal theory if $\phi_i(x)$ is a vector (say in $U(N)$), then we can introduce a polarization vector $v^i$ which is a covector, and work with the scalar operator $\phi(v,x)$. If we wanted to match this language, the $\ba$ should have been associated with the covector representation (that is, the representation of $v$) rather than the vector representation. We thank Cagin Yunus for drawing our attention to this terminological clash. When dealing with categorical symmetries, the only meaning of the phrase `$\phi(x)$ transforms as ${\ba}$' is that operator $\phi(\ba,x)$ is associated to an object $\ba$. The fact that in normal group theory this means that $\phi$ itself would transform in $\bar{\ba}$ is not actually relevant. So although the terminology is a bit unfortunate, nothing we say below is wrong. Also, this issue is less bothersome for the $O(N)$ symmetry where there is a natural identification of $\ba$ and $\bar{\ba}$.} Some operators, such as the identity operator and the stress--tensor, may transform trivially, in which case they are associated with the object ${\bf 1}$.
 	\item Correlation functions are morphisms in $\fcy C$. More specifically, we have
 	\begin{equation}\langle \phi_1({\bf a_1},x_1)\phi_2({\bf a_2},x_2)...\phi_n({\bf a_n},x_n)\rangle \in \text{Hom}({\bf a_1}\otimes{\bf a_2}\otimes...\otimes{\bf a_n}\rightarrow {\bf 1}).\end{equation}
 	\item Different orderings of operators in a correlator are related to each other through braiding. For instance, 
 	\begin{equation*}\langle \phi_2({\bf a_2},x_2)\phi_1({\bf a_1},x_1)\phi_3({\bf a_3},x_3)\dots\rangle  = \langle \phi_1({\bf a_1},x_1)\phi_2({\bf a_2},x_2)\phi_3({\bf a_3},x_3)\dots\rangle\circ\beta_{{\bf a}_2,{\bf a}_1} \end{equation*}
 	where $\beta_{{\bf a}_2,{\bf a}_1}$ is the braiding which maps ${\bf a}_2\otimes{\bf a}_1\rightarrow{\bf a}_1\otimes{\bf a}_2$.
 \end{enumerate}
 
 Some comments are in order for the third axiom. If we swap the order of operators twice, we return to the original correlator:
 \begin{equation}
 \langle \phi_1({\bf a_1},x_1)\phi_2({\bf a_2},x_2)\phi_3({\bf a_3},x_3)\dots\rangle = \langle \phi_1({\bf a_1},x_1)\phi_2({\bf a_2},x_2)\phi_3({\bf a_3},x_3)\dots\rangle\circ\beta_{{\bf a}_1,{\bf a}_2}\circ\beta_{{\bf a}_2,{\bf a}_1}.
 \end{equation}
 For this reason it is natural to restrict our attention to symmetric tensor categories, for which braiding twice always gives the identity:
 \begin{equation}
 \beta_{{\bf a}_1,{\bf a}_2}\circ\beta_{{\bf a}_2,{\bf a}_1} = \text{id}_{{\bf a}_1}\otimes\text{id}_{{\bf a}_2}.
 \end{equation}
Other branches of physics allow categories with more general braidings, such as when describing anyonic statistics and also in 2d CFTs. We will however restrict our discussion here to the symmetric case.

Although our axioms so far allow operators to transform as any object $\ba\in\fcy C$, given any object $\phi(\ba,x)$ and morphism $\iota:{\bf b}\rightarrow\ba$ embedding a simple object $\bf b$ in $\ba$, we can consider the operator
\beq
\tilde \phi({\bf b},x) \equiv \phi(\ba,x)\circ \iota\,,
\label{eq:subop}
\eeq whose correlation functions are defined as
\begin{equation}
\langle \tilde\phi({\bf b},x)\phi_1(\ba_1,x_1)...\phi_n(\ba_n,x_n)\rangle  \equiv \langle \phi(\ba,x)\phi_1(\ba_1,x_1)...\phi_n(\ba_n,x_n)\rangle\circ (\iota\otimes \text{id}_{\ba_1\otimes...\otimes\ba_n}).\end{equation}
So without loss of generality we can restrict to operators transforming as simple objects, just as for group symmetries operators transform in irreducible representations of the symmetry group.
 
%
Many notions in quantum field theory naturally extend to this more general setting. A simple example is the operator product expansion. Given operators $\phi_1(\ba_1,x_1)$ and $\phi_2(\ba_2,x_2)$ we expand
\begin{equation}\label{eq:OPE}
\phi_1(\ba_1,x)\phi_2(\ba_2,0) \longrightarrow \sum_{{\bf b}\in \ba_1\otimes\ba_2} \sum_k  \fcy O_k({\bf b},0)\circ C^{\fcy O_k}_{\phi_1\phi_2}(x)
\end{equation}
where $C^{\fcy O_k}_{\phi_1\phi_2}(x) \in \Hom(\ba_1\otimes\ba_2\rightarrow{\bf b})$. This is interpreted to mean that
\begin{equation}
\langle \phi_1(\ba_1,x_1)\phi_2(\ba_2,x_2)\dots\phi_n(\ba_n,x_n)\rangle \underset{x_1\rightarrow x_2}\longrightarrow \sum_{{\bf b}\in \ba_1\otimes\ba_2} \sum_k \langle\fcy O_k({\bf b},x_2)\dots\phi_n(\ba_n,x_n)\rangle \circ C^{\fcy O_k}_{\phi_1\phi_2}(x_1-x_2)\,.
\end{equation} 
 
Although so far we have described correlation functions, we can also think of states as being associated with objects of $\fcy C$. Given a bra $\bra{A,\bf a}$ and a ket $\ket{B,\bf b}$ the inner product is a morphism
\begin{equation}
\langle A,{\bf a}|B,{\bf b}\rangle \in\Hom(\ba\otimes{\bf b}\rightarrow{\bf 1}).
\end{equation}
 
{Correlation functions for a theory with categorical symmetry are morphisms in a category. To extract actually numbers from the model, quantities that could be measured or computed in a Monte Carlo simulation, we can contract correlators with morphisms, just as we did in \eqref{eq:measD} for the $O(n)$ case. For instance, given 
\begin{equation}\langle \phi_1({\bf a_1},x_1)\phi_2({\bf a_2},x_2)...\phi_n({\bf a_n},x_n)\rangle \in \text{Hom}({\bf a_1}\otimes{\bf a_2}\otimes...\otimes{\bf a_n}\rightarrow {\bf 1})\,,\end{equation}
we can compute for any ${f:{\bf 1}\rightarrow {\bf a_1}\otimes{\bf a_2}\otimes...\otimes{\bf a_n}}$ the quantity
\begin{equation}
\langle \phi_1({\bf a_1},x_1)\phi_2({\bf a_2},x_2)...\phi_n({\bf a_n},x_n)\rangle\circ f \in \mathbb C\,.
\end{equation}
As a consequence of {proposition \ref{pr:ndeg}}, if we compute this quantity for every such $f$, this suffices to reconstruct $\langle \phi_1({\bf a_1},x_1)\phi_2({\bf a_2},x_2)...\phi_n({\bf a_n},x_n)\rangle$. 

}

 \subsection{Example: the continuum free scalar $O(n)$ model}
 The theory of $N$ free scalars $\phi^1,....,\phi^N$ is $O(N)$ symmetric, with the field $\phi^I$ transforming in the fundamental representation of $O(N)$. We can generalize this to a free theory with a categorical symmetry $\Reptilde\,O(n)$ for any value of $n$. To do so we begin with a scalar operator $\phi(x)\approx\bf n.$\footnote{This is the object we so far denoted by $[1]$, the analogue of the fundamental representation.} We take the two-point function to be
 \begin{equation}
 \langle \phi(x_1)\phi(x_2) \rangle = \frac {\delta^{\bf n,\bf n}}{|x_1-x_2|^{2\Delta}}\,,\qquad \Delta = \frac{d-2}2\,,
 \end{equation}
 where $d$ is the spacetime dimension. Here $\delta^{\bf n,\bf n}$ is the analogue of the delta-tensor, the morphism from $\bf n\otimes \bf n\to {\bf 1}$ represented by the string diagram $\rotatebox{180}{\includegraphics[trim=0 1em 0 0,scale=0.4]{fig-semicirc.pdf}}$.\footnote{More abstractly, for any object $\ba$ one can define a dual object $\overline{\ba}$ and two morphisms $\delta_{\bf a,\overline{\bf a}}$, $\delta^{\overline{\bf a},\bf a}$
with some natural properties, making $\Reptilde\,O(n)$ a \emph{rigid} category, appendix \ref{sec:RIG}. The object $\bf n$ is self-dual.} 
 We then compute higher-point functions through Wick contractions. 
 For example, the four-point function is
 \begin{equation}\label{eq:4ptAbs}
 \langle\phi(x_1)\phi(x_2)\phi(x_3)\phi(x_4)\rangle = \frac{\delta^{\bf n,\bf n}\otimes\delta^{\bf n,\bf n}}{|x_1-x_2|^{2\Delta}|x_3-x_4|^{2\Delta}} + \frac{(\delta^{\bf n,n}\otimes\delta^{\bf n,\bf n})\circ\sigma_{(14)}}{|x_1-x_3|^{2\Delta}|x_2-x_4|^{2\Delta}} + \frac{(\delta^{\bf n,\bf n}\otimes\delta^{\bf n,\bf n})\circ\sigma_{(13)}}{|x_1-x_4|^{2\Delta}|x_2-x_3|^{2\Delta}}\,,
 \end{equation}
 where $\sigma_{(ij)}$ is the morphism ${\bf n}^{\otimes 4}\ra{\bf n}^{\otimes 4}$ which interchanges the $i^{\text{th}}$ and $j^{\text{th}}$ copy of $\bf n$. Below we will see how to obtain these correlation functions from the properly defined categorical path integral. 
 
 We can construct fields isomorphic to other objects in $\Reptilde\,O(n)$ as composites of $\phi(x)$. For instance, the singlet $\phi^2(x)\approx {\bf 1}$ is defined as the composite operator
 \begin{equation}
 \phi^2(x) = \lim_{y\ra0}\  \frac 1 n (\phi(x+y)\phi(x))\circ \delta_{\bf n,\bf n} - \frac 1 {|y|^{2\Delta}}\,,
 \end{equation}
 and the symmetric tensor $T(x)$ as 
 \begin{equation}
 T(x) =\left(\lim_{y\ra0}\  \phi(x+y)\phi(x)\right)\circ P_{\bf S}^{{\bf n},{\bf n}}\,,
 \end{equation}
 where $P^{\bf S}_{{\bf n},{\bf n}} : \bf S\rightarrow{\bf n}\otimes{\bf n}$ is the morphisms embedding the symmetric traceless representation $\bf S$ into ${\bf n}^{\otimes 2}$.

\subsection{Path integrals and perturbation theory}
In section \ref{sec:lattice} we have seen that lattice models with categorical symmetries can be described by a formal integration procedure. Given some interaction between objects living on the lattice, we could then integrate over the lattice variables to compute physical observables. By taking an increasingly fine lattice, we can define a continuum path-integral. 

Recall that for lattice models there was considerable freedom in defining the categorical integral, parametrized by arbitrary coefficients $\calJ_m$ in \reef{eq:intloops}. Physically, this freedom comes from the fact that the radial distribution of lattice variables may be arbitrary. As an example in the usual $O(N)$ model we can impose the condition that spins are confined to the unit sphere $s_x\cdot s_x = 1$, or any other condition. 

When defining the continuum limit Gaussian path integral, this ambiguity will be fixed by imposing the constraints that the basic integral $\int d\phi$ should be translationally invariant:
\begin{equation}\label{eq:varshift}
\int d\phi\, f(\phi) = \int d\phi\, f(\phi+\varphi)\,,
\end{equation}
and also scale invariant
\begin{equation}\label{eq:rescale}
\int d\phi\, f(\lambda\phi) = \lambda^{-n}\int d\phi\, f(\phi)\,.
\end{equation}
Constructing such an integral is the same problem as the one considered in dimensional regularization.
As established there \cite{Wilson:1972cf,Collins:1984xc}, requirements \reef{eq:varshift}, \reef{eq:rescale} and linearity fix the integral uniquely up to normalization. 

Along with this similarity, we would like to stress a difference of principle: in dimensional regularization studies, one sometimes \cite{Wilson:1972cf,Collins:1984xc} thinks of a vector with non-integer number of components as one with infinitely many components. This is not the point of view we are developing here. For us a vector with a non-integer number of components is an object in Deligne category. So, $f(\phi)$ in \reef{eq:varshift} should not be thought as a function but as a morphism in $\text{Hom}\left({\bf 1}\to \bigoplus_n  S^n(\phi)\right)$. When we substitute $\phi \to \phi+\varphi$ we produce by algebraic rules another morphism in $\text{Hom}\left({\bf 1}\to \bigoplus_{n,m}  S^n(\phi) \otimes S^m(\varphi)\right)$. The integral in \reef{eq:varshift} is a linear operation which maps morphisms to morphisms.

 We will use the usual normalization
\begin{equation}
\int d\phi\, \exp\left[-\frac 12(\phi\cdot\phi)\right] =  ({2\pi})^{n/2},
\end{equation}
which when $n$ is integer reduces to the standard integral formula. Here $\phi\cdot\phi$ is a notation for the morphism from ${\mathbf 1} \to \phi \otimes \phi$ corresponding to the string diagram \includegraphics[trim=0 0em 0 0,scale=0.4]{fig-semicirc.pdf}.

Then by the usual manipulations one derives the more general Gaussian integral
\begin{equation}\label{eq:genNInt}
\int d\phi_1\ldots d\phi_k\, \exp\left[-\frac 12 \sum_{i,j} A_{ij}(\phi_i\cdot\phi_j) + \sum_i J_i\cdot\phi_i\right] = \left(\frac{(2\pi)^k}{\det(A)}\right)^{n/2}\exp\left(\frac 12 \sum_{i,j} (A^{-1})_{ij}J_i\cdot J_j\right)\,,
\end{equation}
where $\phi_i$, $J_i $ are objects isomorphic to $\bf n$, and $A$ is a numerical matrix.\footnote{The reader may notice that Eq.~\reef{eq:genNInt} reduces to the Grassmann integration formula for $n = -1$. We will return to this relationship in section \ref{sec:SPN}.}

Taking the continuum limit, we are led to define the free $O(n)$ model by a Gaussian path integral:
\begin{multline}
Z[J] = \int \fcy D\phi\, \exp\left[-\int d^dx\ (\frac 12 \partial^\mu\phi(x)\cdot\partial_\mu\phi(x) + J(x)\cdot\phi(x))\right] \\= Z[0] \exp\left(\frac 12 \int d^dx\,d^dy\, \frac{C_d }{|x-y|^{d-2}}\, J(x)\cdot J(y)\right)\,,\qquad C_d=\frac{\Gamma(d/2-1)}{4\pi^{d/2}}\,.
\end{multline}
Formally taking derivatives with respect to the source $J(x)$ allows us to compute $\phi$ correlators, and these match the ones given in the previous section (up to rescaling of $\phi$). 
 
Having defined a free $O(n)$ path-integral, we can now introduce interaction terms, for instance a $(\phi\cdot\phi)^2$ interaction
\begin{equation}
Z_\lambda[J] = \int \fcy D\phi\, \exp\left[-\int d^dx\ \left(\partial^\mu\phi(x)\cdot\partial_\mu\phi(x) + \lambda (\phi\cdot\phi)^2+ J(x)\cdot\phi(x)\right)\right]\,.
\end{equation}
We can compute correlation functions order by order in $\lambda$ using the usual Feynman diagram expansion. Any Feynman diagram is a product of a spacetime dependent, $n$ independent quantity, and of an internal index dependence. The spacetime dependence is evaluated as usual, while the internal index dependence can be encoded by an $\Reptilde\,O(n)$ string diagram, by making the replacement of each quartic vertex as
\begin{equation}
\includegraphics[trim=0 0em 0 0,scale=0.5]{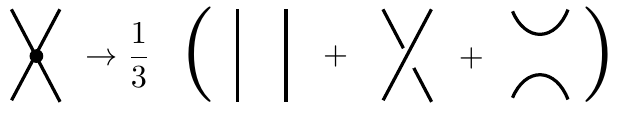}\,.
\end{equation}
One can then discuss renormalization, IR fixed points (such as Wilson-Fisher), etc., uniformly for non-integer and integer $n$.

Simplifying string diagrams, polynomial $n$-dependent factors will arise, which will be of course the same as in the usual intuitive approach to perturbation theory at non-integer $n$. The added value of the categorical construction comes mostly from the enhanced understanding of what exactly is being computed via formal manipulations at non-integer $n$. It also provides the appropriate language for describing non-perturbative completions of the perturbative calculations.

\subsection{Conserved currents}
\label{sec:current}
In a quantum field theory with a continuous global symmetry $G$, we typically expect there to exist a conserved current $J^\mu_a$ transforming in the adjoint representation $\mathfrak g\in\Rep\,G$. When the theory can be described by a local Lagrangian this statement is known as Noether's theorem. In this section we shall generalize Noether's theorem to theories with categorical symmetries. 

To begin we must consider what distinguishes representation categories of Lie groups from those of more general groups. An important feature of Lie representation theory is the existence of an adjoint representation $\mathfrak g\in\Rep\,G$. For any other representation $\ba\in\Rep\,G$ we have a (potentially trivial) action of $\mathfrak g$ on $\ba$,  Concretely, we could describe this action using the group generators $(T_i)^a_{\ b}$, where $i$ is an adjoint index and $a$ and $b$ are $\ba$ indices. More abstractly, we can describe the group generators as a morphism ${\tau_\ba:\mathfrak g\otimes\ba\rightarrow\ba}$.

The morphisms $\tau_\ba$ have a number of special properties. The first is \emph{naturality}, which states that for every morphism $f:\ba\rightarrow{\bf b}$ in $\Rep\,G$, the following diagram commutes:
\begin{equation}\label{eq:ctcon1}
\begin{tikzcd}
\mathfrak g\otimes\ba \arrow[r,"\text{id}_{\mathfrak g}\otimes f"] \arrow[d,"\tau_\ba"]
& \mathfrak g\otimes{\bf b}  \arrow[d, "\tau_{\bf b}"] \\
\ba \arrow[r, "f"]
& {\bf b}
\end{tikzcd}\end{equation}
In plain language this can be stated as the fact that morphisms $f:\ba\rightarrow{\bf b}$ commute with the action of Lie algebra generators, i.e.~they are invariant tensors.
 
The second special property is that the generators in $\ba\otimes{\bf b}$ are essentially the sum of the generators in $\ba$ and ${\bf b}$. In the abstract language this condition becomes 
\begin{equation}\label{eq:ctcon2}
\tau_{\bf a\otimes\bf b} = \tau_{\bf a}\otimes\text{id}_{\bf b} + (\text{id}_{\bf a}\otimes\tau_{\bf b})\circ(\beta_{\mathfrak g,\bf a}\otimes\text{id}_{\bf b})\,.
\end{equation}
(The second factor in the second term just reorders the tensor product from $\mathfrak g\otimes \bf b\otimes\bf c$ to $\bf b\otimes \mathfrak g\otimes\bf c$ so that $\mathfrak g$ stands next to $\bf c$ and can act on it.)
The final special property is that $\tau_{\mathfrak g}$ is antisymmetric:
\begin{equation}\label{eq:ctcon3}
\tau_{\mathfrak g}\circ\beta_{\mathfrak g,\mathfrak g} = -\tau_{\mathfrak g}\,.
\end{equation}

Although these conditions may look abstract, they have a simple diagrammatic description. We will denote the adjoint $\mathfrak g$ by a dashed line, and the morphism $\tau_\ba$ as a trivalent vertex:
\begin{equation}
\includegraphics[trim = 0 0 0 0, scale=0.5]{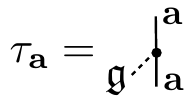}\,.
\end{equation}
The conditions \eqref{eq:ctcon1} and \eqref{eq:ctcon2} translate to:
\begin{equation}
\includegraphics[trim = 0 0 0 0, scale=0.5]{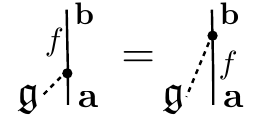}\,\qquad \raisebox{1.4em}{\text{and}}\qquad
\includegraphics[trim = 0 0 0 0, scale=0.5]{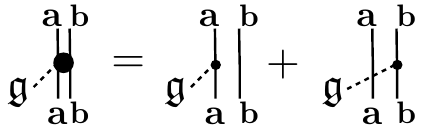}\,.
\end{equation}
In particular, if we apply these to $\tau_{\mathfrak g}:\mathfrak g\otimes\mathfrak g\rightarrow\mathfrak g$ we find that
\begin{equation}
\includegraphics[trim = 0 0 0 0, scale=0.5]{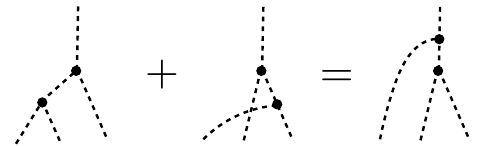}\,,
\end{equation}
which is the Jacobi identity. Combining this with \eqref{eq:ctcon3} we see that $\tau_{\mathfrak g}$ satisfies the usual axioms of a Lie bracket.\footnote{In category theory parlance $\mathfrak g$ is a Lie algebra \emph{internal} to the category $\fcy C$. This is another example of internalization, as discussed in footnote \ref{foot:int}. A discussion specific to the Lie algebraic case can be found on page 292 of \cite{etingof2016tensor}.
}

By analogy to the group theoretic case, we shall define an \emph{adjoint}\footnote{This concept has no relation to the notion of an adjoint functor.} in a symmetric tensor category $\fcy C$ to be an object $\mathfrak g\in\fcy C$ along with a morphisms $\tau_\ba:\mathfrak g\otimes\ba\rightarrow\ba$ for each $\ba\in\fcy C$ such that the conditions \eqref{eq:ctcon1}, \eqref{eq:ctcon2} and \eqref{eq:ctcon3} are satisfied, along with the technical condition:
\begin{equation}\label{eq:adjCon}
\text{If a morphism }f:\mathfrak g\rightarrow\mathfrak g \text{ satisfies } \tau_{\ba}\circ(f\otimes\text{id}_{\ba}) = 0 \text{ for every }\ba\in\fcy C\,,\text{ then } f = 0\,.
\end{equation} 
This condition rules out such trivial cases as allowing $\tau_\ba = 0$ for every $\ba$, or starting with an adjoint $(\mathfrak g,\tau)$ and constructing for any object $\ba$ a new adjoint $(\mathfrak g\oplus\ba,\tau \oplus 0)$. 
A categorical symmetry with an adjoint is called \emph{continuous}\footnote{Our terminology is natural from a physicists perspective. It is not related to a previously existing notion of a continuous category as the categorification of the notion of a continuous poset.}, otherwise it is \emph{discrete}.

We should note that there may be many adjoint objects in a continuous tensor category. In appendix \ref{sec:CONTCATS} however we show under some broad finiteness conditions that there are only a finite number of adjoint objects in a tensor category. In this case we prove there is a unique (up to unique isomorphism) \emph{maximal adjoint} $(\mathfrak m, T)$. This has the property that for any other adjoint $(\mathfrak g,\tau)$ there exists a unique morphism
\begin{equation}
\iota_{\mathfrak g}:\mathfrak g\rightarrow\mathfrak m \quad \text{with}\quad \tau_\ba = T_\ba\circ(\iota_{\mathfrak g}\otimes\text{id}_\ba)\,.
\end{equation}
As an example, consider a semisimple Lie group ${G\approx G_1\times\dots\times G_n}$ where each $G_i$ is a simple group. There is then an adjoint object $\mathfrak g_i\in\Rep\,G$ for each subgroup $G_i$, while the maximal adjoint object is their direct sum: $\mathfrak g \approx \mathfrak g_1\oplus\dots\oplus\mathfrak g_n$. 

In the category $\Rep\,O(N)$ the antisymmetric tensor $\bf A$ in ${\bf N}\otimes{\bf N}$ is an adjoint object. This extends to $\Reptilde\,O(n)$. We can diagrammatically represent $\tau_{\bf n}$ as 
\begin{equation}
\includegraphics[scale=0.45]{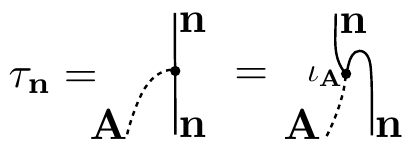}\,,
\end{equation}
where $\iota_{\bf A}$ is the morphism embedding $\bf A$ into ${\bf n}\otimes {\bf n}$.
For any simple object $\ba\in{\bf n}^{\otimes k}$ we can compute $\tau_{\ba}$ using:
\begin{equation}
\includegraphics[scale=0.5]{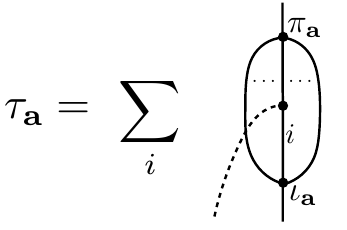}\,,
\end{equation}
where $\iota_{\ba},\pi_{\ba}$ are the embedding and projection morphism into/from ${\bf n}^{\otimes k}$ (this follows from properties \eqref{eq:ctcon1} and \eqref{eq:ctcon2}). This objects is actually the only adjoint object, and hence the maximal adjoint object, in $\Reptilde\,O(n)$, as we show in appendix \ref{sec:CONTCATS}.

An example of a discrete categorical symmetry is $\Reptilde\,S_n$, relevant for the analytic continuation of the Potts model (see section \ref{sec:other}).

We will find it useful to define the family of morphisms $\overline\tau_\ba:\ba\rightarrow\overline{\mathfrak{g}}\otimes\ba$ by:
\begin{equation}
 \includegraphics[scale=0.45]{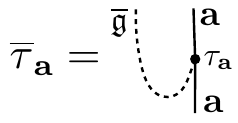} \raisebox{1.2em}{ $=(\text{id}_{\overline{\mathfrak g}} \otimes \tau_\ba)\circ(\delta_{\overline{\mathfrak g},\mathfrak g}\otimes\text{id}_\ba) \,.$}
\end{equation}
It is then not hard to check that $\overline{\mathfrak g}$ and $\overline\tau_\ba$ satisfy \eqref{eq:ctcon1}, \eqref{eq:ctcon2} and \eqref{eq:ctcon3}, but with all of the arrows and function compositions reversed. For this reason we can term $\overline{\mathfrak g}$ to be a \emph{coadjoint} object. For simple Lie algebras $\mathfrak g$ and $\overline{\mathfrak g}$ are isomorphic, but this is not true in general.

We can now extend the Noether argument to theories with a continuous categorical symmetry $\fcy C$. Consider a quantum field theory described by a Lagrangian $\fcy L[\phi_i]$, where the fields $\phi_i(x)\approx\ba_i$ for simple objects $\ba_i\in\fcy C$. Recall that in the categorical setting the Lagrangian is not a functional of the fields, but a morphism from ${\bf 1}$ to some tensor product of $\phi_i$'s. That this morphism originates in ${\bf 1}$ is the categorical version of the usual condition that Lagrangian must be a global symmetry singlet. Now consider varying each $\phi_i(x)$ by
\begin{equation}\label{eq:symvar}
\phi_i(x) \rightarrow \phi_i(x)+ (\alpha(x)\phi_i(x))\circ\overline\tau_{\ba_i}
\end{equation}
for some arbitrary infinitesimal field $\alpha(x)\approx\overline{\mathfrak g}$, where $(\mathfrak g,\tau)$ is an adjoint in $\fcy C$. Using \eqref{eq:ctcon1} and \eqref{eq:ctcon2} we find that for constant $\alpha$ the variation of the Lagrangian vanishes:
\begin{equation}
\fcy L[\phi_i(x)+ (\alpha\,\phi_i(x))\circ\overline\tau_{\ba_i}] = \fcy L[\phi_i(x)]\,.
\end{equation}
Because of this only terms proportional to $\partial_\mu\alpha$ can exist when we compute a more general variation:\footnote{For simplicity we assume that the Lagrangian involves only up to first order derivatives of $\phi_i(x)$.}
\begin{equation}
\fcy L[\phi_i]\rightarrow\fcy L[\phi_i] + \left[(\partial_\mu\alpha(x)\,\phi_i(x))\circ\overline\tau_{\ba_i})\otimes\frac{\partial\fcy L}{\partial(\partial_\mu\phi_i)}\right]\circ \delta_{\ba_i,\overline\ba_i}\,.
\end{equation}
Using standard path-integral manipulations, we conclude that
\begin{equation}
J^\mu(x) = \phi_i(x)\frac{\partial\fcy L}{\partial(\partial_\mu\phi_i)}\circ(\tau_{\ba_i}\otimes\text{id}_{\overline\ba_i})\circ(\text{id}_{\mathfrak g}\otimes\delta_{\ba_i,\overline\ba_i})
\end{equation}
is a conserved current transforming in the adjoint $\mathfrak g$. The operator $\partial_\mu J^\mu(x)$ acts on a local operator $\phi(x)\approx\ba$ as
\begin{equation}
\partial_\mu J^\mu(x)\,\phi(y) = -i\delta(x-y)\phi(y)\circ\tau_\ba\,.
\end{equation}

For $N>1$ the free $O(N)$ model has a conserved current
\begin{equation}
J^\mu_{ij}(x) = \phi_i(x)\partial^\mu\phi_j(x) - \phi_j(x)\partial^\mu\phi_i(x)
\end{equation}
Extending this operator to $\Reptilde\,O(n)$ is straightforward:
\begin{equation}\label{eq:consCur}
J^\mu(x) = \left(\lim_{y\ra0}\ \phi(x+y)\partial^\mu\phi(x)-\phi(x)\partial^\mu\phi(x+y)\right)\circ P_{\bf A}^{{\bf n},{\bf n}}\,,
\end{equation}
where $P_{\bf A}^{{\bf n},{\bf n}} : {\bf A}\rightarrow{\bf n}\otimes{\bf n} $ is the morphisms embedding the antisymmetric representation into ${\bf n}^{\otimes 2}$. Noether's theorem allows us to conclude that this current remains conserved in interacting $O(n)$ models. 

We have proved categorical Noether's theorem for Lagrangian theories, but it should remain true also for continuum limits of lattice models. For example continuous phase transitions of the loop $O(n)$ models should possess a conserved current operator of canonical scaling dimension $d-1$, transforming in the adjoint. This can be tested for the 2d $O(n)$ models with $n\in[-2,2]$, whose spectrum and multiplicity of states are exactly known from the Coulomb gas techniques \cite{diFrancesco:1987qf}. From there one can indeed see the presence of a spin 1, dimension 1 state of multiplicity $n(n-1)/2$ equal to the dimension of the adjoint, which can be identified with the conserved current. We will discuss this model further in section \ref{sec:pasha}.

\subsection{Explicit symmetry breaking}
In the previous section we described how to construct perturbative $O(n)$ models, by deforming the path-integral by $O(n)$ singlets. We often however wish to study models with less symmetry. For integer $N$ we can construct such models by perturbing with an operator which is not an $O(N)$ singlet. We wish to generalize this construction to categorical symmetries.
 
Consider a group $G$ with a subgroup $H$ embedded into $G$ via the group homomorphism $f:H\rightarrow G$. We can use this to define a functor $\fcy F_f : \Rep\,G\rightarrow\Rep\,H$ between the representation categories. Under this functor a representation $\rho: G \rightarrow GL(\mathbb C^k)$ of $G$ is mapped to the restricted representation of $H$, given by $\fcy F_f(\rho) = \rho\circ f: H\rightarrow GL(\mathbb C^k)$. 
 
Given a QFT with symmetry group $G$, we can use the functor $\fcy F_f$ to rewrite irreducible $G$-invariant operators as the sum of irreducible $H$-invariant operators. Some operators, while transforming non-trivially under $G$, will contain $H$ singlets. Perturbing by these singlets will then break the group $G$ down to the subgroup $H$. 
 
Let us now generalize this to QFTs with categorical symmetries. Given a theory with a symmetry $\fcy C$, the analogue of a ``subgroup'' is a category $\fcy D$ along with a functor $\fcy F:\fcy C\rightarrow\fcy D$. We can use the functor to rewrite our QFT as a theory with categorical symmetry $\fcy D$. Operators $\phi_i(x)$ in the QFT are translated to $\fcy F(\phi_i)(x)$, with correlation functions given by:
\begin{equation}\label{eq:brSymCor}
\langle \fcy F(\phi_1)(x_1)\dots\fcy F(\phi_n)(x_n) \rangle = \fcy F\left(\langle \phi_1(x_1)\dots\phi_n(x_n)\rangle \right)\,.
\end{equation}
In general $\fcy D$ will have a larger space of morphisms $\fcy F(\phi_1)\otimes...\otimes \fcy F(\phi_n)\rightarrow\bf 1$ than $\fcy C$, but since we began with a $\fcy C$-symmetric theory these additional morphisms cannot appear in $\eqref{eq:brSymCor}$.

A simple object $\ba\in\fcy C$ does not necessarily remain simple in $\fcy D$, but will instead split into a direct sum of simple objects:
\begin{equation}
\fcy F (\ba) \approx {\bf b}_1\oplus...\oplus{\bf b}_n\,.
\end{equation}
We can then represent any operator $\phi\approx\ba$ as the sum of operators which are simple under $\fcy D$ (see Eq.~\reef{eq:subop})
\begin{equation}
\psi_i(x) = \fcy F(\phi)(x)\circ\iota_i\,,
\end{equation}
where $\iota_i: {\bf b}_i \rightarrow \fcy F(\ba)$ are morphisms embedding ${\bf b}_i$ in $\fcy F(\ba)$. Suppose one of these objects, say ${\bf b}_1$, is isomorphic to the ${\bf 1}\in\fcy D$, i.e.~$\psi_1(x)$ is a $\calD$-singlet. Consider then the following deformation of our theory:
\begin{equation}
S\rightarrow S + \lambda \int d^dx\,  \psi_1(x)\,.
\end{equation}
This perturbation breaks $\fcy C$ symmetry, while preserving $\calD$. Morphisms which were previously forbidden by $\fcy C$ can now appear in correlators.

As an example, let us break $O(n)$ model down to the $O(n-1)$ model. Within $O(N)$, the subgroup of matrices which preserves a vector $E^I$ is $O(N-1)$. Without loss of generality we can take ${E^I = (1,0,...0)}$, and can think of it as new invariant tensor mapping ${\bf N}\rightarrow{\bf 1}$, which takes a vector $V^I$ and returns the first component $V\cdot E= V^1$. Diagrammatically we can write this as:
\begin{equation}\label{eq:bONE} E^I = 
\includegraphics[trim=0 1em 0 0,scale=0.5]{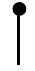} \,,\qquad E_IE^I = 1 = \includegraphics[trim=0 1em 0 0,scale=0.5]{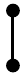}\,. \end{equation}
Although $E^I$ cannot be extended to non-integer $N$, these diagrams can be. We therefore define for any $n$ a new category $\Rephat\,O_E(n)$ where the morphisms are the string diagrams which can built from $\Rephat\,O(n)$ string diagrams supplemented by the new diagram \eqref{eq:bONE}. The inclusion functor $\hat{\fcy I}:\Rephat\,O(n)\rightarrow\Rephat\,O_E(n)$ takes any object $[k]\in\Rephat\,O(n)$ to the object $[k]\in\Rephat\,O_E(n)$, and any string diagram $[k_1]\rightarrow[k_2]$ in $\Rephat\,O(n)$ to the identical string diagram in $\Rephat\,O_E(n)$.

Following our procedure for constructing $\Reptilde\,O(n)$ case, we can take the Karoubi envelope and additive completion of $\Rephat\,O_E(n)$ to define a new semisimple category $\Reptilde\,O_E(n)$. The functor $\hat {\fcy I}$ lifts 
to a functor ${\tilde{\fcy I}:\Reptilde\,O(n)\rightarrow\Reptilde\,O_E(n)}$. Unlike in $\Reptilde\,O(n)$, the object ${\bf n}\in\Reptilde\,O_E(n)$ is not simple, because there are two linearly independent morphism ${\bf n}\rightarrow{\bf n}$. These can be used to define two indecomposable, mutually orthogonal projectors:
\begin{equation} P_{\bf 1} = \includegraphics[trim=0 3em 0 0,scale=0.4]{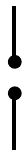}\,,\quad \text{tr}(P_{\bf 1}) = 1\,,\qquad P_{\bf {n-1}} = \text{id}_{\bf n}-P_{\bf 1}\,,\quad \text{tr}(P_{\bf {n-1}}) = n-1\,, \end{equation}
and so we can decompose ${\bf n} \approx {\bf 1} \oplus (\bf {n-1})$. 

It is now not hard to check that $\Reptilde\,O_E(n)$ and $\Reptilde\,O(n-1)$ are equivalent tensor categories, with an invertible functor $\fcy J: \Reptilde\,O_E(n)\rightarrow\Reptilde\,O(n-1)$ which takes the $\bf{n-1}$ in $\Reptilde\,O_E(n)$ to the $\bf{n-1}$ in $\Reptilde\,O(n-1)$. We have thus constructed a functor $\fcy J\circ\hat{\fcy I}$ mapping $\Reptilde\,O(n)$ to $\Reptilde\,O(n-1)$, and can use this to rewrite the $O(n)$ model in terms of $O(n-1)$ invariants. 

As another example, consider the subgroup $O(N-k)\times O(k)\subset O(N)$. In order to generalize this to non-integer $N$ and $k$, we must define the categorical analogue of the direct product of two groups.  Given any two categories $\fcy C$ and $\fcy D$, we define the product category $\fcy C\times\fcy D$ to be the category where the objects are pairs $(\ba,{\bf b})$ of objects $\ba\in\fcy C$ and ${\bf b}\in\fcy D$, and where the morphisms $(f,g):(\ba,{\bf b})\rightarrow({\bf c},{\bf d})$ are pairs of morphisms $f:\ba\rightarrow{\bf c}$ and $g:{\bf b}\rightarrow{\bf d}$ in $\fcy C$ and $\fcy D$ respectively. Given two groups $G_1$ and $G_2$, any simple representations of $G_1\times G_2$ is the tensor product of a simple representation of $G_1$ and $G_2$, and hence we find that the categories $\Rep\,(G_1\times G_2)$ and $(\Rep\,G_1)\times(\Rep\,G_2)$ are equivalent. We therefore define $\Reptilde\,(O(n-k)\times O(k))$ to be the product category $\Reptilde\,O(n-k)\times\Reptilde\,O(k)$.

Let us now define the functor $\hat{\fcy F}:\Rephat\,O(n)\rightarrow \Rephat\,O(n-k)\times\Rephat\,O(k)$ which takes the object ${[k]\in\Rephat\,O(n)}$ to the object $([k],[k])\in \Rephat\,O(n-k)\times\Rephat\,O(k)$, and takes string diagram $f:[k_1]\rightarrow[k_2]$ to the pair of string diagrams $(f,f):([k_1],[k_1])\rightarrow([k_2],[k_2])$. By taking the Karoubi envelope and additive completion of each category, it is not hard to verify that $\hat{\fcy F}$ lifts to a functor $\tilde{\fcy F}:\Reptilde\,O(n)\rightarrow \Reptilde\,O(n-k)\times\Reptilde\,O(k)$.

\subsection{Spontaneous symmetry breaking}
\label{sec:spont}

Symmetries in physics are often spontaneously broken. The textbook example of such a theory is the $O(N)$ vector model
\begin{equation}\label{eq:brokenON}
\fcy L = \frac 12 (\partial\phi^i)^2 - \frac 12 \mu^2 (\phi^i\phi_i)^2 + \frac 14 \lambda (\phi^i\phi_i)^2\,.
\end{equation}
For $\mu > 0$ the vector field $\phi^i(x)$ develops a vacuum expectation value:
\begin{equation}
\langle \phi^i(x)\rangle = v^i
\end{equation}
for some vector $v^i$. In this phase only an $O(N-1)$ subgroup is linearly realized, and the spectrum contains $N-1$ massless Goldstone bosons transforming in the fundamental of $O(N-1)$.

At first glance it seems impossible to describe spontaneous symmetry breaking using the categorical language. After all, correlation functions are always morphisms in the category $\Reptilde\,O(n)$, but in this seems to contradict what happens in the integer $N$ case, where only $O(N-1)$ symmetry is manifest. The key to this puzzle is to recall that the vacuum itself transforms non-trivially when a symmetry is spontaneously broken, and so we must compute expectation values with respect to a specific vacuum. 

To understand this in more detail, let us consider the $O(n)$ loop model on a finite lattice but with boundary conditions which are not $\Reptilde\,O(n)$ invariant. We can achieve this by first using the functor ${\tilde I:\Reptilde\,O(n)\rightarrow\Reptilde\,O(n-1)}$, so that the boundary points now transform in the ${\bf 1}\oplus({\bf n-1})$ of $\Reptilde\,O(n-1)$. Rather than integrating over each boundary site using the $\Reptilde\,O(n)$ invariant, we instead fix the boundary points to be in the $\bf 1$. To describe this graphically we can use $\Rephat\,O_E(n)$ from the previous section. 

Let us now compute the expectation value of some spin $s_{[i]}\approx\bf n$ in the bulk. Under $\Reptilde\,O(n-1)$ this splits into two representations, the $\vec s_{[i]}\approx{\bf n-1}$ and the $\hat s_{[i]} \approx{\bf 1}$. The expectation value of the former always vanishes due to $\Reptilde\,O(n-1)$ invariance, but $\langle\hat s_{[i]}\rangle$ will be in general non-zero. 

Now we can take the continuum limit while keep $\hat s_{[i]}$ far away from the boundary. Then there are two possibilities, depending on whether we are in the broken or unbroken phase. If we are in the unbroken phase then $\langle\hat s_{[i]}\rangle$ will go to zero in the continuum limit, and we find that full $\Reptilde\,O(n)$ invariance is restored. In the broken phase however, $\langle\hat s_{[i]}\rangle$ will limit to some constant but non-zero value. More generally, correlation functions will be $\Reptilde\,O(n-1)$ invariant (and also invariant under the Euclidean group), but not $\Reptilde\,O(n)$ invariant. The sole purpose of the non-invariant boundary condition is to pick out the correct vacuum for the theory.

We can also calculate correlators directly in the continuum using the path integral over \eqref{eq:brokenON}. To do so we again rewrite the $\Reptilde\,O(n)$ variables in terms of $\Reptilde\,O(n-1)$ variables using the functor $\tilde I$, under which $\phi(x)$ splits into a singlet $\varphi(x)\approx\bf 1$ and a vector $\Reptilde\,O(n-1)$ vector $\pi(x)\approx{\bf n-1}$. The potential will be minimized when $\varphi(x)$ takes some constant value $\varphi_0$. Introducing the field $\sigma(x) = \varphi(x) - \varphi_0$ and perturbing around $\sigma(x) = 0$, we can compute correlators in the broken phase. Here the choice of vacuum to perturb around determines the boundary conditions at infinity.

It is natural to define the dimension of a continuous categorical symmetry $\dim(\calC)$ as the dimension of its maximal adjoint object $\dim(\mathfrak{g}_\calC)$. The dimension of a discrete categorical symmetry is defined to be zero. We wish to conjecture the \emph{categorical Goldstone theorem}, which should go along these lines: If a theory with a continuous categorical symmetry $\calC$ is spontaneously broken to a category $\calD$, then long distance physics is described by a ``theory of a non-integer number of free massless Goldstone bosons''. In the categorical setting, the massless Goldstone fields should transform in a $\calD$ object of dimension $\dim(\calC)-\dim(\calD)$.

A beautiful example of this phenomenon can be found in Ref.~\cite{Jacobsen:2002wu}, which discussed the low-temperature dense phase of $O(n)$ model and argued that the presence of self-intersections drives the conformal fixed point to a Goldstone phase of spontaneously broken $O(n)$ symmetry. (Ref.~\cite{Jacobsen:2002wu} did not use the categorical language and treated non-integer $n$ symmetry at an intuitive level.)

\subsection{Conformally invariant theories}
\label{sec:CFT}
In this section we shall consider conformal field theories (CFTs) with categorical symmetries. Since these categorical symmetries are generalizations of global symmetries, we will see that the usual OPE expansions and crossing equations still hold in this more general context. 
 
We will restrict ourselves to CFTs satisfying the following technical assumptions:
\begin{enumerate}
	\item All scaling dimensions are real.
	\item For any $\Delta$ there are a finite number of states with scaling dimension up to $\Delta$.
 	\item The vacuum, with $\Delta = 0$, is the state of lowest conformal dimension and is unique.
 	\item The dilatation operator can be diagonalized.
 	\item The OPE expansion converges.
\end{enumerate}
Two-dimensional CFTs satisfying assumptions 1,2,3 are usually called `compact'. Here we consider general $d$. 

Assumption 5 is usually argued from unitarity/reflection positivity \cite{Pappadopulo:2012jk}, but may be reasonably expected to hold more generally (e.g. it holds for the 2d non-unitary minimal models). In the next section we will see that theories with categorical symmetries are in general non-unitary. 

Assumptions 2,3 are likely satisfied for all models of interest to statistical physics in $d>2$ and for most models in 2d.\footnote{In 2d, there do exist some statistical models with a continuous spectrum, see e.g. \cite{2006cond.mat.12037I} for an antiferromagnetic Potts model or \cite{vernier2015new} for a model of polymers, both of which are described by non-compact 2d CFTs. We thank Jesper Jacobsen for a discussion.}
For example, in 2d, the $O(n)$ loop models and the $S_n$ symmetric $n$-state ferromagnetic Potts models are known to have discrete (and real) spectrum for any $n$, integer or not \cite{diFrancesco:1987qf}. We expect the spectrum to remain discrete also in 3d.

Assumption 1 is a consequence of unitarity although it is satisfied for many non-unitary models as well. In general, non-unitary models may have complex spectrum (which will consist of complex-conjugate pairs if the model has an underlying real structure). Assumption 1 would then have to be relaxed, which is a rather trivial modification.

Assumption 4 is arguably the trickiest one. It is violated in logarithmic CFTs whose dilatation operator has Jordan blocks---a possibility for non-unitary theories. Logarithmic CFTs are more complicated than the usual ones and we will not consider them here (see e.g.~\cite{Hogervorst:2016itc}). 
The 2d $n$-state Potts models are logarithmic when $n$ belongs to a discrete sequences of Beraha numbers \cite{Couvreur:2017inl}, and these arguments may also extend to the $O(n)$ loop models \cite{Vasseur:2011fi}.\footnote{We thank Victor Gorbenko and Bernardo Zan for a discussion and for communicating to us further arguments suggesting that the 2d Potts and $O(n)$ model are in fact logarithmic for generic $n$.} In 3d, the $O(n)$ loop models and the $n$-state Potts models are logarithmic in the limit when $n$ approaches an integer \cite{Cardy:2013rqg}, and seemingly only in this case. For non-integer $n$, they could therefore be examples of models to which all our assumptions apply (except, perhaps, assumption 1 which as we said would be easy to relax).

 In the subsequent discussion we will stick to the above minimal set of assumptions.
 
Conformal primaries will transform as simple objects in $\fcy C$. Consider the space $V_{\Delta,\ba}$ of local operators $\phi(x)\approx\ba$ which have conformal dimension $\Delta$. Using both conformal and $\fcy C$ symmetries, operators in $V_{\Delta,\ba}$ have non-zero two-point functions only with operators in $V_{\Delta,\overline\ba}$. Given any basis $\phi_1(x),...,\phi_n(x)$ of $V_{\Delta,\ba}$ we can choose a basis of $\overline\phi_1(x),...,\overline\phi_n(x)$ of $V_{\Delta,\overline\ba}$ so that\footnote{To show that such a nice basis exists, we can start from the more general equation $
	\langle \phi_i(\ba,x)\psi_j(\overline\ba,y)\rangle = \frac{M_{ij}\delta^{\ba,\overline\ba}}{|x-y|^{2\Delta}}
$ for some finite matrix of coefficients $M_{ij}$. If $M_{ij}$ is degenerate then this means that there are fields for which all two-point correlators are zero, but due to our assumptions of OPE convergence this implies all of their correlators vanish. So without loss of generality we can take $M_{ij}$ to be non-degenerate, and then by applying $M^{-1}$ to our basis $\psi_1(x),...,\psi_n(x)$ would recover our nice basis. Furthermore, if $\ba \approx \overline\ba$ one can use Takagi's decomposition to find a basis such that $\phi_i=\overline \phi_i$, but we will not need this.} 
\begin{equation}
\langle \phi_i(x)\overline\phi_j(y) \rangle = \frac{\delta^{\ba,\overline\ba}\delta_{ij}}{|x-y|^{2\Delta}}\,.
\label{eq:nicebasis}
\end{equation}
In what follows we shall assume that such a basis has been chosen for each space of operators $V_{\Delta,\ba}$ and $V_{\Delta,\overline\ba}$.

Three-point functions are also fixed up to a finite number of coefficients, e.g. for scalars:
\begin{equation}\label{eq:3ptF}
\langle\phi_1(\ba_1,x_1)\phi_2(\ba_2,x_2)\phi_3(\ba_3,x_3)\rangle =\frac{F_{\phi_1\phi_2\phi_3}}{|x_{12}|^{\Delta_1+\Delta_2-\Delta_3}|x_{13}|^{\Delta_1+\Delta_3-\Delta_2}|x_{23}|^{\Delta_2+\Delta_3-\Delta_1}}
\end{equation}
for some morphism $F_{\phi_1\phi_2\phi_3}\in\Hom(\ba_1\otimes\ba_2\otimes\ba_3\rightarrow{\bf 1})$, where we define $x_{ij} = x_i - x_j$. We can also write a conformally invariant OPE
\begin{gather}
\phi_1(\ba_1,x_1)\phi_2(\ba_2,x_2) \supset F^{\overline\phi_3}_{\phi_1\phi_2} C_{123}(x_{12},\partial_2) 
{\overline\phi_3(\overline\ba_3},x_2)\,,\\
F^{\overline\phi_3}_{\phi_1\phi_2} = (F_{\phi_1\phi_2\phi_3}\otimes\text{id}_{\overline\ba_3})\circ(\text{id}_{\ba_1\otimes\ba_2}\otimes\delta_{\ba_3,\overline\ba_3})\,,
\end{gather}
where $C_{123}(x_{12},\partial_2)$ is the usual differential operator appearing in OPEs of scalars in CFTs.

When computing higher-point correlators we can use the OPE expansion multiple times, which will require us to compose morphisms of the form $F_{\phi_i\phi_j}^{\phi_k}$ together many times. We will find it convenient to use string diagrams to keep track of these compositions. To this end, let us introduce the diagrams: 
\begin{equation}
\delta^{{\bf b},\overline{\bf b}} = \includegraphics[trim=0 2em 0 0,scale=0.5]{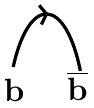}\,,\qquad \beta_{{\bf b},{\bf c}} = \includegraphics[trim=0 3em 0 0,scale=0.5]{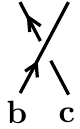}\,,\qquad F_{\phi_1\phi_2}^{\phi_3} = \includegraphics[trim=0 3.5em 0 0,scale=0.5]{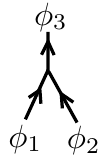}\,.
\end{equation}\\
We can now consider taking the OPE expansion in the simplest possible four-point function:
$$\langle \phi(x_1)\phi(x_2)\phi(x_3)\phi(x_4)\rangle,$$
all four operators being identical scalars $\phi(x) \approx \ba$. Using the OPE expansion as $x_1\rightarrow x_2$ we find that
\begin{equation}\begin{split}
\langle \phi(x_1)\phi(x_2)\phi(x_3)\phi(x_4)\rangle &= \frac 1{x_{12}^{2\Delta}x_{34}^{2\Delta}}\sum_{{\bf b}} \sum_{\fcy O_k\approx{\bf b}} g_{\fcy O_k}(u,v)\,\left [\delta^{{\bf b},{\overline{\bf b}}}\circ (F_{\phi\phi}^{\fcy O_k}\otimes F_{\phi\phi}^{\overline{\fcy O_k}})\right]\,\\
&= \frac 1{x_{12}^{2\Delta}x_{34}^{2\Delta}}\sum_{{\bf b}}\,\sum_{\fcy O_k\approx{\bf b}} g_{\fcy O_k}(u,v)\,\includegraphics[trim=0 1.5em 0 0,scale=0.8]{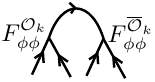}\,.
\end{split}\end{equation}
where $\sum_{\bf b}$ is over such $\bf b$ that both $\bf b$ and $\overline {\bf b}$ are contained in ${\bf a}\times {\bf a}$, and $g_{\fcy O_k}(u,v)$ are the conformal blocks depending on the usual cross-ratios.

But we can instead exchange $x_1\leftrightarrow x_3$ and then take the OPE as $x_2\rightarrow x_3$ to find that:
\begin{equation}\begin{split}
\langle \phi(x_1)\phi(x_2)\phi(x_3)\phi(x_4)\rangle &= \langle \phi(x_3)\phi(x_2)\phi(x_1)\phi(x_4)\rangle \circ \sigma_{13}\\
&= \frac 1{x_{23}^{2\Delta}x_{14}^{2\Delta}}\sum_{{\bf b}} \sum_{\fcy O_k\in{\bf b}} g_{\fcy O_k}\left( v,u \right)\left[ \delta^{{\bf b},{\overline{\bf b}}}\circ (F_{\phi\phi}^{\fcy O_k}\otimes F_{\phi\phi}^{\overline{\fcy O_k}})\circ \sigma_{13}\right] \\
&= \frac 1{x_{23}^{2\Delta}x_{14}^{2\Delta}}\sum_{{\bf b}} \sum_{\fcy O_k\in{\bf b}} g_{\fcy O_k}\left(v,u\right)\,\includegraphics[trim=0 2em 0 0,scale=0.63]{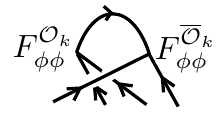}\,.
\end{split}\end{equation}
Here $\sigma_{13}= (\beta_{\ba,\ba}\otimes\text{id}_\ba\otimes\text{id}_\ba)\circ (\text{id}_\ba\otimes\text{id}_\ba\otimes \beta_{\ba,\ba})\circ (\beta_{\ba,\ba}\otimes\text{id}_\ba\otimes\text{id}_\ba)$ is the morphism interchanging the first and the third factors in ${\bf a}^{\otimes 4}$.

Equating these two expressions gives us the crossing equation:
\begin{equation}\label{eq:CrossAb}
v^\Delta \sum_{{\bf b}} \sum_{\fcy O_k\approx{\bf b}} g_{\fcy O_k}(u,v) 
\left [\delta^{{\bf b},{\overline{\bf b}}}\circ (F_{\phi\phi}^{\fcy O_k}\otimes F_{\phi\phi}^{\overline{\fcy O_k}})\right ]\\
= u^\Delta\sum_{{\bf b}} \sum_{\fcy O_k\approx\bf b} g_{\fcy O_k}\left(v,u\right)\left[ \delta^{{\bf b},{\overline{\bf b}}}\circ (F_{\phi\phi}^{\fcy O_k}\otimes F_{\phi\phi}^{\overline{\fcy O_k}})\circ\sigma_{13}\right]\,.
\end{equation}
 This form of the crossing equation may look rather abstract as an equality between two abstract morphisms in a tensor category. By picking a basis of morphisms $f_{1},...,f_{m}$ for the space $\Hom(\ba^{\otimes 4}\rightarrow{\bf 1})$ we can rewrite \eqref{eq:CrossAb} as a series of $m$ equations. Alternatively, we can choose a basis of morphisms $h_1,...,h_m$ for the space $\Hom({\bf 1}\rightarrow\ba^{\otimes 4})$ and impose the conditions
 \begin{equation}\label{eq:CrossCon}
 v^\Delta \sum_{{\bf b}} \sum_{\fcy O_k\approx{\bf b}} g_{\fcy O_k}(u,v) \left [\delta^{{\bf b},{\overline{\bf b}}}\circ (F_{\phi\phi}^{\fcy O_k}\otimes F_{\phi\phi}^{\overline{\fcy O_k}}) \circ h_l \right ]\\
 = u^\Delta\sum_{{\bf b}} \sum_{\fcy O_k\approx\bf b} g_{\fcy O_k}\left(v,u\right)\left[ \delta^{{\bf b},{\overline{\bf b}}}\circ (F_{\phi\phi}^{\fcy O_k}\otimes F_{\phi\phi}^{\overline{\fcy O_k}})\circ\sigma_{13}\circ h_l\right]\,.
 \end{equation}
for each $h_l\in\Hom({\bf 1}\rightarrow\ba^{\otimes 4})$. These two methods are equivalent, as $\Hom({\bf 1}\rightarrow\ba^{\otimes 4})$ is isomorphic to the space of linear functionals on $\Hom(\ba^{\otimes 4}\rightarrow{\bf 1})$, see proposition \ref{pr:ndeg}.

As an illustration of the introduced language, we will now prove the following result. We state it in the general setting of categorical symmetries, but it might be a new result even for ordinary, group symmetries.

\begin{thm}\label{thm:AllObs} {\rm(``Completeness of the global symmetry spectrum'')} If a CFT contains operators $\phi_1\approx\ba_1$ and $\phi_2\approx\ba_2$, then there must be operators transforming in every ${\bf B} \in \ba_1\otimes\ba_2$\,.
\end{thm}

\begin{proof} We will restrict ourselves to the case where both $\phi_1$ and $\phi_2$ are scalars, as the generalization to spinning operators is straightforward. Consider the four-point function 
\begin{equation}
\langle\phi_1(x_1)\overline\phi_1(x_2)\phi_2(x_3)\overline\phi_2(x_4)\rangle\,.
\end{equation}
Performing the OPE in the limit $x_1\to x_2$, the leading singular term is proportional to the morphism $\delta^{\ba_1,\overline\ba_1}\otimes\delta^{\ba_2,\overline\ba_2}$.
If we instead perform the OPE expansion between the first and third points, we find an infinite sum of terms 
of the form
\begin{equation}\label{eq:x13Exp}
 \sum_{{\bf b}\in\ba_1\otimes\ba_2}\sum_{\fcy O_k\approx\bf b}  \alpha_{\fcy O_k}\left(x_i\right)
 \raisebox{-1.5em}{\includegraphics[trim=0 0 0 0,scale=0.6]{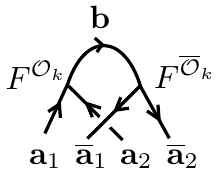}}\,,
 \end{equation}
 where the shown morphism diagrams can be written formally as
 $\delta^{{\bf b},{\overline{\bf b}}}\circ (F^{\fcy O_k({\bf b})}_{{\bf a}_1{\bf a}_2}\otimes F^{\overline{\fcy O}_k({\overline{\bf b}})}_{\overline{\bf a}_1\overline{\bf a}_2})\circ(\text{id}_{\ba_1}\otimes\beta_{\overline\ba_1,\ba_2}\otimes\text{id}_{\overline\ba_2})$. The coordinate-dependent coefficients $\alpha_{\fcy O_k}$ can be expressed via conformal blocks but this is unimportant for the present argument.

Since \eqref{eq:x13Exp} must somehow reproduce the OPE expansion in the limit $x_1\to x_2$, we conclude that there must exist a representation 
\begin{equation}\label{eq:t1CatEq}
\delta^{\ba_1,\overline\ba_1}\otimes\delta^{\ba_2,\overline\ba_2} = 
\raisebox{-1.5em}{\includegraphics[scale=0.6]{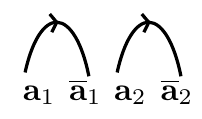}}= 
\sum_{{\bf b}\in\ba_1\otimes\ba_2}\sum_{\fcy O_k\approx\bf b}\beta_{\fcy O_k} \raisebox{-1.5em}{\includegraphics[scale=0.6]{fig-Th2.pdf}}
\end{equation}
for some numerical coefficients $\beta_{\fcy O_k}$. The theorem now reduces to a category theoretic statement, that for any ${\bf b}\in \ba_1\otimes\ba_2$, $\beta_{\fcy O_k}$ must be non-zero for some $\fcy O_k\approx\bf b$ for this representation to hold. 

To prove this last claim, consider any particular simple $\bf B$ which appears in $\ba_1\otimes\ba_2$. This means that there are projection and embedding morphism $\pi:\ba_1\otimes\ba_2\to{\bf B}$ and $\iota:{\bf B}\rightarrow\ba_1\otimes\ba_2$, written diagrammatically as:
\begin{equation}
\pi = \raisebox{-1.5em}{\includegraphics[trim=0 0 0 0,scale=0.6]{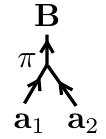}}\,,\qquad \iota = \raisebox{-1.5em}{\includegraphics[trim=0 0 0 0,scale=0.6]{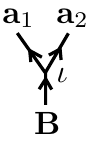}}\,.
\end{equation}
Let us now define the morphism $h_{\bf B}:{\bf 1}\rightarrow \ba^{\otimes 4}$:
 \begin{equation}
 h_{\bf B} = \raisebox{-1.5em}{\includegraphics[scale=0.6]{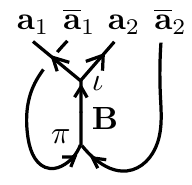}}\,.
\end{equation}
Composing both sides of \eqref{eq:t1CatEq} with $h_{\bf B}$, we find on the l.h.s.
\begin{equation}
 \raisebox{-2.7em}{\includegraphics[scale=0.6]{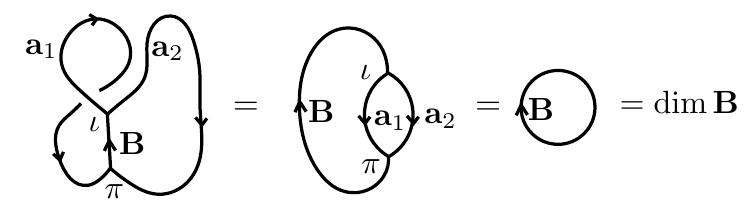}}\,,
 \end{equation}
 where we used among other things that $\pi\circ\iota = \text{id}_{\bf B}$. By proposition \ref{pr:zeroDim},
 all simple objects have non-zero dimensions, and thus the l.h.s. is non-zero.
 
On the r.h.s. we instead find a sum of terms like:
\begin{equation}
 \includegraphics[trim=0 3em 0 0,scale=0.6]{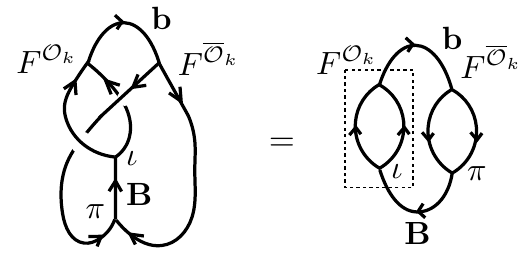}
\end{equation}\\
weighted with the numerical $\beta_{\calO_k}$ coefficients.
Now note that the sub-diagram in the r.h.s.~surrounded by a dotted box
is a morphism $F^{\fcy O_k}_{\phi_1\phi_2}\circ\iota$ between the simple objects $\bf b$ and $\bf B$. This morphism vanishes unless $\bf b\approx\bf B$. We conclude that to match the non-zero l.h.s., the r.h.s.~must contain at least one non-zero term corresponding to some $\fcy O_k\approx\bf B$. This completes the proof.
\end{proof}

We say that an object $\bf g$ \emph{generates} $\fcy C$ if any simple object in $\fcy C$ appears in ${\bf g}^{\otimes k}$ for sufficiently large $k$. Decompose such a $\bf g$ into simple objects: ${{\bf g}\approx \ba_1\oplus...\oplus\ba_n}$.
If a CFT with symmetry $\calC$ has operators $\phi_i\approx\ba_i$ for every $\ba_i$ then, by theorem \ref{thm:AllObs}, \emph{it contains a local operator $\phi(x)\approx\ba$ for every simple object $\ba\in\fcy C$.} We then say that the categorical symmetry is \emph{faithful}. E.g. since both the free and interacting $O(n)$ models contain operators isomorphic to $\bf n$, and since $\bf n$ generates $\Reptilde\,O(n)$, we conclude that $\Reptilde\,O(n)$ is a faithful symmetry of these theories.

For any faithful (in the usual sense) representation $\ba$ of a compact group $G$, $\ba$ and $\overline\ba$ together generate the category $\Rep\,G$. (See \cite{Levy:2003my} for a simple analytic proof and \cite{milneAGS} for a more advanced discussion.) Combining this with theorem \ref{thm:AllObs}, we conclude that given a CFT with a group symmetry $G$, if there is a local operator which transforms faithfully under $G$, then operators transforming under arbitrary representations of $G$ must occur in the spectrum, and $\Rep\,G$ is faithful.

Let us now consider the case where $\fcy C$ is not faithful. We can then define a set $D$ which contains the simple objects which are isomorphic to local operators. By theorem \ref{thm:AllObs} the tensor product of any two objects in $D$ will be the direct sum of objects in $D$, and so we see that the objects in $D$ and their direct sums form a subcategory $\fcy D\subset\fcy C$. By construction, $\fcy D$ is a faithful symmetry of our theory. For this reason we can always restrict to theories with faithful symmetries with no loss of generality.

\subsection{Reality and unitarity}
\label{sec:unitarity}
Before defining unitarity, or rather its Euclidean counterpart reflection positivity, let us consider the simpler notion of complex conjugation. Given any representation $\rho:G\rightarrow GL(\mathbb C^k)$ we can define a conjugate representation by simply taking the complex conjugate of $\rho$. In more abstract language, this associates to each representation $\ba\in\Rep\,G$ a conjugate representation $\ba^* \approx \overline\ba$. Furthermore for any invariant tensor $f:\ba\rightarrow{\bf b}$ we can complex conjugate it to construct a tensor $f^*:\ba^*\rightarrow{\bf b}^*$. The braiding is trivially preserved by this:
\begin{equation}
\beta^*_{\ba,{\bf b}} = \beta_{\ba^*,{\bf b}^*}
\end{equation}
 
We would like to extend this definition to any symmetric tensor category $\fcy C$. Let us define a \emph{conjugation} $*$ as a anti-linear braided monoidal functor $*:\fcy C\rightarrow\fcy C$ which satisfies $** = \text{id}_{\fcy C}$, and for which $\ba^*$ is dual to $\ba$.\footnote{While $\ba^*$ is isomorphic to $\overline\ba$, it is not necessarily equal to $\overline\ba$, and for this reason we shall use $\ba^*$ rather than $\overline\ba$ when referring to the dual of an object in a category with conjugation. Indeed $\overline \ba$ for simple objects is defined up to a rescaling by a complex number, and $\ba^*$ is a representative in the isomorphism class for which \eqref{eq:braidDelta} is true. This may seem overly pedantic, but it helps us avoid confusion between $\delta^{\ba,\ba^*}$ (which we require to satisfy \eqref{eq:braidDelta}), and $\delta^{\ba,\overline\ba}$ (which we do not).} As we show in proposition \ref{pr:realCap}, for every simple object $\ba\in\fcy C$ there exist cap and cocap maps satisfying
\begin{equation}\label{eq:braidDelta}
(\delta^{\ba^*,\ba})^* = \delta^{\ba,\ba^*}\,,\quad(\delta_{\ba,\ba^*})^* = \delta_{\ba^*,\ba}\,.
\end{equation}
 
As in \cite{Gorbenko:2018ncu}, we will say that a quantum field theory is \emph{real} if there exists a map $*$ acting on local operators:
\begin{equation*}
(\phi(\ba,x))^* = \phi^*(\ba^*,x)
\end{equation*}
which is involutive
\begin{equation}\phi^{**}(x) = \phi(x),\end{equation}
and such that 
\begin{equation}
(\langle \phi_1(x_1)...\phi_n(x_n)\rangle)^* = \langle \phi_1^*(x_1)...\phi_n^*(x_n)\rangle.
\end{equation}
 
Let us first consider the implications of reality for two-point functions. As in the previous section can we consider the space $V_{\Delta,\ba}$ of local operators $\phi(x)\approx\ba$ with conformal dimension $\Delta$. Given any basis $\phi_1,...,\phi_n$ of $V_{\Delta,\ba}$ we can then compute:
\begin{equation}\label{eq:2ptconj}
\langle \phi_i(x)\phi_j^*(y)\rangle = \frac{M_{ij}}{|x-y|^{2\Delta}}\delta^{\ba,\ba^*}
\end{equation}
for some matrix $M_{ij}$. By conjugating both sides of this equation we find that
\begin{equation}\label{eq:tpConj}\begin{split}
\langle \phi_i^*(x)\phi_j(y)\rangle = \frac{M_{ij}^*}{|x-y|^{2\Delta}}\delta^{\ba,\ba^*}\,,
\end{split}\end{equation}
while by interchanging the two operators we find that:
\begin{equation}\label{eq:tpBraid}
\langle \phi_i^*(x)\phi_j(y)\rangle = \langle \phi_j(y)\phi_i^*(x)\rangle \circ\beta_{\ba^*,\ba} = \frac{M_{ji}}{|x-y|^{2\Delta}}\delta^{\ba^*,\ba}\,,
\end{equation}
and so we can conclude that $M_{ij}$ is Hermitian. We can therefore always choose a basis of $V_{\Delta,\ba}$ such that
\begin{equation}\label{eq:conjbasis}
\langle \phi_i(x)\phi_j^*(y)\rangle = \frac{\delta_{ij} \fcy N_{\phi_i}}{|x-y|^{2\Delta}}\delta^{\ba,\ba^*} \text{ with } \fcy N_{\phi_i} = \pm1\,.
\end{equation}
From now on we shall assume that such a basis has been chosen for each $V_{\Delta,\ba}$.

We should note that despite superficial similarities, \eqref{eq:conjbasis} is quite distinct from the basis \eqref{eq:nicebasis} used in the previous section. In a general CFT there is no relationship between correlators involving $\phi_i(x)$ and those of $\overline\phi_i(x)$. Furthermore the linear map we constructed between $V_{\Delta,\ba}$ and $V_{\Delta,\overline\ba}$ depended on our choice of basis for $V_{\Delta,\ba}$. By contrast, in real CFT conjugation gives us a natural mapping between $V_{\Delta,\ba}$ and $V_{\Delta,\overline\ba}$ and this turns $V_{\Delta,\ba}$ into an indefinite Hilbert space. Conjugation also relates correlators of $\phi(x)$ and $\phi^*(x)$ with any number of operators.

Let us now consider the constraints imposed by reality on three-point functions. Consider a $\fcy C$ singlet $\Phi_l(x)$ with spin $l$. We have
\begin{equation}\label{eq:3ptUnit}
\langle \phi^*(x)\phi(y)\Phi_l(z) \rangle = \delta^{\ba^*,\ba}\, f_{\phi^*\phi\Phi}\, S_l(x,y,z)
\end{equation}
where $S_l(x,y,z)$ is a real function of $x,y,z$ which is completely fixed by conformal invariance. In particular, it satisfies the crossing property
\begin{equation}
S_l(y,x,z) = (-1)^lS_l(x,y,z).
\end{equation}
{Then, by the same logic we used to constrain $\fcy N_\phi$, we find a relation:
\beq
(f_{\phi^*\phi\Phi})^*=(-1)^l  f_{\phi^*\phi\Phi^*}\,.
\eeq 
Suppose further that $\Phi_l^*=\Phi_l$ (such fields are called \emph{real}).} Then the previous equation implies that $f_{\phi^*\phi\Phi}$ is real if $l$ is even and imaginary if $l$ is odd.

Having defined what it means for a theory to be real, let us move on to unitarity. In Euclidean signature this manifests as reflection positivity, which for a regular quantum field theory states that
\begin{equation}\label{eq:Uni}
\langle \phi_n^*(-\tau_n,\vec x_n)\ldots \phi_1^*(-\tau_1,\vec x_1)\phi_1(\tau_1,\vec x_1)\ldots \phi_n(\tau_n,\vec x_n) \rangle \geq 0.
\end{equation}
for any positive $\tau_i$ and any $\vec x_i$.\footnote{In fact, full reflection positivity is a stronger condition, which involves integrating Eq.~\reef{eq:Uni} with reflection-symmetric test functions, and also considering linear combinations of $(m+n)$-point functions \cite{osterwalder1973,osterwalder1975}. Here we will just consider the partial case \reef{eq:Uni} for simplicity.} 

For fields with non-trivial spin, the $SO(d)$ indices of fields in the l.h.s. of \reef{eq:Uni} have to be contracted with external polarization tensors in the conjugation-reflection-symmetric way. E.g.~a certain $(\phi_i)_\mu$ is contracted with $\xi^\mu$ then the corresponding $(\phi^*_i)_\mu$ has to be contracted with $(\theta \xi^*)^\mu$ where $\theta=\text{diag}(-1,1,\ldots,1)$ is the reflection. Positivity should then hold for all possible such contractions.

If some fields in the l.h.s.~of \reef{eq:Uni} transform in non-trivial global symmetry representations, those indices should also be contracted with external global symmetry tensors, in a conjugation-symmetric way.
	
When generalizing reflection positivity to quantum field theory with a categorical symmetry, we have to think how to implement the latter property. We do this by requiring the positivity condition
\begin{equation}
\langle \phi_n^*(-\tau_n,\vec x_n)\ldots\phi_1^*(-\tau_1,\vec x_1)\phi_1(\tau_1,\vec x_1)\ldots\phi_n(\tau_n,\vec x_n) \rangle \circ U \geq 0
\end{equation}
{where $U\in \text{Hom}({\bf 1}\to \ba_n^*\otimes\ldots \ba_1^*\otimes \ba_1\otimes\ldots \ba_n)$ is any \emph{conjugation-reflection-symmetric} morphism, i.e. one satisfying $U^*=R U$ 
	where $R$ is the morphism from $\ba_n \otimes\ldots \ba_1\otimes \ba_1^*\otimes\ldots \ba_n^*$ to the same tensor product in the opposite order ($\ba_n^*\otimes\ldots \ba_1^*\otimes \ba_1\otimes\ldots \ba_n$)
which just connects reflection-symmetrically the tensor product factors by identities ($\ba_i$ to $\ba_i$, $\ba_i^*$ to $\ba_i^*$). One example of such $U$ can be constructed using the cocap maps:
\begin{equation}
U = \includegraphics[trim=0 3em 0 0,scale=0.4]{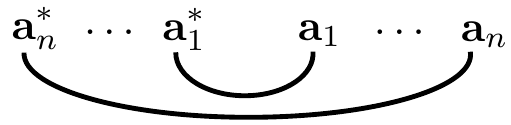}\,;
\end{equation}\\[-0.5em]
many others can be constructed by decomposing the tensor product $\ba_1\otimes\ldots \ba_n$ into simple objects. }

Restricting our attention to singlets, we can use the simplest definition of reflection positivity \reef{eq:Uni}, which in particular requires that, for any primary scalar $\phi\approx\bf 1$:
\begin{equation}\label{eq:singUnit}
\langle \phi^*(x)\phi(y)\rangle = \frac{\fcy N_\phi}{|x-y|^{2\Delta}} \text{ with } \ \fcy N_\phi > 0\,.
\end{equation}
This generalizes for non-scalar singlets, whose two-point function includes a conformally invariant tensor structure. Imposing reflection positivity for descendants, we obtain the usual unitarity bounds on scaling dimensions.

With these natural definitions, it turns out that none of the CFTs with tensor categorical symmetries can be reflection positive, unless it's an ordinary group symmetry. Moreover, the lack of reflection positivity manifests itself even in the singlet sector. More precisely we have the following theorem:

\begin{thm}\label{thm:nonUnitary} {\rm (``Lack of unitarity'')} If a real CFT has a faithful tensor categorical symmetry $\fcy C$, has a reflection-positive singlet sector, and satisfies the technical conditions listed at the beginning of the previous section, then $\fcy C$ is equivalent to $\Rep\,G$ for some group $G$.
\end{thm}
\noindent This follows by combining Theorem \ref{thm:AllObs} with the following two results:
\begin{prop}\label{pr:negDimCFT}
	Let $\phi(\ba,x)$ be an operator in a real CFT such that $\mathrm{dim}(\ba)<0$. Then the $\phi \times \phi^*$ OPE contains an operator $\Phi(x)$ which transforms trivially under $\fcy C$ and which violates reflection positivity (either because $\fcy N_\Phi < 0$ or because $\Delta_\Phi$ violates unitarity bounds). 
\end{prop}
\begin{prop}\label{pr:CRepG}
 	If a symmetric tensor category $\fcy C$ contains no negative dimensional objects, then it is equivalent to $\Rep\,G$ for some group $G$.
\end{prop}

\begin{proof}
	Proposition \ref{pr:CRepG} is a consequence of a theorem by Deligne, which classifies all symmetric tensor categories satisfying a certain finiteness condition, see Appendix \ref{sec:Deligne}.
	
	Let us prove proposition \ref{pr:negDimCFT} for $\phi$ a primary scalar; see below for the changes needed otherwise. Consider the correlator $\langle\phi(x_1)\phi^*(x_2)\phi(x_3)\phi^*(x_4)\rangle$. 
Like in our proof of Theorem \ref{thm:AllObs}, the short-distance limit of the OPE expansion in the $2\to1$ channel is dominated by the unit operator contribution
\beq
\fcy N_\phi^2 |x_{12}|^{-2\Delta_\phi}|x_{34}|^{-2\Delta_\phi} \delta^{\ba,\ba^*}\otimes\delta^{\ba,\ba^*}\,. 
\eeq
This must be reproduced by an infinite sum of terms from the $2\to 3$ channel, of the form
\beq
\sum_{{\bf b}\in\ba\otimes\ba^*}\sum_{\fcy O_k\approx\bf b} \fcy N_{\fcy O_k} \includegraphics[trim=0 3em 0 0,scale=0.5]{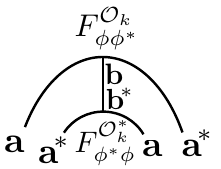} \times \text{(conformal blocks)}\,,
\eeq
where the diagram represents the morphism $\delta^{\bf b,\bf b^*}\circ(F_{\phi\phi^*}^{\fcy O_k}\otimes F_{\phi^*\phi}^{{\fcy O_k}^*})$. We then contract both parts of the crossing equation with the morphism 
\begin{equation}
\includegraphics[trim=0 2em 0 0,scale=0.4]{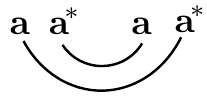} \in \text{Hom}({\bf 1}\to \ba\otimes \ba^*\otimes \ba\otimes \ba^*)\,.
\end{equation}
In the l.h.s.~we get the leading $x_2\to x_1$ singularity times $\calN_\phi^2 \dim(\ba)<0$. 
In the r.h.s.~the contraction projects on the ${\bf b}\approx {\bf 1}$ part since otherwise $F_{\phi\phi^*}^{\calO_k} \circ \delta_{\ba,\ba^*}=0$. The ${\bf b}\approx {\bf 1}$ terms contribute $\fcy N_{\fcy O_k} \dim(\ba)^2 |f_{\phi \phi^* \calO_k}|^2 $ times conformal blocks. 
Consider a reflection positive kinematic configuration (e.g.~all 4 points along a line at $-1,-z,z,1$). As is well known, conformal blocks in such configurations are positive \emph{provided that the primary $\calO_k$ is above the unitary bounds.} If we assume in addition that $\fcy N_{\fcy O_k}\ge0$ then all terms in the r.h.s. are positive, in manifest disagreement with the negative sign of the l.h.s. Thus there must be a singlet $\fcy O_k$ for which either $\fcy N_{\fcy O_k}<0$, or the dimension is below the unitarity bounds.

The argument still works if $\phi$ is not a scalar, nor does it have to be a primary. We consider the same correlator, but now we have to choose polarizations. We pick some arbitrary identical polarization for both $\phi$'s. We then assign the reflected polarization to both $\phi^*$'s. With this choice, the sign of the leading singularity in the $x_2\to x_1$ channel will be as for scalars, controlled by $\dim(\ba)$, hence negative. In the $x_2\to x_3$ channel, instead of thinking in terms of conformal blocks, we observe that the contribution of any positive-norm state is positive. Hence a negative-norm state must exist, which can be either a primary $\fcy N_{\fcy O_k}<0$, or a descendant (if the primary is below unitarity). 
\end{proof}

Theorem \ref{thm:nonUnitary} allows us to disprove unitarity for many theories. For instance any $O(n)$ model for non-integer $n$ is necessarily non-unitary (clearly, the corresponding category, having objects of non-integer and negative dimensions, cannot be equivalent to $\Rep\,G$). Unfortunately, this implies that we cannot use the most robust numerical bootstrap techniques to rigorously study phase transitions in these models, as these require positivity conditions on squared OPE coefficients and operator norms. For sufficiently large $n$, the unitary violations need only occur for operators with large conformal dimensions, and so standard numerical methods could potentially give reasonable, albeit non-rigorous, results.\footnote{Such attempts were made in \cite{Shimada:2015gda} to study the non-integer $O(n)$ models, but it is unclear whether the non-unitarity was sufficiently small to allow this.} Alternatively, the truncation method of Gliozzi \cite{Gliozzi:2013ysa} does not rely on positivity but is less systematic. See \cite{Poland:2018epd}, section VIII, for a detailed discussion.

A particularly simple case to consider is the free $O(n)$ model. Unitary violations in this theory have previously been argued for in \cite{Maldacena:2011jn}, section 5.5. They considered the norm of, in their language, the operator:
\begin{equation}\label{eq:MZOp}
\fcy O_k = \delta_{I_1J_1}...\delta_{I_kJ_k}\phi^{[I_1}(x)\partial\phi^{I_2}(x)\dots\partial^k\phi^{I_k]}(x)\phi^{[J_1}(x)\partial\phi^{J_2}(x)\dots\partial^k\phi^{J_k]}(x)\,.
\end{equation}
To make sense of this operator they simply analytically continue the integer $n$ computation, finding that
\begin{equation}\label{eq:MZNorm}
\langle \fcy O_k(x_1)\fcy O_k(x_2) \rangle =  n(n-1)...(n-k+1)\frac{\fcy N(k)}{|x_1-x_2|^{2\Delta_k}}\,,
\end{equation}
where $\fcy N(k)$ is some positive function of $k$ and $\Delta_k = k(k+d-3)$. If $k = \lceil n\rceil +1$ and $n$ is not an integer then the state has negative norm.

Using our categorical language, we can make the arguments of \cite{Maldacena:2011jn} more precise. This allows us to eliminate some potential doubts in its validity. E.g., one could wonder whether the operator \eqref{eq:MZOp} is really well defined, or required by the theory, or whether other analytic continuations could be possible. To properly define the operator $\fcy O_k$ we can first construct
\begin{equation}
A_k(x) = \phi(x)\partial\phi(x)\dots\partial^k\phi(x)\circ P_{{\bf A}^k} \,,
\end{equation}
where $P_{{\bf A}^k}$ is the projector defined in \eqref{eq:PAk}, projecting onto the simple object ${\bf A}^k$ with
\begin{equation}
\mathrm{dim}({\bf A}^{k}) = \frac{n(n-1)...(n-k+1)}{k!}\,.
\end{equation}
We now define 
\begin{equation}
\fcy O_k = A_k(x)A_k(x) \circ \gamma\,,\qquad\gamma = \raisebox{-0.5em}{\includegraphics[scale=0.5]{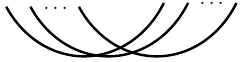}}\,,
\end{equation}
which gives us a precise definition of \eqref{eq:MZOp} without the need for analytic continuation. One can then check \eqref{eq:MZNorm} using string diagram manipulations. For $k= \lceil n\rceil+1$ the object ${\bf A}^k$ has negative dimension. So by proposition \ref{pr:negDimCFT} we find that the OPE of $A_{\lceil n\rceil+1}$ with itself contains a singlet of negative norm, and $\fcy O_{\lceil n\rceil+1}$ is precisely such an operator.

\section{Why this story is not entirely algebraic}
\label{sec:pasha}
Much of what we have written above about Deligne categories may create an impression that the whole theory is about taking some polynomials and interpolating them from integers to reals. Since Deligne categories are algebraic objects, one may ask if there is an interesting interplay between them and structures of an analytic nature.\footnote{We are grateful to Pavel Etingof for this question which prompted us to write this section.} The answer to this question is yes. This may be rather obvious to physicists, so this sketchy and by no means exhaustive section is mainly intended for mathematicians. Many more examples of non-trivial interplay can be given.

This interesting interplay is bound to appear when one considers lattice models or quantum field theories with categorical symmetries and computes observables in these theories. As a first example, consider a lattice model with $O(n)$ categorical symmetry, of the type discussed in section \ref{sec:lattice}. Assuming for simplicity we are on a 2d square lattice, we can e.g.~consider the transfer matrix of this model on a cylinder of circumference $k$ and of a unit height, which is a particular endomorphism $T$ of the object $[k]$: $T\in \text{End}([k])=\text{Hom}([k]\to [k])$. The exact form of $T$ depends on the lattice model we are considering. Endomorphisms of $[k]$ act by composition on $\text{Hom}(\ba \to [k])$ for any simple object $\ba$ and we can ask what are the eigenvalues $\lambda_i$ of the transfer matrix $T$ acting in this way. These eigenvalues are analogues of energy levels. Categorical symmetry tells us that the eigenvalues are classified by simple objects of the associated Deligne category, but the numerical values of the eigenvalues depend on the model. Using the Brauer algebra, we find that matrix elements of $T$ will be polynomials in $n$, and so the eigenvalues will be algebraic functions of $n$. The transfer matrix of the same model on a cylinder of height $L$ will be $T^L$, the $L$-th power of $T$, with eigenvalues $\lambda_i^L$. We can also consider the trace of $T^L$ which will be a linear combination of these eigenvalues times the dimensions of the corresponding simple objects $\ba$, and is known as the torus partition function of the lattice model.

Physically, the most interesting lattice quantities are those which appear in the thermodynamic limit, i.e.~on an infinite lattice (also known as the continuum limit, if we send the lattice spacing to zero keeping the volume fixed). If the parameters of the lattice models are tuned to a second-order phase transition, this limit will be described by a conformal field theory. Scaling dimensions of CFT operators are related to logarithms of the transfer matrix eigenvalues. A classic example is the continuum limits of $O(n)$ models, which exist for $n\in [-2,2]$ on two dimensional lattices, while in 3d it should exist for arbitrarily large $n$. The torus partition function of the $O(n)$ model CFTs in 2d was computed many years ago using Coulomb gas methods \cite{diFrancesco:1987qf}. This is a non-trivial modular-invariant analytic function of the torus parameter, which depends continuously on $n$, which has the schematic form:
\beq
Z(n,q,\bar q) = \sum_{i=0}^\infty M_i(n) \chi(h_i(n),q)\chi(\bar h_i(n),\bar q) \,.
\eeq
Here $q,\bar q$ are the torus modular parameters, $\chi(h,q)$ are the Virasoro characters, $h_i(n)$ and $\bar h_i(n)$ are the holomorphic and antiholomorphic weights of the Virasoro primary fields, and $M_i(n)$ are their `multiplicities'. By consistency with the categorical $O(n)$ symmetry, multiplicities $M_i(n)$ are  linear combinations of dimensions of simple $O(n)$ objects with positive integer coefficients.\footnote{There was some uncertainty in \cite{diFrancesco:1987qf} whether this property holds for their partition function, but a more careful check \cite{Bernardo} shows that it does, see also note 19 in \cite{Gorbenko:2018dtm}. This property is related to the semisimplicity of the Deligne property, i.e.~that every object is a direct sum of an integer number of simple objects.} On the other hand, the Virasoro weights and the central charge of the theory are more complicated, non-algebraic functions of $n$ (they are rational functions of the Coulomb gas coupling $g$ which depends non-algebraically on $n$). An analogous story holds for the $q$-state Potts models, which have a second-order phase transition for $q\le q_c$ where $q_c=4$ in 2d, while in 3d it is believed that $q_c\approx 2.45$ although it is not known exactly (see \cite{Gorbenko:2018ncu}). Scaling dimensions have a branch point at $q_c$, as another sign of non-trivial analytic structure, and interesting physics can be described by analytically continuing beyond this branch point \cite{Gorbenko:2018dtm}. 

In three dimensions, scaling dimensions and OPE coefficients of the critical $O(N)$ model have been studied by conformal bootstrap methods \cite{Kos:2015mba} (see \cite{Poland:2018epd} for review). That these data can be analytically continued in $N$ in a way which correspond to categorically symmetric CFTs is a highly non-trivial fact. Little is known about these analytic continuations in 3d, and they are bound to exhibit a highly non-trivial dependence on $N$.

\section{Other Deligne categories}
\label{sec:other}
So far we have considered the category $\Reptilde\,O(n)$, which extends the representation of the group $O(N)$ to non-integer values. We will now discuss analogous $U(N)$, $Sp(N)$ and $S_N$ constructions, also due to Deligne. 
 
Each construction begins by considering invariants in the category $\Rep\,G_N$ for some family of groups $G_N$ depending on a discrete parameter $N$. The task is to find ``fundamental invariants'' from which all other invariants can be built through the braiding and tensor product. For instance, in $\Rep\,O(N)$ all invariants can be constructed from the tensor $\delta^{IJ}$. We then represent these invariants using string diagrams. To compose these diagrams we stack them horizontally, and simplify the diagrams using rules which are dependent on a parameter $N$. These diagrams form a category $\Rephat\,G_N$, with a functor $\fcy F:\Rephat\,G_N\rightarrow\Rep\,G_N$ translating the string diagrams back to invariant tensors.
 
 Unlike $\Rep\,G_N$, the category $\Rephat\,G_n$ makes sense even for non-integer values $n$. We can compute idempotent morphisms, dimensions, and other representation theoretic data in $\Rephat\,G_n$ for any value of $n$.
 
 Unlike in $\Rep\,G_N$, in $\Rephat\,G_n$ we cannot in general take direct sums of objects, and we cannot decompose objects into simple objects. We can fix this by taking the Karoubi envelope and additive completion of $\Rephat\,G_n$, constructing a new category $\Reptilde\,G_n$. At integer values $\Rephat\,G_N$ will differ from $\Rep\,G_N$ due to the presence of null objects, and by quotienting these we can recover the original category $\Rep\,G_N$.

 \subsection{$\Reptilde\,U(n)$}
 We will begin by describing the category $\Rephat\,U(n)$.\footnote{While physicists work with representations of $U(N)$, mathematicians instead usually consider holomorphic representations of $GL(N,\mathbb C)$. Since $GL(N,\mathbb C)$ is the complexification of $U(N)$, these are equivalent notions.} The objects in this category are finite strings $[s]$ of the characters $+$ and $-$, which we can think of diagrammatically as labeled points, e.g.~$[s]=[+-+--]$.
 
 The morphisms in $\Rephat\,U(n)$ are linear combinations of string diagrams. Each string diagram from $[s_1]\rightarrow[s_2]$ consists of arrows connecting characters from $[s_1]$ and $[s_2]$ pairwise, which have to start/end at $+/-$ in $[s_1]$ or at $-/+$ in $[s_2]$. We do not show pluses and minuses in the diagram as they can be reconstructed from the arrow directions. E.g., the following string diagram is a morphism from $[++-+]\rightarrow[++]$:
 \begin{equation}
\includegraphics[trim=0 0 0 0,scale=0.5]{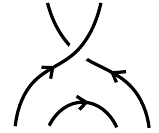}\,.\end{equation}
 We compose diagrams in the same way as for the Brauer algebra, by deleting the midpoints and replacing circles by $n$:
 \begin{equation}
 \includegraphics[trim=0 1.7em 0 0,scale=0.4]{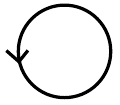}=  \includegraphics[trim=0 1.7em 0 0,scale=0.4]{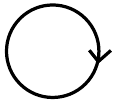}= n\,.
 \end{equation}\\[-0.5em]
 The algebra defined by these rules is called the walled Brauer algebra. 
 
 To relate the Brauer algebra to $\Rep\,U(N)$, we note that for integer $N$ there is a functor $\fcy F:\Rephat\,U(N)\rightarrow\Rep\,U(N)$. This maps objects $[s]$ into tensor products of the fundamental $\bf N$ and $\overline{\bf N}$ of $U(N)$, and string diagrams to $U(N)$ invariants. For instance:
 \begin{equation}
 \fcy F([+-++-]) \approx {\bf N}\otimes\overline{\bf N}\otimes{\bf N}^{\otimes 2}\otimes\overline{\bf N}\,.
 \end{equation}
 
 Much like in $\Rephat\,O(n)$, in $\Rephat\,U(n)$ we can decompose $\text{id}_{[s]}$ as the sum of mutually orthogonal idempotents. For instance, in $\text{Hom}([++]\rightarrow[++])$ there are two such idempotents, one symmetric and one antisymmetric:
 \begin{equation}
 P_{\bf S} = \frac 12\Bigl(\,\includegraphics[trim=0 1.7em 0 0,scale=0.4]{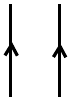}+
 \includegraphics[trim=0 1.7em 0 0,scale=0.4]{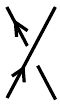}\,\Bigr)\,,\qquad P_{\bf A} = \frac 12\Bigl(\,\includegraphics[trim=0 1.7em 0 0,scale=0.4]{fig-PS1.pdf}-
 \includegraphics[trim=0 1.7em 0 0,scale=0.4]{fig-PS2.pdf}\,\Bigr)\,.
 \end{equation}
 Likewise, in $\text{Hom}([+-]\rightarrow[+-])$ there are also two idempotents:
 \begin{equation}
 P_{\bf 1} = \frac 1n\,
 \includegraphics[trim=0 1.7em 0 0,scale=0.4]{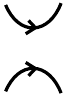}\,,\qquad P_{\bf a} = \includegraphics[trim=0 1.7em 0 0,scale=0.4]{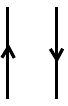}-\frac 1n\,
 \includegraphics[trim=0 1.7em 0 0,scale=0.4]{fig-Pa2.pdf}\,.
 \end{equation}
 We can define the trace of a diagram $f:[s]\rightarrow[s]$ by the analogue of Eq.~\reef{eq:traceo(n)}:
 \begin{equation}
 \label{eq:traceu(n)}
 \text{tr}(f)= \raisebox{-4em}{\includegraphics[scale=0.5]{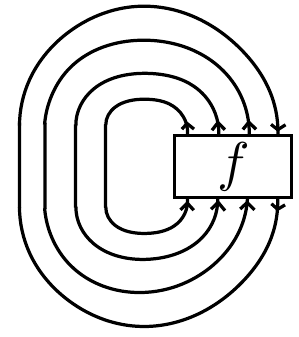}}\,,\end{equation}
 (connecting outgoing to the corresponding incoming arrow on the other side)
 and we then find that
 \begin{equation}
 \text{tr}(P_{\bf S}) = \frac{n(n+1)}2\,,\quad \text{Tr}(P_{\bf A}) = \frac{n(n-1)}2\,,\quad \text{tr}(P_{\bf 1}) = 1\,,\quad \text{tr}(P_{\bf a}) = n^2-1\,,
 \end{equation}
 reproducing for integer $n$ the dimensions of the corresponding representations in $\Rep\,U(n)$. After taking the Karoubi envelope and additive completion, we find that each of these idempotents corresponds to a simple object in $\Reptilde\,U(n)$.
 
 Finding the simple objects in $\Reptilde\,U(n)$ and computing their dimensions is now just an exercise in combinatorics. Many examples of diagrammatic computations for $U(N)$ can be found in chapter 9 of \cite{Cvitanovic:2008zz}, and these extend to the non-integer case without change. More generally, the simple objects in $\Reptilde\,U(n)$ can be labeled by pairs of partitions $\lambda = (\lambda_1,...,\lambda_r)$ and $\mu = (\mu_1,...,\mu_s)$. The dimensions of this simple object, which we can denote by ${\bf x}_{\lambda\mu}$, is given by the interpolation of the Weyl dimension formula:
 \begin{equation}\begin{split}
 \text{dim}({\bf x}_{\lambda\mu}) &= d_\lambda d_\mu \prod_{i = 1}^r\frac{(1-i-s+n)_{\lambda_i}}{(1-i+r)_{\lambda_i}}\prod_{i = 1}^s\frac{(1-i-r+n)_{\mu_i}}{(1-i+s)_{\mu_i}} \prod_{i = 1}^r\prod_{j = 1}^s \frac{n+1+\lambda_i+\mu_j-i-j}{n+1-i-j}\,, \\
 d_\lambda& = \prod_{1\leq i<j\leq r} \frac{\lambda_i-\lambda_j+j-i}{j-i}\,.
 \end{split}\end{equation}

 \subsection{$\Reptilde\,Sp(n)$ and negative dimensions}
 \label{sec:SPN}
 The group $O(N)$ is defined as the group of matrices preserving the unit matrix $\delta^{ab}$ which is symmetric. The group $Sp(N)$ is the group that preserves the antisymmetric symplectic matrix $\Omega^{ab}$. At the level of string diagrams this means that:
 \begin{equation}\label{eq:braidFlip}
 O(N):\, \includegraphics[trim=0 1.5em 0 0,scale=0.5]{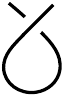} = 
 \includegraphics[trim=0 1.5em 0 0,scale=0.5]{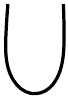}\,,\qquad Sp(N):\, \includegraphics[trim=0 1.5em 0 0,scale=0.5]{fig-br1.pdf}   = -\,\includegraphics[trim=0 1.5em 0 0,scale=0.5]{fig-br2.pdf}\,.
 \end{equation}\\[-0.5em]
 Indeed the l.h.s. diagrams are interpreted as the invariant tensor $\delta_{ab}$, resp.~$\Omega_{ab}$, contracted with the braiding $\includegraphics[trim=0 0.5em 0 0,scale=0.5]{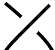}=\delta_a^{b'}\delta_b^{a'}$ and the sign on the right reflects symmetry/antisymmetry of the invariant tensor.
 
 A minus sign of the same origin then affects the value of the circle, because in $Sp(N)$
 \begin{equation} \Omega_{ab}\Omega^{ab}  = - N\,.
 \end{equation}
 In terms of string diagrams we have:
 \begin{equation}\label{eq:dimFlip}
 O(N):\, \includegraphics[trim=0 1.5em 0 0,scale=0.5]{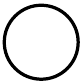} = \includegraphics[trim=0 2em 0 0,scale=0.5]{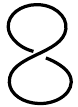} = N
 \,,\qquad Sp(N):\, \includegraphics[trim=0 1.5em 0 0,scale=0.5]{fig-br4.pdf} = - \includegraphics[trim=0 2em 0 0,scale=0.5]{fig-br3.pdf} = -N\,
 \,.
 \end{equation}
More generally, one can show that the algebra of string diagrams of $Sp(N)$ can be obtained from that of $O(N)$ by flipping the signs of $N$ and of the braiding: $N\to -N$, $\includegraphics[trim=0 0.5em 0 0,scale=0.5]{fig-braiding.pdf}\to - \includegraphics[trim=0 0.5em 0 0,scale=0.5]{fig-braiding.pdf}$. We now extend this construction to non-integer $n$ to construct $\Rephat\,Sp(n)$. For this we can start with $\Rephat\,O(-n)$ and then flip the sign of the braiding. More formally, we can write this as a functor $\hat{\fcy G}: \Rephat\,O(-n)\rightarrow\Rephat\,Sp(n)$ which takes:
 \begin{equation}
 \hat{\fcy F}(\includegraphics[trim=0 0.5em 0 0,scale=0.5]{fig-braiding.pdf}) = -\includegraphics[trim=0 0.5em 0 0,scale=0.5]{fig-braiding.pdf}\ ,
 \end{equation}
 but otherwise leaves the diagrams unchanged. The inverse functor $\hat{\fcy F}^{-1}:\Rephat\,Sp(n)\rightarrow\Rephat\,O(-n)$ acts in exactly the same way, and we conclude that $\Rephat\,O(-n)$ and $\Rephat\,Sp(n)$ are equivalent tensor categories.

 To define the trace on $\Rephat\,Sp(n)$ we will require that
 \begin{equation}\text{tr}(\text{id}_{[1]}) = n\,.\end{equation}
 Because of the additional minus sign in \eqref{eq:dimFlip}, the trace of a general morphism $f:[k]\rightarrow[k]$ is defined as
 \begin{equation}
 \text{tr}(f) = 
 (-1)^k  \includegraphics[trim=0 5.8em 0 0,scale=0.4]{fig-trace.pdf} = (-1)^k\text{tr}\left(\hat{\fcy F}^{-1}(f)\right)\,.
 \end{equation}\\[0.5em]
 We can now find all the idempotents in $\Rephat\,Sp(n)$ and compute their dimensions. Taking the Karoubi envelope and additive completion gives us $\Reptilde\,Sp(n)$.
 
 The Karoubi envelope and additive completion do not make use of the trace, and so we can lift $\hat{\fcy F}$ to a functor $\fcy F:\Reptilde\,O(-n)\rightarrow\Reptilde\,Sp(n)$. Therefore $\Reptilde\,O(-n)$ and $\Reptilde\,Sp(n)$ are equivalent tensor categories, differing only in some minus signs appearing in the braiding and the trace.
 
 To keep track of these minus signs, let us define for any simple object $\ba\in\Reptilde\,O(n)$ the $\mathbb Z_2$ grading:
 \begin{equation}
 Z_\ba = \begin{cases} 0 & \text{ if } \ba \in{\bf n}^{\otimes k} \text{ for some even }k\,, \\
 1 & \text{ if } \ba \in{\bf n}^{\otimes k} \text{ for some odd }k \,.
 \end{cases}
 \end{equation}
 In $\Rep\,O(N)$ this $\mathbb Z_2$ grading corresponds to the action of the $\mathbb Z_2$ center of $O(N)$ on the representation. We can then compute
 \begin{equation}\label{eq:FSigns}
 \fcy F(\beta_{{\bf a},{\bf b}}) = (-1)^{Z_{\bf a}Z_{\bf b}}\beta_{\fcy F(\ba),\fcy F({\bf b})}\,,\qquad \tr\left(\fcy F(f)\right) = (-1)^{Z_\ba}\tr(f)\,,\qquad \text{dim}\left(\fcy F(\ba)\right) = (-1)^{Z_\ba}\text{dim}(\ba)\,,
 \end{equation}
 where $\beta_{\ba,{\bf b}}:\ba\otimes{\bf b}\rightarrow{\bf b}\otimes\ba$ denotes the braiding, and where in each of these equations the left-hand side is computed in $\Reptilde\,Sp(n)$ and the right-hand side in $\Reptilde\,O(-n)$. More concretely, \eqref{eq:FSigns} means that for any $\Reptilde\,O(-n)$ morphism there is an equivalent $\Reptilde\,Sp(n)$ morphism where the symmetrizations and antisymmetrizations are swapped. A similar relationship holds between $\Reptilde\,U(n)$ and $\Reptilde\,U(-n)$.
 
 Relationships between representations of $Sp(N)$ and analytic continuations of $O(N)$ have been noticed many times in the literature, e.g.~\cite{king_1971,Penrose1971ApplicationsON,Parisi:1979ka,Cvitanovic:1982bq}. In particular theories of $N$ scalars with $O(N)$ symmetry are known to analytically continue at negative values of $N$ to theories of Grassmannian scalars with $Sp(N)$ symmetry, as studied by \cite{LeClair:2006kb,LeClair_2007} in a condensed matter setting and by \cite{Anninos:2011ui} in the context of $dS/CFT.$ The match between the $O(-1)$ and Grassmann path-integral in \eqref{eq:genNInt} is a simple case of this more general relationship.
 
 The categorical perspective makes it clear that any QFT with $\Reptilde\,O(-n)$ symmetry is equivalent to a $\Reptilde\,Sp(n)$ theory. The functor $\fcy F$ translate operators and morphisms from one language to the other. A bosonic operator $\phi$ with $Z_\phi = 1$ becomes a fermionic operator in the $\Reptilde\,Sp(n)$ theory, to compensate for the flipped braiding in \eqref{eq:FSigns}.

 \subsection{$\Reptilde\,S_n$}
 As a final task we shall construct Deligne categories $\Reptilde\, S_n$ interpolating representations of the symmetric group $S_N$. These categories describe the symmetry of the $n$-state Potts model for non-integer values of $n$. They also describe the symmetry of replicated theories. In particular, the replica trick relates the $n\rightarrow0$ limit of a deformed replicated theory to disorder averaging.
While all the other categorical symmetries discussed so far are continuous, this one is discrete (see 
section \ref{sec:current}).

 We can think of $S_N$ as the subgroup of $O(N)$ preserving the tensors
 \begin{equation}
 T_{I_1...I_k} = \begin{cases} 1 &\text{if } I_1 = I_2 = ... = I_k\,, \\ 0 & \text{otherwise}. \end{cases}
 \end{equation}
 for $k\ge 1$. (In particular the $k=1$ tensor $T_I=1$ for any index $I$.) These tensors are fully symmetric and satisfy the composition rules:
 \begin{equation}\label{eq:SNcomp}
 T_{I_1...I_kI_{k+1}...I_i}T_{J_1...J_kJ_{k+1}...J_j}\delta^{I_1J_1}...\delta^{I_kJ_k} = T_{I_{k+1}...I_i J_{k+1}...J_j}\,.
 \end{equation}
 To describe them diagrammatically, we start with the Brauer algebra, and introduce additional vertices:
 \begin{equation}
 T_{I_1...I_k} = \includegraphics[trim=0 1em 0 0,scale=0.4]{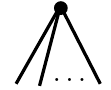}\quad(k\text{ legs})\,.
 \end{equation}
 Symmetry translates to the equation:
 \begin{equation}
\includegraphics[trim=0 1em 0 0,scale=0.4]{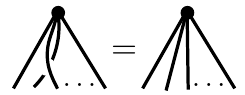}\quad(\text{for any pair of legs})\,,
 \end{equation}
 while \eqref{eq:SNcomp} becomes
 \begin{equation}
 \raisebox{-1.5em}{\includegraphics[trim=0 0 0 0,scale=0.4]{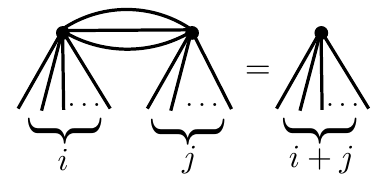}}\,.
 \end{equation}
 The final rule we need is that $ T^{IJ} = \delta^{IJ}$ which is encoded as
 \begin{equation}
  \raisebox{-0.2em}{\includegraphics[trim=0 0 0 0,scale=0.4]{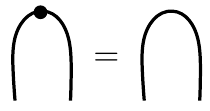}}\,.
 \end{equation}
 
 With these rules we perform any manipulations involving the $T_{I_1...I_k}$ tensors diagrammatically, with $N$ entering only as the value of the circle. As a simple example,
 \begin{equation}
  \includegraphics[trim=0 1.9em 0 0,scale=0.65]{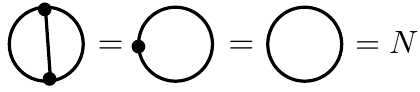}\,.
 \end{equation}
 We can then extend the algebra to continuous values of $n$, defining the category $\Rephat\,S_n$ to be the category of such string diagrams. The strings diagrams from $[k]\rightarrow[k]$ form an algebra known as the partition algebra $P_k(n)$, which has been studied in \cite{doi:10.1142/0983,doi:10.1142/S0218216594000071,Jones1993ThePM,HALVERSON2005869}. Much like how $\Rephat\,O(n)$ packages all of the Brauer algebras into a single algebraic structure, $\Rephat\,S_n$ packages all of the partition algebras together.
 
 As is usual, we can now compute any idempotents and their dimensions. Because any string diagram in $\Rephat\,O(n)$ is also a string diagram in $\Rephat\,S_n$, we see that there is a faithful\footnote{\label{note:faithful}I.e.~one which does not map any morphism to zero.} functor $\Rephat\,O(n)\rightarrow\Rephat\,S_n$. This is the categorical analogue of the fact that $S_N$ is a subgroup of $O(N)$, so that any $O(N)$ invariant tensors is automatically $S_N$ invariant. Due to the additional diagrams, indecomposable idempotents in $\Rephat\,O(n)$ will become decomposable in $\Rephat\,S_n$. For instance, the identity $\text{id}_{[1]}$ now splits
 \begin{equation}
 \text{id}_{[1]} = P_{\bf 1}+P_{\bf f}\,,\qquad P_{\bf 1} = \frac 1 n\ \raisebox{-1.3em}{\includegraphics[trim=0 0em 0 0,scale=0.4]{fig-EP.pdf}}\,,\qquad P_{\bf f} = \text{id}_{[1]} - P_{\bf 1}\,.
 \end{equation}
 
\subsection{Other families of categories}
\label{sec:otherfamilies}
We constructed $\Rephat\,S_n$ from $\Rephat\,O(n)$ by introducing new vertices into the string diagrams, along with new rules which allow us to eliminate any loops in the diagrams. This procedure can be performed very generally. Beginning with $\Rephat\,O(n)$ (or more generally with $\Rephat\,U(n)$), we can introduce new vertices into our string diagrams. We must then give rules telling us how to compose these diagrams, and in particular these tell us how to turn any closed diagram into a number. There are many possible diagrammatic rules one could consider, opening up a vast and largely unexplored world of possible categories. Here we will note some families which may be of interest in physical applications.
 
If we start with a theory with a symmetry $G$ and replicate it $N$ times, then each copy of the theory individually has a $G$ symmetry. The full symmetry of the replicated theory is then $S_N\ltimes G^N$. A construction due to Knop \cite{2006math5126K,2006math10552K} allows us to define a family of categories $\Reptilde\,S_n\ltimes G^n$, generalizing Deligne's construction of $\Reptilde\,G^n$. These categories describe the symmetries of theories such as the cubic fixed point, where $G \approx \mathbb Z_2$, and more generally the $MN$ models where $G = O(M)$. They are also relevant for the replicated theories arising when applying the replica trick to disordered systems, when the original theory has symmetry $G$.
 
All above-mentioned continuous families of categories interpolate infinite series of groups.
 
Notice however that \emph{not every} infinite family of groups admits a corresponding interpolating continuous family of Deligne categories. A necessary condition is that the invariant tensors have the same number of legs for all groups in the family. E.g.~$SO(N)$, compared to $O(N)$, has an additional invariant $N$-leg $\eps$-tensor, and this makes interpolation impossible.\footnote{If in spite of this argument, someone does want to try to define $SO(n)$, here is a preview of the kind of difficulties they will have to fight against. For each $k$ one would have to define an antisymmetric morphism $\eps:{\bf n}^{\otimes k}\rightarrow 1$. One would then have to somehow make $\eps_k$ vanish at any integer $n$ except at $n=k$. We don't have a no-go theorem that this is impossible, but this certainly does not look easy and we are not aware of mathematical literature accomplishing this.} A similar reason precludes interpolation for $SU(N)$ and $\bZ_N$. So it does not make sense to speak about $SO(N)$, $SU(N)$, and $\bZ_N$ symmetries for non-integer $N$. 
  
 On the other hand, one may also be wondering if there exist continuous families which interpolate \emph{finite} families of groups. There does not seem to be a fundamental reason precluding this, and attempts have been made to place the exceptional groups into a continuous family of categories. Many authors have noticed relationships between exceptional groups and other Lie group of low rank, the ``magic square'' construction of Freudenthal and Tits, and the ``magic triangle'' constructions of Cvitanovic. A detailed review can be found in \cite{Cvitanovic:2008zz}. These relationships led Deligne to conjecture a number of families of categories interpolating exceptional group \cite{DeligneLie,DeligneGross}, known as the $F_4$, $E_6$, $E_7$ and $E_8$ series. The $E_8$ series in particular contains all of the exceptional Lie algebras in series
 \begin{equation}
 \mathfrak{a}_1 \subset \mathfrak{a}_2 \subset \mathfrak{d}_4\subset \mathfrak{f}_4\subset \mathfrak{e}_6\subset \mathfrak{e}_7\subset \mathfrak{e}_8 \,.
 \end{equation}
 We note in passing that this series has appeared in a number of physical contexts, the $q$-state Potts model \cite{Dorey:2002tf}, in the theory of 2d chiral algebras associated to $4d$ $\fcy N = 2$ theories \cite{Beem:2013sza,Lemos:2015orc}, in $F$-theory \cite{Grassi:2000we,Shimizu:2016lbw}, and in the study of $2d$ RCFTs \cite{Mathur:1988na}. The $E_8$ series, along with Deligne series for the classical groups, have been further conjectured to combine into a two parameter family of tensor categories called the Vogel plane, introduced by Vogel in unpublished work \cite{Vogel99theuniversal}. The Vogel plane can be used to derive universal formulas for the representations of any simple Lie algebras, see e.g.~\cite{LANDSBERG2006379} and references therein.
 
 The difficulty with verifying these conjectures arises due to two possible issues:
 \begin{enumerate}
 	\item Are the rules complete, that is, do they allow us to reduce any closed diagram to numerical value?
 	\item Are the rules consistent, so that evaluating the same diagram in different ways gives the same answer?
 \end{enumerate}
 Unpublished work by Thurston \cite{dthurston} shows that the consistency of the $F_4$ and $E_6$ series requires a certain polynomial equation to be satisfied, and so these series can only exist at discrete values of the parameter. It is an open question whether or not similar issues occur for the $E_7$ and $E_8$ series or the Vogel plane. 
 
 Finally, we note that the representation theory of more general algebraic objects can be generalized to the categorical setting, including certain Lie superalgebras and affine Lie algebras \cite{ETINGOF2014,ETINGOF2016473,pakharev2019weylkac}. By generalizing the Schur-Weyl duality to $\Reptilde\,S_n$, one can even make sense of a ``non-integer tensor power'' of an object in a tensor category \cite{ETINGOF2014,10.1093/imrn/rnu214}.\footnote{This theory uses the notion of a unital vector space (vector space with a distinguished non-zero vector), for which one can define a non-integer tensor power as an object in the category $\Reptilde\,S_n$, and think of it as a vector space of non-integer dimension. Now, several times in this article we emphatically stated that vector spaces of non-integer dimensions do not exist. What we meant by that is that they don't ``exist'' in the usual physical sense, as sets of vectors which can be expanded in a basis etc. The above-mentioned constructions allow to give abstract algebraic meaning to some operations with ``non-integer dimensional vector spaces'', but they should not be viewed as to imply that such objects are on the same footing as concrete finite-dimensional vector spaces from linear algebra textbooks.}


 \section{Discussion and conclusions}
 \label{sec:conclusions}
 
As we have explained in this work, Deligne categories put on firm footing the notion of analytically continued symmetries usually considered in physics at an intuitive level. 
They are the algebraic structures which replace groups from the textbook definition of symmetry.
Looking back at section \ref{sec:need}, we can see how categories fulfil the symmetry wish list: $(1')$ simple objects replace irreducible representations;  ($2'$) morphism replace invariant tensors as the algebraic concept classifying correlators and transfer matrices; ($3'$) categorical symmetries are preserved under the Wilsonian RG (section \ref{sec:RG}), and thus can be used to classify universality classes with non-integer $n$. 

We thus dispelled a lot of conceptual fog, and explained the true meaning of computations in theories with analytically continued symmetries. Interestingly, results of computations performed in the usual intuitive way remain correct, when reinterpreted categorically. In particular one can eschew such voodoo notions as spaces of non-integer dimensions in the intermediate steps. The readers who have done such computations in the past no longer have to lose any sleep about the validity of their final answers.
 
We have illustrated how the categorical language works in many situations of interest to quantum field theory: perturbation theory, conserved currents, explicit and spontaneous symmetry breaking etc. 
We have also developed a theory of continuous categorical symmetries, and conjectured a natural categorical generalization of the Goldstone theorem.

We hasten to add that information provided by any symmetry, and categorical symmetry in particular, is mostly qualitative in nature. Let us take the critical point of the 3d $O(n)$ loop model as an example. Symmetry implies that critical exponents will be the same for different lattices and different lattice actions with the same symmetry. It also predicts the existence of the conserved current operator transforming as the adjoint object of the associated Deligne category and having scaling dimension $d-1=2$. But symmetry does not by itself fix values of other critical exponents. Such quantitative predictions would require explicit calculational techniques such as the RG or the conformal bootstrap.

Of course, categories nowadays appear in many branches of physics, e.g.~fusion categories are the language of topological QFT and of anyonic physics (excitations of topological states of matter). Our work adds another example where categories provide the needed language. Deligne categories have symmetric braiding and exhibit superexponential growth of the number of simple objects (see appendix \ref{sec:Deligne}). This makes them rather different from fusion categories, which have finitely many objects and non-trivial braiding.

In this paper we focused on global symmetries, like $O(n)$ with non-integer $n$. Such symmetries can arise both in perturbative context, and non-perturbatively (via loop models). Clearly, the language of Deligne categories would also work for spacetime symmetries, like $O(d)$ with non-integer $d$. Field theories in non-integer dimensions have been considered in physics since the seminal work of Wilson and Fisher \cite{Wilson:1971dc,Wilson:1972cf}, mostly in perturbation theory. Unlike for $O(n)$, it is not yet clear if these theories can be defined non-perturbatively (see \cite{Hogervorst:2014rta} for one attempt). This was one of the reason stopping us from presenting the corresponding theory (the other being the length of this paper).  We mention just some interesting parallels between the $O(n)$ and $O(d)$ stories. Theories in non-integer spacetime dimension should violate unitarity, as has already been discussed for free theories and in the $4-\epsilon$ expansion \cite{Hogervorst:2014rta,Hogervorst:2015akt}. Parity-violating theories should not allow analytic continuation in dimension, since there is no Deligne category for non-integer $SO(d)$ (section \ref{sec:otherfamilies}).

Conformally invariant theories may be easier to make sense for non-integer spacetime dimensions \cite{El-Showk:2013nia}. Indeed, the CFT four point function depends on two cross ratios for whatever the number of dimensions $d\ge 2$. Conformal blocks can be also analytically continued in $d$. Perhaps the language of Deligne categories will turn out useful in this setting, as was hypothesized in \cite{Isachenkov:2017qgn}.

Categorical language may turn out particularly useful when considering the analytic continuation of fermionic and supersymmetric theories to non-integer dimensions. As one potential application, consider the Gross-Neveu-Yukawa model. It is believed that this model with ``half'' a Majorana fermion in the $4-\epsilon$ expansion can be analytically continued to the $\fcy N = 1$ super-Ising model in 3d \cite{ThomasSem,Fei:2016sgs}. This is evidenced by the relation $\Delta_\psi = \Delta_\phi + \frac12$ among the scaling dimensions, checked up to $O(\epsilon^2)$ \cite{Fei:2016sgs} and suggestive of a supersymmetry. 
Whether there is a supersymmetry actually underlying this relation is however obscure, as both the number of spacetime dimensions and number of fermions is non-integer.\footnote{We are grateful to Pavel Etingof for forwarding to us the unpublished construction of categories involving spinorial representations by Pierre Deligne \cite{Deligne-letter}.}  Our categorical language should allow one to make precise the sense in which the model is supersymmetric, and in particular to prove the above relation to all orders in $\epsilon$.

We finish with a list of some of the open problems and questions to consider:
\begin{itemize}
\item Can one prove a categorical version of the Goldstone theorem that we conjectured, and develop an effective theory of categorical Goldstone bosons?

\item Can categorical symmetries be gauged? 

\item How do categorical symmetries interact with anomalies? More specific, massless QCD with $N_f$ massless Dirac fermions naively has a $U(N_f)\times U(N_f)$ symmetry, but this is broken to $SU(N_f)\times SU(N_f)\times U(1)$ by the axial anomaly. As mentioned, $SU(N)$ does not admit a continuation to non-integer $N$. Taken at face value, this suggests that massless QCD cannot be extended to non-integer $N_f$. On the other hand, there seems to be no obstruction for the extension of massive QCD with its $U(N_f)$ symmetry.

\item Can one formulate categorical symmetries in terms of topological surface operators? This seems challenging, especially for the discrete case.
 
\item What is the resolution of the apparent paradox mentioned in footnote \ref{note:paradox}?
\end{itemize}

 \section*{Acknowledgements}
 
 D.B.~would like to thank Scott Morrison, who first introduced him to the topic of tensor categories, and to Angus Gruen and Yifan Wang for useful discussions.
 
 S.R.~thanks Alexei Borodin, Matthijs Hogervorst, and Maxim Kontsevich from whom he first heard about the Deligne categories, as well as Pavel Etingof, Dennis Gaitsgory, Bernard Nienhuis, Yan Soibelman for related discussions. S.R.~is also grateful to Victor Gorbenko and Bernardo Zan for an earlier collaboration on the Potts model, which led to this project, and to Jesper Jacobsen and Hubert Saleur for many discussions about loop and cluster models in general.
 
 We thank Malek Abdesselam, John Cardy, Pavel Etingof, Victor Gorbenko, Misha Isachenkov, Hubert Saleur, and Bernardo Zan for comments on the draft. We thank Christopher Ryba for carefully reading the draft and pointing out misprints and inaccuracies.
 
 D.B.~is supported in part by the General Sir John Monash Foundation, and in part by Simon Foundation Grant No.~488653. S.R.~is supported by the Simons Foundation grant 488655 (Simons Collaboration on the Nonperturbative Bootstrap), and by Mitsubishi Heavy Industries as an ENS-MHI Chair holder.

 \appendix

 \section{Tensor categories}
 \label{sec: category}
 
 The purpose of this appendix is to allow interested reader to quickly pick up some information about category theory, before plunging into mathematical literature. Let us mention some useful resources. The classic book introducing category theory is \cite{maclane:71}; a modern textbook introduction to the basics can be found in \cite{leinster_2014}. Introductions to the theory of tensor categories, aimed at mathematicians, can be found in \cite{etingof2016tensor} and \cite{turaev2017monoidal}. More basic introductions can be found in the papers \cite{2009arXiv0903.0340B} and \cite{2009arXiv0905.3010C}. See \cite{2009arXiv0908.3347S} for a survey of string diagrams, reviewing the various different types of monoidal categories and their associated diagrams.
 
 \subsection{Basic definitions}
 Broadly speaking, a tensor category is a category which abstracts the notion of vector spaces and linear operators between them. This involves defining a number of structures, each of which corresponds to structures appearing the vector spaces:
 
 \begin{enumerate}
 	\item A \emph{category} abstracts the notion of function composition,
 	\item A \emph{monoidal} category is a category with a notion of a tensor product $\otimes$.
 	\item A \emph{braiding} on a monoidal category gives us morphisms from $\beta_{\bf a,b}: \bf a\otimes b\rightarrow b\otimes \bf a$. A \emph{symmetric} braiding is one such that performing the braiding twice gives us the identity..
 	\item A \emph{rigid} category is a category where every object $\bf a$ has a dual $\overline{\bf a}$.
 	\item A $\mathbb C$-\emph{linear} category is one for which the set of morphisms is a vector space and for which morphism composition is linear.
 	\item An \emph{semisimple} category is one where we have a direct sum $\oplus$, and for which objects can be decomposed into a finite sum of \emph{simple} objects.
 \end{enumerate}
 A \emph{tensor category} is a $\mathbb C$-linear semisimple rigid monoidal category, where $\bf 1$ is simple. A \emph{symmetric} tensor category is a tensor category with a symmetric braiding. We should note that there is no canonical definition of tensor category in the literature. Our definition is a stronger condition than that given in Chapter 4 of \cite{etingof2016tensor}, where for convenience we require our categories to be semisimple rather than abelian. 
 We will give more detailed explanations of each of these terms below. 
 
 \subsubsection{Categories}
 A \emph{category} $\fcy C$ consists of:
 \begin{enumerate}
 	\item A collection of \emph{objects}, ${\bf a},{\bf b},\ldots $
 	\item For any objects ${\bf a}, {\bf b}$, a set $\Hom({\bf a}\to{\bf b})$ of \emph{morphisms}, denoted by arrows $f:{\bf a}\rightarrow{\bf b}.$
 	\item Given morphisms $f:{\bf a}\rightarrow{\bf b}$ and $g:{\bf b}\rightarrow{\bf c}$ there is a way to compose them $\circ$, to create a new morphism $(g\circ f):{\bf a}\rightarrow{\bf c}.$
 	\item For every object ${\bf a}$ there is an \emph{identity} morphism $\text{id}_{\bf a}:{\bf a}\rightarrow{\bf a}$.
 \end{enumerate}
 Morphism composition is associative, and the identity morphisms composes trivially:
 \begin{equation}(f\circ g)\circ h = f\circ(g\circ h)\,,\quad \text{ and }\ f\circ\text{id}_{\bf a} = f = \text{id}_{\bf a}\circ f.\end{equation}
 
 A simple example of a category is $\mathsf{Vec}(\mathbb C)$, the category of complex vector spaces. The objects are the vector spaces $\mathbb C^n$ for $n = 1,2,3\dots$ and the morphisms between $\mathbb C^n\rightarrow \mathbb C^m$ are $m\times n$ matrices. We compose these morphisms using matrix multiplication. The identity morphisms on $\mathbb C^n$ is just the $n\times n$ identity matrix.
 
 The finite-dimensional complex representations of a group $G$ form a category $\Rep\,G$, where the objects are representations and the morphisms are covariant maps between representations. 
 
 A morphism $f:{\bf a}\rightarrow {\bf b}$ is an \emph{isomorphism} there exists an inverse morphisms $f^{-1}:{\bf b}\rightarrow{\bf a}$ such that $f^{-1}\circ f = \text{id}_{\bf a}$ and $f\circ f^{-1} = \text{id}_{\bf b}$. Two objects are \emph{isomorphic} if there is an isomorphism between them.
 
 A \emph{functor} generalizes the notion of group homomorphisms to categories. More precisely, given two categories $\fcy C$ and $\fcy D$, a functor $F:\fcy C\rightarrow \fcy D$ is a associates to every object $\bf a\in\fcy C$ an object $F(\bf a)\in\fcy D$ and every morphism $f:\bf a\rightarrow\bf b$ a morphisms $F(f):F(\bf a)\rightarrow F(\bf b)$, such that
 \begin{equation}F(f\circ g) = F(f)\circ F(g)\,,\quad\text{ and }F(\text{id}_{\bf a}) = \text{id}_{F({\bf a})}.\end{equation}
 
 As an example, given two groups $H\subset G$ there is a functor $F:\Rep(G)\rightarrow\Rep(H)$. This functor takes representation ${\bf a}\in \Rep(G)$ and turns it into the restricted representation $F({\bf a)}\in\Rep(H)$. Since any morphism $f:{\bf a}\rightarrow{\bf b}$ preserves $G$, they also preserve the restricted representation on the subgroup $H$ and so $F(f)$ is a morphism in $\Rep(H).$ 

 A functor $F:\fcy C\rightarrow\fcy D$ is $\emph{faithful}$ if for every $\ba$ and $\bf b$ in $\fcy C$ the function $F:\Hom(\ba\rightarrow{\bf b})\rightarrow\Hom\left(F(\ba)\rightarrow F({\bf b})\right)$ is injective. The functor is \emph{full} if for every $\ba$ and $\bf b$ in $\fcy C$ the function $F:\Hom(\ba\rightarrow{\bf b})\rightarrow\Hom\left(F(\ba)\rightarrow F({\bf b})\right)$ is surjective.

 \subsubsection{Linear categories}
 \label{sec:LINCAT}
 A \emph{linear} category is one for which the sets of morphisms $\Hom({\bf a}\rightarrow{\bf b})$ are vector spaces (for our purposes, over $\mathbb C$) and for which morphism composition is bilinear. The category $\mathsf{Mat}(\mathbb C)$ is clearly linear, as is $\Rep\,G.$
 
 A linear category $\fcy C$ is \emph{additive} if there is some abstract notion of a direct sum $\oplus$. More precisely, for a pair of objects ${\bf a}$ and ${\bf b}$, the \emph{direct sum}, if it exists, is defined to be an object ${\bf a}\oplus{\bf b}$ such that:
 \begin{enumerate}
 	\item There exists \emph{embedding} morphisms $\iota_1:{\bf a}\rightarrow {\bf a}\oplus{\bf b}$ and $\iota_2:{\bf b}\rightarrow {\bf a}\oplus{\bf b}$
 	\item There exists \emph{projection} morphisms $\pi_1:{\bf a}\oplus{\bf b}\rightarrow{\bf a}$ and $\pi_2:{\bf a}\oplus{\bf b}\rightarrow{\bf b}$
 	\item These maps satisfy the equations
 	\begin{equation}\pi_1\circ \iota_1 = \text{id}_{\bf a}\,,\quad \pi_2\circ \iota_2 = \text{id}_{\bf b}\,,\quad \iota_1\circ \pi_1 + \iota_2\circ \pi_2 = \text{id}_{\bf a\oplus \bf b}.
 	\end{equation}
 \end{enumerate}
 If such an object exists it is unique up to unique isomorphism (see section 8.2 of \cite{maclane:71}).
 
 This definition of a direct sum may look a little abstract, so let us unpack it for the case of  $\mathsf{Mat}(\mathbb C)$. In this category the direct sum of $\mathbb C^m$ and $\mathbb C^n$ is the object $\mathbb C^m \oplus \mathbb C^n = \mathbb C^{m+n}$. In this case the embedding and projection morphisms are
 \begin{equation}\begin{aligned}
 \iota_1 &= \begin{pmatrix} I_{m\times m} \\ 0_{n\times m}\end{pmatrix}\,,&\quad \iota_2 &= \begin{pmatrix} 0_{m\times n} \\ I_{n\times n}\end{pmatrix} \\
 \pi_1 &= \begin{pmatrix} I_{m\times m} &  0_{m\times n}\end{pmatrix}\,,&\quad \pi_2 &= \begin{pmatrix} 0_{n\times m} &  I_{n\times n}\end{pmatrix},\\
 \end{aligned}\end{equation}
 where $I_{m\times m}$ is the $m\times m$ identity matrix and $0_{n\times m}$ is the $n\times m$ matrix where all entries are $0$. 
 
 A linear\footnote{This definition, and the definition of semisimplicity which follows, can be extended to more general categories where the $\Hom$-sets are merely required to be abelian groups rather than vector spaces. For simplicity we will only consider linear categories; more general definitions can be found in Chapter 1 of \cite{etingof2016tensor}.} category is additive if for each pair of objects $\bf a$ and $\bf b$ there exists a direct sum $\bf a\oplus\bf b$, and if there exists a zero object ${\bf 0}\in\fcy C$ such that $\text{Hom}({\bf 0}\rightarrow{\bf 0}) = 0$.
 
 Given a linear category $\fcy C$, one can always construct a new category $\fcy C^{\text{add}}$ which is additive by formally adding objects ${\bf a}\oplus{\bf b}$ and ${\bf a}\oplus{\bf b}\oplus{\bf c}$ and so on to $\fcy C$. This procedure is the \emph{additive completion} we made use of to construct $\Reptilde\,O(n)$ in section \ref{sec:reptON}. It is a very general construction, and is detailed in section 16.2 of \cite{turaev2017monoidal}. The upshot is that we do not lose much generality by restricting ourself to additive categories.

 An additive category $\fcy C$ is \emph{semisimple} if:
  \begin{enumerate}
 	\item There exists a set of \emph{simple objects} $\ba_i$ such that $\Hom(\ba_i\rightarrow\ba_i)$ is a one-dimensional vector space, and $\Hom(\ba_i\rightarrow\ba_j)$ for $i\ne j$ contains only the $0$ morphism
 	\item Every object in $\fcy C$ is the direct sum of simple objects
 \end{enumerate}

 Since for a simple object $\ba$ the space $\Hom(\ba\rightarrow\ba)$ is one dimensional every morphism $f:\ba\rightarrow\ba$ is of the form $\lambda\,\text{id}_{\bf a}$ for some $\lambda\in\mathbb C.$ In $\Rep\,G$ the simple objects are precisely the irreducible representations. For a finite or compact Lie group $G$, $\Rep\,G$ is always semisimple. For non-compact groups there can be representations which are indecomposable but not irreducible, and this spoils semisimplicity. These categories are still \emph{abelian}\footnote{A definition of what this means can again be found in Chapter 1 of \cite{etingof2016tensor}; it is a stronger notion than additive but weak than semisimplicity.} and many results extend naturally to this more general case. We will however focus on semisimple categories in this paper.

 \subsubsection{Monoidal categories and braidings}
 \label{sec:MONOID}
 A \emph{monoidal category} is a category with a tensor product. More precisely, it is a category with: 
 \begin{enumerate}
 	\item A functor $\otimes: \fcy C\times \fcy C\rightarrow \fcy C$. This means that we can form tensor product ${\bf a}\otimes{\bf b}$ of any two objects ${\bf a}$ and ${\bf b}$. Also, given morphisms $f: {\bf a}\rightarrow{\bf b}$ and $g: {\bf c}\rightarrow{\bf d}$ we can form their tensor product: $f\otimes g:{\bf a}\otimes{\bf c}\rightarrow{\bf b}\otimes{\bf d}.$ Being a function, $\otimes$ should respect morphism composition:
 	\begin{equation}(f_1\circ f_2)\otimes(g_1\circ g_2) = (f_1\otimes g_1)\circ(f_2\otimes g_2)\end{equation}
 	\item A special object ${\bf 1}\in\fcy C$.
 \end{enumerate}
 The tensor product is required to be associative, and $\bf 1$ acts as the identity:\footnote{Here we are only describing a \emph{strict} monoidal category. More generally the associativity and identity hold only up to a unique isomorphism, and additional axioms need to be introduced. We will not worry about such subtleties here.}
 \begin{equation}({\bf a}\otimes{\bf b})\otimes{\bf c} = {\bf a}\otimes({\bf b}\otimes{\bf c})\,,\quad \text{ and }{\bf a}\otimes{\bf 1} = {\bf a} = {\bf 1}\otimes{\bf a}.\end{equation}
 On a linear category, we require the tensor product to be bilinear. Both $\mathsf{Mat}(\mathbb C)$ and $\Rep\,G$ form linear monoidal categories with the usual tensor products. 
 
 A \emph{braiding} in a monoidal category relates ${\bf a}\otimes{\bf b}$ to ${\bf b}\otimes{\bf a}.$ More precisely, it is a family of isomorphisms
 \begin{equation}\beta_{{\bf a},{\bf b}}:{\bf a}\otimes{\bf b}\rightarrow{\bf b}\otimes{\bf a}\end{equation}
 satisfying the identities\footnote{The first equation means, in plain language, that to braid ${\bf a}$ with ${\bf b}\otimes{\bf c}$ we can first braid ${\bf a}$ with $\bf b$ doing nothing to $\bf c$ and then with $\bf c$ doing nothing to $\bf b$. The second is analogous.}
 \begin{equation}
 \beta_{{\bf a},{\bf b}\otimes{\bf c}} = (\text{id}_{\bf b}\otimes\beta_{{\bf a},{\bf c}})\circ(\beta_{{\bf a},{\bf b}}\otimes\text{id}_{\bf c})\,,\quad \beta_{{\bf a}\otimes{\bf b},{\bf c}} = (\beta_{{\bf a},{\bf c}}\otimes\text{id}_{\bf b})\circ(\text{id}_{\bf a}\otimes\beta_{{\bf b},{\bf c}})\,,
 \end{equation}
 and for any pair of morphisms $f:{\bf a}\rightarrow{\bf b}$ and $g:{\bf c}\rightarrow{\bf d}$,
 \begin{equation}
 \beta_{{\bf b},{\bf d}}\circ(f\otimes g) = (g\otimes f)\circ\beta_{{\bf a},{\bf c}}\,.
 \end{equation}
 The braiding is \emph{symmetric} if 
 \begin{equation}
 \beta_{{\bf b},{\bf a}}\circ\beta_{{\bf a},{\bf b}} = \text{id}_{{\bf a}\otimes{\bf b}}.
 \end{equation}
 A monoidal category with a braiding (resp.~symmetric braiding) is called \emph{braided} (resp.~\emph{symmetric}).
 
 Both $\mathsf{Mat}(\mathbb C)$ and $\Rep\,G$ have symmetric braidings. In $\mathsf{Mat}(\mathbb C)$, fix some basis ${\bf e}^{(k)}_i$ of $\mathbb C^k$ and write vectors in $\mathbb C^{m}\otimes \mathbb C^n$ in the form
 \begin{equation}
 v = \sum_{i=1}^m\sum_{j = 1}^n v_{ij}\ {\bf e}^{(m)}_i\otimes {\bf e}^{(n)}_j
 \end{equation}
 The braided vector $\beta_{m,n}(v)\in \mathbb C^{n}\otimes \mathbb C^m$ then corresponds to transposing the matrix $(v_{ij})$. 
 The braiding on $\Rep\,G$ can be defined likewise.

 \subsubsection{Symmetrizing and antisymmetrizing}
 \label{sec:sym}
 
 Let $\bf a$ be an object in a symmetric tensor category $\fcy C$, and consider the tensor product ${\bf a}^{\otimes k}$. For any $i<k$ we can define the morphism
 \begin{equation}\sigma_i = \text{id}_{\bf a}^{\otimes i-1}\otimes\beta_{\bf a,\bf a}\otimes\text{id}_{\bf a}^{\otimes k-i-1}\end{equation}
 which interchanges the $i$ and $i+1$ copies of $\bf a$ in the tensor product. It follows from the braiding axioms that the operators $\sigma_i$ satisfy the relations
 \begin{equation}\begin{split}
 \sigma_i^2 &= 1 \\
 \sigma_i\sigma_j &= \sigma_j\sigma_i \text{ for } |i-j|>1 \\
 \sigma_i\sigma_{i+1}\sigma_{i} &= \sigma_{i+1}\sigma_i\sigma_{i+1}\\
 \end{split}\end{equation}
 which are the generating relations for the group symmetric group $S_k$. Thus in any symmetric tensor category we have a morphism $S_k\rightarrow\Hom({\bf a}^{\otimes k}\rightarrow{\bf a}^{\otimes k}).$ We can then decompose the space $\Hom({\bf a}^{\otimes k}\rightarrow{\bf a}^{\otimes k})$ into representations of $S_k$, so that all morphisms can be classified by their symmetry properties. We can for instance define the fully symmetric morphisms to be those that that transform trivially under $S_k$ and the fully antisymmetric morphisms to be those that change sign under the braiding. This generalizes the notion of symmetric and antisymmetric tensors. 
 
 By classifying idempotent morphisms by their symmetry properties under $S_k$, we can classify the simple objects ${\bf b}_i$ in
 $${\bf a}^{\otimes k} = \bigoplus_{i = 1}^K {\bf b}_i$$
 by their symmetry properties as well. The collection of all simple objects which transform trivially under $S_k$ together form an (in general not simple) object in ${\bf a}^{\otimes k}$ which is called the \emph{symmetrized product} $S^k({\bf a})$. We can likewise define the \emph{antisymmetrized} product $\Lambda^k({\bf a})$ by restricting to objects transforming in the fully antisymmetric representation of $S_k$.
 
 \subsubsection{Tensor functors}
 If we have categories with extra structure, then it is natural to define special functors which preserve these structures. An \emph{additive functor} is a functor which preserves direct sums and $\bf 0$. A \emph{strict monoidal functor} is a functor which preserves tensor products and $\bf 1$. A \emph{strict braided monoidal functor} is a strict monoidal functor which furthermore preserves the braiding. We shall define a \emph{strict braided tensor functor} to be a functor which is both strict braided monoidal and additive.

 \subsection{Rigidity, traces and dimensions}
 \label{sec:RIG}
 
 In this section we will now define \emph{rigidity} in the context of a symmetric monoidal category. Given an object $\ba$ in a monoidal category, the object $\overline\ba$ is \emph{dual} to $\ba$ if there exists morphisms $\delta^{\overline\ba,\ba}:\overline\ba\otimes\ba\rightarrow{\bf 1}$ and  $\delta_{\ba,\overline\ba}:{\bf 1}\rightarrow\ba\otimes\overline\ba$, which satisfy the zig-zag relations
 \begin{equation}\label{eq:zigzag}
 ({\bf 1}\otimes \delta^{\overline\ba,\ba})\circ (\delta_{\ba,\overline\ba}\otimes {\bf 1})= \text{id}_{\ba}\,,\qquad (\delta^{\overline\ba, \ba}\otimes {\bf 1})\circ ({\bf 1}\otimes \delta_{\ba,\overline\ba})=\text{id}_{\overline\ba}\,.
 \end{equation}
 These morphisms are called the \emph{cap} and \emph{cocap} for their graphical representations:
 $$\delta^{\overline\ba,\ba}=\includegraphics[trim=0 1.5em 0 0, scale=0.4]{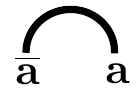}\qquad \text{and}\qquad \delta_{\ba,\overline\ba}=\includegraphics[trim=0 1.5em 0 0,scale=0.4]{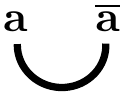}\,.$$
 The first zig-zag relation can be expressed graphically as:
 \begin{equation}
\includegraphics[trim=0 3em 0 0, scale=0.3]{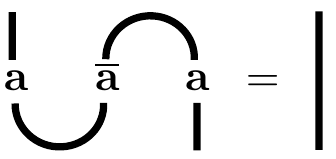},
\end{equation}
and an analogous expression holds for the second zig-zag relation.

 If $\overline\ba$ is dual to $\ba$, then $\ba$ is dual to $\overline\ba$. This is a consequence of the braiding, for we can define morphisms:
 \begin{equation}
 \delta_{\overline \ba,\ba} = \delta_{\ba,\overline\ba}\circ \beta_{\overline\ba,\ba}\,,\quad \delta^{\ba,\overline\ba} = \beta_{\overline\ba,\ba} \circ \delta^{\overline \ba,\ba}\,,
 \end{equation}
 which also satisfy the zig-zag relations, but with $\ba$ and $\overline\ba$ swapped. Proving this is an exercise in graphical manipulations. 
 
 It can be proven (see e.g.~section 2.10 of \cite{etingof2016tensor}) that for a given object $\ba$ its dual $\overline\ba$ is unique up to unique isomorphism. For this reason we can safely talk about \emph{the} dual of an object, without having to worry about which dual we are actually talking about.
 
 A symmetric monoidal category is \emph{rigid} if every object has a dual.
 
 \subsubsection{Dimensions and traces}
 Let $\fcy C$ be a rigid symmetric tensor category, and $\bf a$ an object in $\fcy C$. For any morphism $f:\bf a\rightarrow\bf a$, we define the \emph{trace}:
 \begin{equation}\label{eq:trDef}
 \text{tr}_{\bf a}(f) = \delta^{\overline\ba,\ba} \circ (\text{id}_{\overline\ba}\otimes f) \circ \delta_{\overline\ba,\ba} = \includegraphics[trim=0 5em 0 0, scale=0.35]{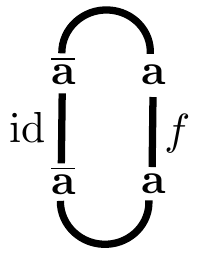}\,.
 \end{equation}\\[-0.5em]
 Using string diagrams, the reader may convince themselves that an equivalent definition is
 \begin{equation}
 \text{tr}_{\bf a}(f) = \delta^{\ba,\overline\ba} \circ (f\otimes\text{id}_{\overline\ba}) \circ \delta_{\ba,\overline\ba}\ , 
 \end{equation}
 so that our choice to place $f$ on the right, rather then the left, in \eqref{eq:trDef} was unimportant.
 
 \begin{prop} The trace satisfies the following identities:
 	\begin{enumerate}
 		\item 
 		$\mathrm{tr}_{\bf a}(\lambda f+\mu g) = \lambda\mathrm{tr}_{\bf a}(f) + \mu \mathrm{tr}_{\bf a}(g)\text{ for } f,g:\ba\rightarrow\ba\text{ and }\lambda,\mu\in\mathbb C$
 		\item $\mathrm{tr}_{\bf a\otimes \bf b}(f\otimes g)  = \mathrm{tr}_{\bf a}(f)\mathrm{tr}_{\bf b}(g) \text{ for } f:\ba\rightarrow{\bf a} \text{ and } g:{\bf b}\rightarrow{\bf b}$
 		\item $\mathrm{tr}_{\bf a}(g\circ f) = \mathrm{tr}_{\bf b}(f\circ g) \text{ for } f:\ba\rightarrow{\bf b} \text{ and } g:{\bf b}\rightarrow\ba$
 \end{enumerate}\end{prop}
 Each of these properties is straightforward to prove, and they generalize the usual linearity and cyclicity properties of the matrix trace. Further, we define the \emph{dimension} of $\bf a$ to be $\text{dim}({\bf a}) = \text{tr}_{\bf a}(\text{id}_{\bf a}).$ This generalizes the usual notion of the dimension of a representation.
 
 \begin{prop}\label{pr:dimForm}We have the following identities:
 	\begin{equation}
 	\mathrm{dim}({\bf a}\otimes{\bf b}) = \mathrm{dim}({\bf a})\mathrm{dim}({\bf b})\,,\quad \mathrm{dim}({\bf a}\oplus{\bf b}) = \mathrm{dim}({\bf a})+\mathrm{dim}({\bf b})\,,\quad \mathrm{dim}(\overline\ba) = \mathrm{dim}(\ba)\,.
 	\end{equation}
 \end{prop}
 
 \subsubsection{Implications of semisimplicity}
 \label{sec:SEMSIM}
 Semisimplicity places many non-trivial constraints on a tensor category. We will now explore a few non-trivial consequences of semisimplicity.
 
 \begin{prop}\label{pr:zeroDim} For any simple object $\ba$ in a symmetric tensor category $\fcy C$, $\mathrm{dim}(\ba)\neq 0$.
 \end{prop}
 \begin{proof} 
 	Because $\ba$ is simple, $\Hom(\ba\rightarrow\ba)$ is one-dimensional. From this it follows that both
 	${\Hom({\bf 1}\rightarrow\ba\otimes\overline\ba)}$ and
 	${\Hom(\ba\otimes\overline\ba\rightarrow{\bf 1})}$ 
 	are also one-dimensional. Using semisimplicity we then find that
 	\begin{equation} 
 	\ba\otimes\overline\ba = {\bf 1}\oplus...
 	\end{equation}
 	where the additional simple objects are not isomorphic to $\bf 1$. Since both ${\Hom({\bf 1}\rightarrow\ba\otimes\overline\ba)}$ and
 	${\Hom(\ba\otimes\overline\ba\rightarrow{\bf 1})}$ are one-dimensional, we can deduce that
 	\begin{equation}
 	\delta_{\ba,\overline\ba} = \lambda\iota_{\bf 1}\,,\quad \delta^{\ba,\overline\ba} = \mu\pi_{\bf 1}
 	\end{equation}
 	for non-zero $\lambda,\mu\in\mathbb C$. We can then use this to compute:
 	\begin{equation}
 	\text{dim}(\ba) = \text{tr}_{\ba}(\text{id}_\ba) = \delta^{\ba,\overline\ba}\circ\delta_{\ba,\overline\ba} = \lambda\mu (\pi_{\bf 1}\circ \iota_{\bf 1}) = \lambda\mu 
 	\end{equation}
 	which is non-zero.
 \end{proof}
 
 \begin{prop}
 	If ${\bf b} \approx \ba_1\oplus...\oplus\ba_n$ then $\overline{\bf b}\approx \overline{\ba_1}\oplus...\oplus\overline{\ba_n}$.
 \end{prop}
 \begin{proof} It suffices to prove this for the simpler case ${\bf b}\approx\ba_1\oplus\ba_2$. In this case we have a pair of morphisms
 	\begin{equation}
 	\pi_k : {\bf b}\rightarrow\ba_k\,,\quad \iota_k : \ba_k\rightarrow{\bf b}\,
 	\end{equation}
 	for $k = 1,2$, satisfying the conditions defining the direct sum. Let us now define
 	\begin{equation}
 	\overline\pi_k = \raisebox{-2em}{\includegraphics[scale=0.45]{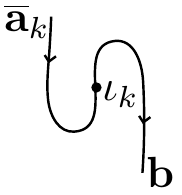}}\,,\quad \overline\iota_k = \raisebox{-2em}{\includegraphics[scale=0.45]{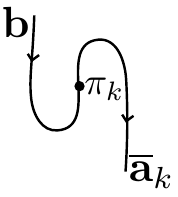}}\,.
 	\end{equation}
 	It is easy to verify that these maps also satisfy the direct sum conditions, and so $\overline{\bf b}\approx \overline{\ba_1}\oplus\overline{\ba_2}$.
 \end{proof}

 Given any two morphisms $f:\ba\rightarrow{\bf b}$ and $g:{\bf b}\rightarrow\ba$, we can define the bilinear pairing
 \begin{equation}
 \langle g,f\rangle = \text{tr}_{\ba} (g\circ f) = \text{tr}_{\bf b}(f\circ g)\,.
 \end{equation}
 
 \begin{prop}\label{pr:ndeg}
 	The pairing $\langle\cdot,\cdot\rangle$ on $\mathrm{Hom}(\ba\rightarrow{\bf b})\times\mathrm{Hom}({\bf b}\rightarrow\ba)$ is non-degenerate.
 \end{prop}
 \begin{proof}
 	For any simple object $\bf c$ which appears in the decomposition of an object $\ba$ we can choose a basis of projection and embedding morphisms $\iota^i_{\ba,\bf c}$ and $\pi^i_{\ba,\bf c}$ satisfying
 	\begin{equation}\pi^i_{\ba,\bf c}\circ \iota^j_{\ba,\bf c} =\delta^{ij} \text{id}_{\bf c}.\end{equation}
 	If $\bf c$ also appears in $\bf b$ we can likewise choose a basis of morphisms $\iota^i_{\bf b,\bf c}$ and $\pi^i_{\bf b,\bf c}$. We then construct a basis of $\Hom(\ba\rightarrow\bf b)$ comprising of morphisms $P^{ij}_{\bf c} = \iota_{{\bf b},\bf c}^i\circ\pi_{\ba,\bf c}^j$ and likewise a basis of $\Hom({\bf b}\rightarrow\ba)$ of morphisms $Q^{ij}_{\bf c} = \iota_{{\ba},\bf c}^i\circ\pi_{{\bf b},\bf c}^j$. Now we compute the pairing 
 	\begin{equation}
 	\langle Q^{ji}_{\bf c},P^{i'j'}_{\bf c}\rangle = \text{tr}_{\ba}(Q^{ji}_{\bf c}\circ P^{i'j'}_{\bf c}) = \delta^{ii'}\text{tr}_{\bf a}(\iota^j_{\ba,{\bf c}}\circ\pi^{j'}_{\ba,{\bf c}}) = 
 	\delta^{ii'}\text{tr}_{\bf c}(\pi^{j'}_{\ba,{\bf c}}\circ\iota^j_{\ba,{\bf c}}) = \delta^{ii'}\delta^{jj'}\text{dim}({\bf c}).\, 
 	\end{equation}
 	Since $\text{dim}({\bf c})\ne 0$ is always non-zero by proposition \ref{pr:zeroDim}, the pairing $\langle\cdot,\cdot\rangle$ is non-degenerate. \end{proof}

 \subsubsection{A note on definitions}
 \label{sec:DEFNOTE}
 
 As in the paper we only consider symmetric tensor categories, we have defined rigidity, trace and dimensions to be as simple as possible in this context. In a more general tensor category, if \eqref{eq:zigzag} is satisfied by $\ba$ and $\overline\ba$ then we say that $\overline\ba$ is a left dual to $\ba$, and that $\ba$ is a right dual to $\overline\ba$. Unlike the symmetric case, in a general tensor category the left dual of $\ba$ (which we can denote $\ba^-$) and the right dual of $\ba$ (which we can denote $^-\ba$) do not have to be isomorphic.
 
 A \emph{pivotal} tensor category is a tensor category with a series of isomorphisms $^-\ba\rightarrow\ba^-$ satisfying various naturalness conditions. This pivotal structure can be used to define left and right traces; if these traces are always equal we say that the tensor category is \emph{spherical}. Even on symmetric tensor categories there may be multiple inequivalent pivotal structures, and hence inequivalent traces and dimensions. In some applications (such as to 3d TQFTs) distinguishing between different structures may be important.
 
 In a symmetric tensor category, the braiding gives rise to a \emph{canonical} pivotal structure which is automatically spherical. The potential existence of other pivotal structures is not important for our purposes.

 \subsection{Further facts about Deligne categories}
 \subsubsection{Positive tensor categories}
 \label{sec:PosCat}
 A \emph{positive} symmetric tensor category is one for which all objects have non-negative dimension. For any group $G$ the category $\Rep\,G$ is clearly positive. On the other hand, as we have seen already, Deligne categories are in general not positive. To see a general reason, consider some object $\ba$ in a symmetric tensor category $\fcy C$ with dimension $\mathrm{dim}(\ba) = \alpha$. We can then compute that the dimension of the antisymmetrization $\Lambda^k(\ba)$ is
 \begin{equation}\label{eq:asymDim}
 \mathrm{dim}(\Lambda^k(\ba)) = \frac{\alpha(\alpha-1)...(\alpha-k+1)}{k!}\,.
 \end{equation}
 If $\alpha$ is not an integer than this quantity will become negative for some sufficiently large values of $k$. So we conclude that all dimensions must be integers in a positive category, which excludes Deligne categories apart from the integer parameter case when they coincide with $\Rep\,G$.

\subsubsection{Universal properties of $\Reptilde\,U(n)$, $\Reptilde\,O(n)$, $\Reptilde\,Sp(n)$ and $\Reptilde\,S_n$}
Irreducible representations of a group can be classified as real, pseudoreal, or complex. This classification generalizes to simple objects in a symmetric tensor category. Since the dual of a simple object is itself simple, either ${\Hom(\ba\rightarrow\overline\ba)}$ is zero-dimensional, or there exists an isomorphism $i_{\ba}:\ba\rightarrow\overline\ba$, unique up to a scale factor. In the latter case, we can define a morphism
\begin{equation}
B^{\ba,\ba} \equiv \delta^{\ba,\overline\ba}\circ(\text{id}_\ba\otimes i_{\ba})\in \text{Hom}(\ba\otimes\ba \rightarrow\bf 1)\,.
\end{equation}
Going the other way, given a $B^{\ba,\ba}$ we can construct $i_{\ba}=(B^{\ba,\ba} \otimes \text{id}_\ba)\circ(\text{id}_\ba \otimes \delta_{\ba,\overline\ba})$, so the two Hom-spaces are isomorphic, in particular $\text{Hom}(\ba\otimes\ba \rightarrow\bf 1)$ is one-dimensional.
This implies that the constructed morphism is either symmetric or antisymmetric:
\begin{equation}B^{\ba,\ba}\circ\beta_{\ba,\ba} = \pm B^{\ba,\ba}\,.\end{equation}
We hence have the following tripartite classification of simple objects:
\begin{enumerate}
	\item If $\ba$ and $\overline\ba$ are not isomorphic, that is, if $\Hom(\ba\rightarrow\overline\ba)$ is trivial, then we say that $\ba$ is \emph{complex}.
	\item If $\ba$ and $\overline\ba$ are isomorphic and $B^{\ba,\ba}$ is symmetric then we say that $\ba$ is \emph{real}.
	\item If $\ba$ and $\overline\ba$ are isomorphic and $B^{\ba,\ba}$ is pseudoreal, then we say that $\ba$ is \emph{pseudoreal}.
\end{enumerate}
 
The categories $\Reptilde\,U(n)$, $\Reptilde\,O(n)$ and $\Reptilde\,Sp(n)$ are special in the theory of symmetric tensor categories:
\begin{thm} {\rm (Deligne \cite{Deligne1982,Deligne1990,Del02})} Let $\fcy C$ be a symmetric tensor category and $\ba\in\fcy C$ any object with $\text{dim}(\ba) = n$:
 	\begin{enumerate}
 		\item There is a unique (up to isomorphism) functor $F:\Reptilde\,U(n)\rightarrow\fcy C$ with $F({[+]}) = \ba.$
 		\item If there is a symmetric isomorphism $\ba\rightarrow\overline\ba$, then there is a unique (up to isomorphism) functor $F:\Reptilde\,O(n)\rightarrow\fcy C$ with $F([1]) = \ba.$
 		\item If there is a skew-symmetric isomorphism $\ba\rightarrow\overline\ba$, then there is a unique (up to isomorphism) functor $F:\Reptilde\,Sp(n)\rightarrow\fcy C$ with $F([1]) = \ba.$
 	\end{enumerate}
 	If $n\not\in\mathbb Z$ these functors are all faithful.\textsuperscript{\ref{note:faithful}}
 \end{thm}
Essentially, this says that out of all the tensor categories with an object $\ba$ of dimension $\text{dim}(\ba) = n \not\in\mathbb Z$, the category $\Reptilde\,U(n)$ is the most symmetric, in the sense that it contains the minimal space of morphisms in $\Hom(\ba^{\otimes k}\rightarrow\ba^{\otimes k}).$ In a similar vein, $\Reptilde\,O(n)$ is the most symmetric tensor category with a real representation $\text{dim}(\ba) = n$ and $\Reptilde\,Sp(n)$ is the most symmetric tensor category with a pseudoreal representation.

The category $\Reptilde\,S_n$ also satisfies a universality property; it is the universal category with a commutative Frobenius algebra of dimension $n$ \cite{Deligne1990}. To explain the precise meaning of this statement is beyond the scope of the appendix, so instead we shall simply note that Frobenius algebras arise naturally as the algebraic structure underlying 2d TQFTs (see \cite{kock_2003} for a pedagogic introduction to this topic). While traditionally defined over vector spaces, they can be generalized to any category.\footnote{See footnote \ref{foot:int} for further discussion.} The universality property of $\Reptilde\,S_n$ states that for any category $\fcy C$, functors from $\Reptilde\,S_n\rightarrow\fcy C$ are in one-to-one correspondence with commutative Frobenius algebras in $\fcy C$.

\subsection{Deligne's theorem on classification of tensor categories}
\label{sec:Deligne}

We have seen that the representations of any group $G$ are naturally described by a symmetric tensor category $\Rep\,G$. By taking this abstract approach to representation theory, we then exhibited several families of categories which smoothly interpolate between the representation categories for various groups such as $O(N)\,,$ $U(N)$ and $S_N$ (section \ref{sec:other}). By inspection, for non-integer $n$ these Deligne categories contain objects of non-integer and negative dimensions, and so cannot be equivalent to $\Rep\,G$ for some group $G$.
 
One wonder if there exist some tensor categorical symmetries which have all objects of integer and positive dimensions, and yet which are inequivalent to any group symmetry. Deligne's theorem on tensor categories says that this is impossible. This theorem describes the conditions under which a symmetric tensor category is equivalent to $\Rep\,G$. More generally, it classifies all symmetric tensor categories satisfying a certain finiteness condition, but for positive categories the statement is particularly simple:

 \begin{thm} {\rm (Deligne \cite{Deligne1990,Del02})} Any positive symmetric tensor category $\mathcal C$ is equivalent to $\Rep\,G$ for some group $G$. If $\mathcal C$ is a fusion category\footnote{A \emph{fusion category} is a tensor category with a finite number of simple objects.}, then $\fcy C$ is equivalent to $\Rep\,G$ for some unique\footnote{There exist \emph{isocategorical groups} $G_1$ and $G_2$ for which $\Rep\,G_1$ and $\Rep\,G_2$ are equivalent as fusion categories, but which have different braidings and so are distinct as symmetric fusion categories \cite{DAVYDOV2001273,2002RvMaP14733I,2000math......7196E}} finite group $G$. 
 \end{thm}
 Because any category $\Rep\,G$ is automatically positive, we see that positivity is both necessary and sufficient for a symmetric tensor category to be of the form $\Rep\,G$.
 
 We will now describe Deligne's theorem more generally. Given an $\ba$ in a tensor category $\fcy C$, we define $\text{length}(\ba)$ to be the number of simple objects appearing in $\ba$. We say that $\fcy C$ has \emph{subexponential growth} if for every $\ba\in\fcy C$ there exists a $C_\ba$ such that for all $k>0$, $\text{length}(\ba^{\otimes k})\leq (C_\ba)^k$. 
 
 All positive tensor categories satisfy the subexponential growth condition. Indeed, as discussed in section \ref{sec:PosCat}, in any positive tensor category the dimensions of all objects must be integers and so:
 \begin{equation}
 \text{length}(\ba^{\otimes k}) \leq \text{dim}(\ba)^k\,.
 \end{equation}
 On the other hand, Deligne categories such as $\Reptilde\,O(n)$ do not satisfy the subexponential growth condition. For non-integer $n$ the number of linearly independent morphisms ${{\bf n}^{\otimes k}\rightarrow {\bf n}^{\otimes k}}$ grows factorially in $k$, and this in turn implies that the number of simple objects in ${\bf n}^{\otimes k}$ also grows factorially.

 Let us begin with fusion categories. These always have subexponential growth, as proved in section 9.9 of \cite{etingof2016tensor}.  Given a finite group $G$, let $z\in G$ be in the centre of $G$ and satisfy $z^2 = 1$. On any irreducible representation of $G$, $z$ acts either as the identity or minus the identity. For any simple object $\ba\in\Rep\,G$ we define
 \begin{equation}
 \text{sgn}(\ba) = \begin{cases}
 +1 & \text{ if } z \text{ acts as the identity on } \ba \\
 -1 & \text{ if } z \text{ acts as the minus the identity on } \ba \\
 \end{cases}\end{equation}
 We then define $\Rep(G,z)$ to be the category $\Rep\,G$ but with the modified braiding:
 \begin{equation}\beta_{\mathbf{a},\mathbf{b}}({\bf a}\otimes {\bf b}) = (-1)^{\text{sgn}(\ba)\text{sgn}({\bf b})} {\bf b}\otimes {\bf a}.
 \end{equation}
 We can think of these as representations where if $\text{sgn}({\bf a}) = -1$ we use Grassmann rather than real variables. Indeed the pair $(G,z)$ can be thought of as a finite analogue of a supergroup \cite{2017JMP....58d1704B}. We then have the following result:
 \begin{thm}{\rm (Deligne \cite{Deligne1990,Del02}) }Any symmetric fusion category $\mathcal C$ is equivalent to $\Rep(G,z)$ for a unique finite group $G$ and $z\in G$ as above.
 \end{thm}
 
 Deligne's theorem for arbitrary subexponential symmetric tensor categories is more technical to state. Broadly speaking it shows that such categories are equivalent to representations of supergroups.

 \subsection{Unitarity}
Usually in quantum field theory we have some notion of unitarity. In order to define this for theories with a more general categorical symmetry we need to introduce some new structure.
 
 As in section \ref{sec:unitarity}, \emph{conjugation} $*$ is a anti-linear braided monoidal functor $*:\fcy C\rightarrow\fcy C$ which is involutive, $** = \text{id}_{\fcy C}$,
 and for which $\ba^*$ is dual to $\ba$.
 
 Because $\ba^*$ is dual to $\ba$, there are cap and cocap maps $\delta^{\ba,\ba^*}$ and $\delta_{\ba^*,\ba}$, and also $\delta^{\ba^*,\ba}$ and $\delta_{\ba,\ba^*}$, satisfying the zig-zag relations. As in the previous subsections, we can always choose these maps to satisfy
 \begin{equation}
 \delta_{\ba^*,\ba} = \delta_{\ba,\ba^*}\circ\beta_{\ba^*,\ba}\,,\quad \delta^{\ba,\ba^*} = \beta_{\ba^*,\ba}\circ\delta^{\ba^*,\ba}.
 \end{equation}
 
 \begin{prop}\label{pr:realCap}
 	For any simple object $\ba$, there are cap and cocap maps $\delta^{\ba,\ba^*}$ and $\delta_{\ba^*,\ba}$ such that
 	\begin{equation}
 	\label{eq:needRealCap}
 	(\delta^{\ba,\ba^*})^* = \delta^{\ba^*,\ba}\,,\quad (\delta_{\ba^*,\ba})^* = \delta_{\ba,\ba^*}\,.
 	\end{equation}
 \end{prop}
 \begin{proof}
 	For simple object $\ba$, $\Hom(\ba\otimes\ba^*\rightarrow{\bf 1})$ and $\Hom(\ba^*\otimes\ba\rightarrow{\bf 1})$ are one-dimensional and so
 	\begin{equation}\label{eq:capStar}
 	(\delta^{\ba,\ba^*})^*= \alpha \delta^{\ba^*,\ba}
 	\end{equation} 
 	for some number $\alpha$. We thus find that
 	\begin{equation}\label{eq:capStar2}
 	(\delta^{\ba,\ba^*})^{**} = \alpha(\delta^{\ba,\ba^*}\circ\beta_{\ba,\ba^*})^* = |\alpha|^2 \delta^{\ba,\ba^*}\circ\beta_{\ba,\ba^*}\circ\beta_{\ba^*,\ba} = |\alpha|^2\delta^{\ba,\ba^*}\,.
 	\end{equation} 
 	Because $*$ is involutive,  $|\alpha| = 1$. 
 	Repeating the same logic for $\delta_{\ba,\ba^*}$ implies that
 	\begin{equation}
 	(\delta_{\ba,\ba^*})^* = \beta \delta_{\ba^*,\ba},\qquad |\beta|=1\,.
 	\end{equation}
 	Applying $*$ to both sides of the zig-zag relations:
 	\begin{equation}
 	\left[(\text{id}_{\ba}\otimes\delta^{\ba^*,\ba})\otimes(\delta_{\ba,\ba^*} \otimes\text{id}_{\ba})\right]^* = (\text{id}_{\ba})^*
 	\end{equation}
 	we find that $\alpha \beta=1$. We now define the new cap and cocap morphisms $\hat\delta^{\ba,\ba^*} = \alpha^{1/2} \delta^{\ba,\ba^*}$ and $\hat\delta_{\ba^*,\ba} = \alpha^{-1/2}\delta_{\ba^*,\ba}$, which still satisfy the zig-zag relations and also satisfy \reef{eq:needRealCap}.
 \end{proof}

 From now on we shall only work with cap and cocap maps supplied by proposition \ref{pr:realCap}.
  As a consequence, we see that all dimensions in a category with conjugation must be real:
 \begin{equation}
 \text{dim}(\ba)^* = (\delta^{\ba,\ba^*}\circ\delta_{\ba,\ba^*})^* = \delta^{\ba,\ba^*}\circ\beta_{\ba^*,\ba}\circ\beta_{\ba,\ba^*}\circ\delta_{\ba,\ba^*} = \text{dim}(\ba).
 \end{equation}

 Given a complex conjugation $*$ on a symmetric tensor category $\fcy C$, we can define a notion of Hermitian adjoint. Given any morphism $f:\ba\rightarrow{\bf b}$, we can define the adjoint $f^\dagger:{\bf b}\rightarrow\ba$ to be the morphism:
 \begin{equation}\label{eq:adjDef}
 f^\dagger = (\text{id}_{\ba}\otimes\delta^{ {\bf b}^*,{\bf b}})\circ(\text{id}_{\ba}\otimes f^*\otimes\text{id}_{\bf b})\circ(\delta_{\ba,\ba^*}\otimes\text{id}_{\bf b}) = \includegraphics[trim=0 4em 0 0, scale=0.35]{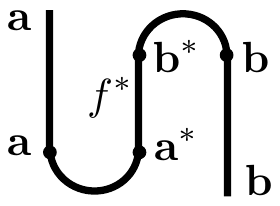}.
 \end{equation}\\[-0.5em]
 We use the adjoint to define a sesquilinear form on $\Hom(\ba\rightarrow{\bf b})$: given any two morphisms $f,g :\ba\rightarrow{\bf b}$ we define 
 \begin{equation}
 \langle g^\dagger,f\rangle = \text{tr}_{\ba}(g^\dagger\circ f) = \includegraphics[trim=0 4em 0 0, scale=0.35]{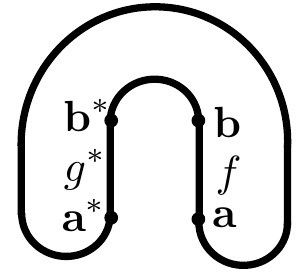}\,. 
 \end{equation}\\[-0.5em]

 \begin{prop} The sesquilinear product $\langle\cdot,\cdot\rangle$ defined above is conjugate symmetric and non-degenerate. If $\langle\cdot,\cdot\rangle$ is positive-definite on all $\Hom$-sets then $\fcy C$ is positive.
 \end{prop}
 \begin{proof}
 	To prove conjugate symmetry we compute
 	\begin{equation}
 	\langle g^\dagger,f\rangle^* = \text{tr}_{\ba^*}(({g^\dagger})^*\circ f^*) = \includegraphics[trim=0 4em 0 0, scale=0.35]{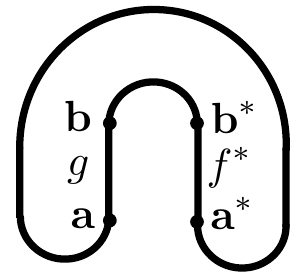}= \text{tr}_{\ba}(f^\dagger\circ g) = \langle f^\dagger,g\rangle.
 	\end{equation}\\[-0.5em]
 	Non-degeneracy follows from proposition \ref{pr:ndeg} and the fact that $\dagger$ is idempotent.
 	
 	Finally, to prove that $\langle\cdot,\cdot\rangle$ can only be positive-definite if $\fcy C$ is positive, we compute
 	\begin{equation}
 	\langle \text{id}_{\ba}^\dagger,\text{id}_\ba\rangle = \text{tr}_{\ba}(\text{id}_{\ba}) = \text{dim}(\ba),
 	\end{equation}
 	and so positive-definiteness hence requires $\text{dim}(\ba) > 0$ for any simple object $\ba$.
 \end{proof}

 We have thus seen that conjugation defines a Hermitian form on the Hom-sets of $\fcy C$, but that this conjugation can only be positive-definite for positive categories. If the Hermitian form $\<\cdot,\cdot\>$ is positive-definite on every Hom-set, $\fcy C$ is called a \emph{unitary category} (see e.g. \cite{2018arXiv180800323P}). Because they are not positive, Deligne categories cannot be unitary for non-integer $n$. 
 
 In section \ref{sec:unitarity} we proved that CFTs with a non-integer $n$ Deligne category symmetry are non-unitary. Note that this does not follow automatically from the non-unitarity of these categories, but requires a separate argument.

\subsubsection{Another note on definitions}
Our definition of conjugation is a little different from that given in the literature, for instance in \cite{SELINGER2007139}. A $\dagger$-category is a monoidal category with an involutive functor $\dagger:\fcy C^{\text{op}}\rightarrow\fcy C$, which acts as the identity on objects. In other words, for every morphism $f:\ba\rightarrow{\bf b}$ there is an adjoint morphism $f^\dagger:{\bf b}\rightarrow\ba$ such that
\begin{equation}
\text{id}_\ba^\dagger = \text{id}_\ba\,,\qquad (g\circ f)^\dagger = f^\dagger\circ g^\dagger\,,\qquad f^{\dagger\dagger} = f\,.
\end{equation}
If $\fcy C$ is a braided tensor category, we say that $\dagger$ is \emph{antilinear} and \emph{braided} if
\begin{equation}
(\lambda f+\mu g)^\dagger = \lambda^* f^\dagger + \mu^* g^\dagger\,,\qquad \beta_{\ba,{\bf b}}^\dagger = \beta_{\ba,\bf b}\,.
\end{equation}
As we saw in \eqref{eq:adjDef}, our notion of conjugation in a symmetric tensor category can be used to define such a map $\dagger$, and it is straightforward to check that this $\dagger$ is \emph{antilinear} and \emph{braided}. Conversely, given an antilinear, braided, and pivotal $\dagger$ we can define for any morphism $f$ a conjugate morphism $f^*:a^*\rightarrow b^*$ by
\begin{equation}
f^* = (\delta^{\ba^*,\ba}\otimes\text{id}_{{\bf b}^*})\circ(\text{id}_{\ba^*}\otimes f^\dagger\otimes\text{id}_{{\bf b}^*})\circ(\text{id}_{\ba^*}\otimes\delta_{{\bf b},{\bf b}^*})\,.
\end{equation}
Hence our notion of a symmetric tensor category with conjugation is the same as that of a symmetric tensor category with a antilinear, braided $\dagger$-functor.

\section{Continuous categories}
\label{sec:CONTCATS}

In this section we shall develop some aspects of the theory of continuous categories. To the best of our knowledge this is not a concept that has explicitly been studied in the literature, although it appears to be related to the notion of a fundamental group of a tensor category as described in \cite{Deligne1990,ETINGOF2014}, with our maximal adjoint being the Lie algebra of that group.\footnote{We would like to thank Pavel Etingof for this remark.}

Let us recall from section \ref{sec:current} that an \emph{adjoint} $(\mathfrak g,\tau)$ in a symmetric tensor category $\fcy C$ is an object $\mathfrak g$ and a family of morphisms $\tau_\ba:\mathfrak g\otimes\ba\rightarrow\ba$ satisfying \eqref{eq:ctcon1}, \eqref{eq:ctcon2}, \eqref{eq:ctcon3}, and \eqref{eq:adjCon}, which we repeat for convenience:
\begin{enumerate}
\item \emph{(Naturalness)} For every morphism $f:\ba\rightarrow{\bf b}$, the following diagram commutes:
\begin{equation}\label{eq:nattau}
\begin{tikzcd}
\mathfrak g\otimes\ba \arrow[r,"\text{id}_{\mathfrak g}\otimes f"] \arrow[d,"\tau_\ba"]
& \mathfrak g\otimes{\bf b}  \arrow[d, "\tau_{\bf b}"] \\
\ba \arrow[r, "f"]
& {\bf b}
\end{tikzcd}\end{equation}
\item For all objects $\ba,{\bf b}\in\fcy C$ the following identities are satisfied:
\begin{align}
\label{eq:tensCond}&\tau_{\ba\otimes\bf b} = \tau_{\ba}\otimes\text{id}_{\bf b} + (\text{id}_\ba\otimes\tau_{\bf b})\circ(\beta_{\mathfrak g,\ba}\otimes\text{id}_{\bf b})\,,\\
\label{eq:symCond}&\tau_{\mathfrak g}\otimes\beta_{\mathfrak g,\mathfrak g} = -\tau_{\mathfrak g}\,.
\end{align}
\item \emph{(Non-degeneracy)}
$\text{ If a morphism }f:\mathfrak g\rightarrow\mathfrak g \text{ satisfies } \tau_{\ba}\circ(f\otimes\text{id}_{\ba}) = 0 \text{ for every }\ba\in\fcy C\,,\text{ then } f = 0\,.$
\end{enumerate}
If a category has an adjoint object $(\mathfrak g,\tau)$ we shall say that the category is \emph{continuous}, otherwise we say that it is \emph{discrete}.

Our definition of an adjoint does not make explicit use of the additive structure on $\fcy C$, and indeed we could generalize the notion of an adjoint to any symmetric monoidal category. Semisimplicity will however prove very constraining:
\begin{prop}\label{pr:sumTau}
Given an adjoint $(\mathfrak g,\tau)$ and any two objects $\ba_1,\ba_2\in\fcy C$,
\begin{equation}
\tau_{\ba_1\oplus\ba_2} = \iota_1\circ\tau_{\ba_1}\circ({\rm id}_{\mathfrak g}\otimes\pi_1) + \iota_2\circ\tau_{\ba_2}\circ({\rm id}_{\mathfrak g}\otimes\pi_2)\,.
\end{equation}
\end{prop}
\begin{proof} 
Let $\pi_i$ and $\iota_i$ be the projection and embedding morphisms from $\ba_i$ into $\ba_1\oplus\ba_2$. Applying naturality to $\iota_i$, we find that
\begin{equation}
\tau_{\ba_1\oplus\ba_2}\circ(\text{id}_{\mathfrak g}\otimes\iota_i) = \iota_i\circ\tau_{\ba_i}\implies \tau_{\ba_1\oplus\ba_2}\circ\left(\text{id}_{\mathfrak g}\otimes(\iota_i\circ\pi_i)\right) = \iota_i\circ\tau_{\ba_i}\circ(\text{id}_{\mathfrak g}\otimes\pi_{a_i})\,.
\end{equation}
We thus find that
\begin{equation}
\iota_1\circ\tau_{\ba_1}\circ({\rm id}_{\mathfrak g}\otimes\pi_1) + \iota_2\circ\tau_{\ba_2}\circ({\rm id}_{\mathfrak g}\otimes\pi_2) = 
\tau_{\ba_1\oplus\ba_2}\circ\left(\text{id}_{\mathfrak g}\otimes(\iota_1\circ\pi_1+\iota_2\circ\pi_2)\right) = \tau_{\ba_1\oplus\ba_2}\,,
\end{equation}
which is what we set out to prove.
\end{proof}

\begin{prop}\label{pr:simpAdj}
	Any adjoint $(\mathfrak g,\tau)$ can be decomposed as $\mathfrak g \approx \ba_1\oplus\dots\oplus\ba_n$, 
	\begin{equation}
	\tau_{\bf b} = \sum_{i = 1}^n (\sigma_i)_{\bf b}\circ(\pi_i\otimes{\rm id}_{\bf b})\,,
	\end{equation}
	where each $\ba_i$ is a simple object and each $(\ba_i,\sigma_i)$ is an adjoint for all $i = 1,\dots,n$.
\end{prop}
\begin{proof} 
Using semisimplicity we can always decompose $\mathfrak g \approx \ba_1\oplus\dots\oplus\ba_n$, with project and embedding morphisms $\pi_i$ and $\iota_i$. For each $i$ we can then define
\begin{equation}
(\sigma_i)_{\bf b} \equiv (\tau_{\bf b})\circ(\iota_i\otimes\rm{id}_{\bf b})\,,\quad \text{so that}\quad \tau_{\bf b} = \sum_{i = 1}^n (\sigma_i)_{\bf b}\circ(\pi_i\otimes{\rm id}_{\bf b})\,.
\end{equation}
 It is straightforward to check that $(\ba_i,\sigma_i)$ satisfies naturality, \eqref{eq:tensCond} and \eqref{eq:symCond}. To prove non-degeneracy, we note that if $f:\ba_i\rightarrow\ba_i$ satisfies
\begin{equation}
(\sigma_i)_{\bf b}\circ(f\otimes\text{id}_\ba) = 0
\end{equation}
for all ${\bf b}\in\fcy C$, then 
\begin{equation}
\tau_{\bf b}\circ\left((\iota_i\circ f\circ\pi_i)\otimes\text{id}_\ba\right) = 0\,.
\end{equation}
The non-degeneracy of $\tau$ then implies that 
\begin{equation}
\iota_i\circ f\circ\pi_i = 0 \implies f = \pi_i\circ0\circ\iota_i = 0\,,
\end{equation}
and so $(\ba_i,\sigma_i)$ is an adjoint.
\end{proof}
Taken together, these two propositions allows us to restrict our attention to simple objects. To specify an adjoint $(\mathfrak g,\tau)$ we need only specify $\tau_\ba$ for each simple object $\ba$, and to find all possible adjoints in a category we can restrict our search to simple adjoints.

We can introduce a partial ordering on the collection of adjoints in a category $\fcy C$, writing $(\mathfrak g,\tau) \succeq (\mathfrak h,\sigma)$ if there exists a morphism $f:\mathfrak h\rightarrow\mathfrak g$ such that
\begin{equation}
\sigma_\ba = \tau_{\ba}\circ(f\otimes\text{id}_\ba) \text{ for all }\ba\in\fcy C\,.
\end{equation}
We shall say that two adjoints are isomorphic, $(\mathfrak g,\tau)\approx(\mathfrak h,\sigma)$, if both $(\mathfrak g,\tau)\succeq(\mathfrak h,\sigma)$ and $(\mathfrak h,\sigma)\succeq(\mathfrak g,\tau)$.

\begin{prop}\label{pr:isoadj}
Given isomorphic adjoints $(\mathfrak g,\tau)\approx(\mathfrak h,\sigma)$, there exists a unique isomorphism $f:\mathfrak g\rightarrow\mathfrak h$ such that
\begin{equation}\tau_\ba = \sigma_\ba\circ(f\otimes\mathrm{id}_\ba)\,,\qquad \sigma_\ba = \tau_\ba\circ(f^{-1}\otimes\mathrm{id}_\ba)\,.\end{equation}
\end{prop}
\begin{proof} Let $f:\mathfrak g\rightarrow\mathfrak h$ and $g:\mathfrak h\rightarrow\mathfrak g$ be morphisms satisfying
\begin{equation}\label{eq:intcond}
\tau_\ba = \sigma_\ba\circ(f\otimes\mathrm{id}_\ba)\,,\qquad \sigma_\ba = \tau_\ba\circ(g\otimes\mathrm{id}_\ba)\,,
\end{equation}
for all $\ba\in\fcy C$. We can then compute
\begin{equation}
\tau_\ba\circ\left((\text{id}_{\mathfrak g}-g\circ f)\otimes\text{id}_\ba\right) = \tau_\ba - \sigma_\ba\circ\left(f\otimes\text{id}_\ba\right) = 0\,.
\end{equation}
and so the non-degeneracy condition on $(\mathfrak g,\tau)$ implies that $g\circ f = \text{id}_{\mathfrak g}$. An identical argument implies that $f\circ g =\text{id}_{\mathfrak h}$ and thus that $g = f^{-1}$. As inverses in a category are unique, it follows that $f$ and $g$ are the unique morphisms satisfying \eqref{eq:intcond}.
\end{proof}

Let us now consider the task of finding all adjoints, up to isomorphism, in a category which is \emph{finitely generated}. Recall that, as defined in section \ref{sec:CFT}, an object $\bf g$ generates a category $\fcy C$ if any simple object in $\fcy C$ appears in ${\bf g}^{\otimes k}$ for sufficiently large $k$. If such an object exists we say that $\fcy C$ is finitely generated. The Deligne categories we have considered in this paper are all finitely generated, and so are the representation categories of finite group and of semisimple Lie algebras, so this condition is not very restrictive. 

\begin{prop}\label{pr:adjinggb}
If a continuous category $\fcy C$ is generated by an object $\bf g$ and $(\mathfrak h,\tau)$ is an adjoint in $\fcy C$, then the morphism
\begin{equation} 
\eta \equiv (\tau_{\bf g}\otimes\rm{id}_{\overline{\bf g}})\circ(\rm{id}_{\mathfrak h}\otimes\delta_{{\bf g},\overline{\bf g}})
\end{equation}
embeds $\mathfrak h$ into ${\bf g}\otimes\overline{\bf g}$.
\end{prop}
\begin{proof}
We will prove this by contradiction, assuming that $\eta$ is not an embedding morphism, i.e.~has a non-trivial kernel. In this case there must exist a simple object $\mathfrak j$ in $\mathfrak h$ (with $\pi_{\mathfrak j}$ and $\iota_{\mathfrak j}$ projection and embedding morphisms) such that
\begin{equation}
\eta\circ\iota_{\mathfrak j} = 0 \implies\tau_{\bf g}\circ(\iota_{\mathfrak j}\otimes\text{id}_{\bf g}) = 0\,.
\end{equation}
We can then use \eqref{eq:tensCond} to show that for any $k$,
\begin{equation}\label{eq:vanfork}
\tau_{{\bf g}^{\otimes k}}\circ(\iota_{\mathfrak j}\otimes\text{id}_{{\bf g}^{\otimes k}}) = 0\,.
\end{equation}

Now consider any other simple object $\ba\in\fcy C$. Because $\bf g$ generates $\fcy C$ there exists embedding and projection operators $\iota_\ba:\ba\rightarrow{\bf g}^{\otimes k}$ and $\pi_{\ba}:{\bf g}^{\otimes k}\rightarrow\ba$ for some sufficiently large value of $k$. Applying \eqref{eq:nattau} to $\pi_\ba$ and then \eqref{eq:vanfork} we find that
\begin{equation}
\tau_\ba\circ(\text{id}_{\mathfrak h}\otimes\pi_\ba)\circ(\iota_{\mathfrak j}\otimes\text{id}_{{\bf g}^{\otimes k}}) = \pi_\ba\circ\tau_{{\bf g}^{\otimes k}}\circ(\iota_{\mathfrak j}\otimes\text{id}_{{\bf g}^{\otimes k}})\circ(\iota_{\mathfrak j}\otimes\text{id}_{{\bf g}^{\otimes k}}) = 0\,.
\end{equation}
If we now contract the l.h.s sides of this equation with $\pi_{\mathfrak j}\otimes\iota_\ba$ we find that
\begin{equation}
\tau_\ba\circ(\text{id}_{\mathfrak h}\otimes\pi_\ba)\circ(\iota_{\mathfrak j}\otimes\text{id}_{{\bf g}^{\otimes k}})\circ(\pi_{\mathfrak j}\otimes\iota_\ba) = \tau_\ba \circ (\iota_{\mathfrak j}\circ\pi_{\mathfrak j})\otimes(\pi_\ba\circ\iota_\ba)) = \tau_\ba\circ((\iota_{\mathfrak j}\circ\pi_{\mathfrak j})\otimes\text{id}_\ba)\,,
\end{equation}
and so conclude that for every $\ba\in\fcy C$,
\begin{equation}
\tau_\ba\circ((\iota_{\mathfrak j}\circ\pi_{\mathfrak j})\otimes \text{id}_\ba) = 0\,.
\end{equation}
Because $\iota_{\mathfrak j}\circ\pi_{\mathfrak j}$ is a non-zero morphism from $\mathfrak h\rightarrow\mathfrak h$, we find that non-degeneracy condition for $(\mathfrak h,\tau)$ has been violated, and so we have a contradiction.
\end{proof}

By this proposition, to find all adjoints we may simply go through the list of all simple objects in ${\bf g}\otimes\overline{\bf g}$. As a simple application, $\Reptilde\,O(n)$ is generated by the object $\bf n$. Since ${{\bf n}^{\otimes 2}\approx {\bf S}\oplus{\bf A}\oplus{\bf 1}}$, to find all of the adjoints in $\Reptilde\,O(n)$ we can simply check these three possibilities, and find that ${\bf A}$ is the only adjoint (for ${\bf S}$ and ${\bf 1}$ the category does not contain morphisms $\tau_{\ba}$ with the needed properties).

In a symmetric tensor category $\fcy C$ we will define the maximal adjoint object $(\mathfrak m,T)$, if it exists, to be an adjoint such that for any other adjoint $(\mathfrak g,\tau)$, $(\mathfrak m,T)\succeq(\mathfrak g,\tau)$. Due to proposition \ref{pr:isoadj}, if a maximal adjoint exists, it is unique up to unique isomorphism. 

\begin{prop}
In a finitely generated tensor category $\fcy C$ there exists a unique maximal adjoint.
\end{prop}
\begin{proof}
We will begin with the following easy observation. Given any two adjoints $(\mathfrak h,\tau_1)$ and $(\mathfrak h,\tau_2)$ where $\mathfrak h$ is simple, then any linear combination $(\mathfrak h, \lambda_1\tau_1+\lambda_2\tau_2)$ is also an adjoint so long as $\lambda_1\tau_1+\lambda_2\tau_2\neq 0$.

Let $\bf g$ generate $\fcy C$ and let $\ba_i$ be a simple object in ${\bf g}\otimes\overline{\bf g}\approx\ba_1\oplus\dots\oplus\ba_n$. We can then define $U_i$ to be the set of morphisms $f:\ba_i\rightarrow{\bf g}\otimes\overline{\bf g}$ such that 
\begin{equation}
f = (\tau_{\bf g}\otimes{\rm id}_{\overline{\bf g}})\circ({\rm id}_{\ba_i}\otimes\delta_{{\bf g},\overline{\bf g}})
\end{equation}
for some adjoint $(\ba_i,\tau)$. Note that if such an adjoint exists it is unique; for if there were two such adjoints $(\ba_i,\tau_1)$ and $(\ba_i,\tau_2)$ then the $(\ba_i,\tau_1-\tau_2)$ would be an adjoint but would violate proposition \ref{pr:adjinggb}.

Since any linear combinations of adjoints is an adjoint, the set $U_i$ is a vector space. Choose a basis $f_1,\dots,f_{k_i}$ of $U_i$, corresponding to adjoints $(\ba_i,\tau_1),\dots,(\ba_i,\tau_{k_i})$. We can then define an adjoint $(\ba_i^{\oplus k_i},\sigma_i)$ with
\begin{equation}
(\sigma_i)_{\bf b} = \sum_{j = 1}^{k_i} \iota_j\circ(\tau_j)_{\bf b}\circ(\pi_j\otimes\text{id}_{\bf b})\,,
\end{equation}
where $\iota_i$ and $\pi_i$ are the embedding and projection morphisms into $\ba_i^{\oplus k_i}$. By construction $(\ba_i^{\oplus k_i},\sigma_i)\succeq (\ba_i,\tau)$ for any adjoint $(\ba_i,\tau)$.

Repeating this construction for each $\ba_i\in{\bf g}\otimes\overline{\bf g}$, we can then define $\mathfrak m = \bigoplus_{i = 1}^n\ba_i^{\oplus k_i}$ and
\begin{equation}
T_{\bf b} = \sum_{i = 1}^n \iota_{\ba_i}\circ(\sigma_i)_{\bf b}\circ(\pi_{\ba_i}\otimes\text{id}_{\bf b})\,.
\end{equation}
It is straightforward to verify that $(\mathfrak m,T)$ is an adjoint, and, using propositions \ref{pr:adjinggb} and \ref{pr:simpAdj}, that for any adjoint $(\mathfrak h,\tau)$ in $\fcy C$, $(\mathfrak m,T)\succeq(\mathfrak h,\tau)$.
\end{proof}


\small

\bibliography{ONbiblio}
\bibliographystyle{utphys}

\end{document}